\let\oldvec\vec
\documentclass[runningheads,envcountsect,envcountsame,final]{llncs}
\let\vec\oldvec 

\usepackage[T1]{fontenc}
\usepackage{graphicx}
\usepackage{tikz}
\usetikzlibrary{calc}
\usetikzlibrary{positioning}
\usetikzlibrary{arrows,automata}
\usetikzlibrary{patterns,patterns.meta}
\usepackage{xspace}
\usepackage{xcolor}
\usepackage{amsmath}
\usepackage{thmtools}
\usepackage{amssymb}
\usepackage{mathrsfs}
\usepackage{ifdraft}
\usepackage{wrapfig}
\usepackage{algpseudocode}
\usepackage{subcaption}
\usepackage{verbatim}
\usepackage{booktabs}
\usepackage{multirow}
\usepackage{paralist}
\usepackage{cite}
\usepackage{lipsum}
\usepackage{amsbsy}
\usepackage{etoolbox}
\usepackage{style/orcidlink}
\usepackage[symbol]{footmisc}

\RequirePackage{centernot}
\RequirePackage[Symbol]{upgreek}

\RequirePackage{etoolbox}
\RequirePackage{mathtools}

\let\originalleft\left
\let\originalright\right
\renewcommand{\left}{\mathopen{}\mathclose\bgroup\originalleft}
\renewcommand{\right}{\aftergroup\egroup\originalright}

\newcommand{\pvec}[1]{\vec{#1}\mkern2mu\vphantom{#1}}

\newcommand{\Procs}{\ensuremath{\mathcal{P}}}
\definecolor{roleColor}{rgb}{0.1, 0.3, 0.1}\newcommand{\roleCol}[1]{{\color{roleColor}#1}}\newcommand{\roleFmt}[1]{\boldsymbol{\roleCol{\mathtt{#1}}}}\newcommand{\roleP}[1][]{\ifempty{#1}{{\color{roleColor}\roleFmt{p}}}{{\color{roleColor}\roleFmt{p}_{#1}}}}
\newcommand{\roleVar}[1][]{\ifempty{#1}{{\color{roleColor}\roleFmt{x}}}{{\color{roleColor}\roleFmt{x}_{#1}}}}

\newcommand{\procA}{{\color{roleColor}\roleFmt{p}}}
\newcommand{\procB}{{\color{roleColor}\roleFmt{q}}}
\newcommand{\procC}{{\color{roleColor}\roleFmt{r}}}
\newcommand{\procD}{{\color{roleColor}\roleFmt{s}}}

\newcommand{\run}{\rho}

\newcommand{\Alphabet}{\Sigma}

\newcommand{\val}{\ensuremath{m}}
\newcommand{\MsgVals}{\ensuremath{\mathcal{V}}}

\newcommand{\msgO}{\textcolor{orange}{o}}
\newcommand{\msgB}{\textcolor{blue}{b}}
\newcommand{\msgM}{\textcolor{magenta}{m}}

\newcommand{\CSM}[1]{\ensuremath{\{\!\!\{#1_\procA\}\!\!\}_{\procA \in \Procs}}}
\newcommand{\CSMl}[1]{\ensuremath{\{\!\!\{{#1}\}\!\!\}_{\procA \in \Procs}}}

\newcommand{\emptystring}{\varepsilon}

\newcommand{\set}[1]{\{#1\}}
\newcommand{\lang}{\mathcal{L}}
\newcommand{\interswaplang}{\mathcal{C}^{\interswap}}
\newcommand{\SyncToAsync}{\ensuremath{\operatorname{\texttt{\upshape{split}}}}}
\newcommand{\channels}{\ensuremath{\mathsf{Chan}}}
\newcommand{\channel}[2]{\ensuremath{#1,#2}}

\newcommand{\trace}{\ensuremath{\operatorname{\texttt{\upshape{trace}}}}}

\newcommand{\GG}{\mathbf{G}}
\newcommand{\getMu}{\mathit{get\mu}}
\newcommand{\getMuG}{\getMu_\GG}

\newcommand{\semglobal}{\ensuremath{\mathsf{GAut}}}
\newcommand{\semglobalsync}{\ensuremath{\mathsf{GAut}}}
\newcommand{\projerasuresymb}{\downarrow}
\newcommand{\projerasure}[2]{\ensuremath{\semglobal(#1)\negmedspace\!\projerasuresymb_{#2}}}
\newcommand{\projerasuretrans}{\ensuremath{\delta_\projerasuresymb}}
\newcommand{\AlphSync}{\ensuremath{Σ_{\mathit{sync}}}}
\newcommand{\AlphAsync}{\ensuremath{Σ_{\mathit{async}}}}

\newcommand{\restrict}[2]{{#1}|_{#2}}
\newcommand{\lfp}{\mathrm{lfp}}
\newcommand{\fin}{\operatorname{fin}}

\newcommand{\infi}{\operatorname{inf}}

\newcommand{\interswap}{\ensuremath{\sim}}

\def \ifempty#1{\def\temp{#1} \ifx\temp\empty }

\newcommand{\snd}[3]{#1\triangleright#2!#3}
\newcommand{\rcv}[3]{#2\triangleleft#1?#3}
\newcommand{\ssnd}[2]{#1!#2}
\newcommand{\srcv}[2]{#1?#2}
\newcommand{\msgFromTo}[3]{#1\!\to\!#2\!:\!#3}

\newcommand{\pref}{\operatorname{pref}}
\newcommand{\preforder}{\ensuremath{\leq}}

\newcommand{\channelcompliant}{channel-compliant\xspace}

\newcommand{\terminated}{terminated\xspace}

\newcommand{\gtcomplete}[1]{$#1$-complete\xspace}

\newcommand{\exampleend}{\hfill$\blacktriangleleft$}
\newcommand{\proofend}{\hfill$\qed$}
\newcommand{\myparagraph}[1]{\smallskip\noindent\textbf{#1}}

\newcommand{\proj}{{\ensuremath{\upharpoonright}}}

\newcommand{\wproj}{{\ensuremath{\Downarrow}}}

\newcommand{\merge}{\sqcap}
\def\mmerge{\mathrel{\ThisStyle{\stretchrel*{\ooalign{\raise0.2\LMex\hbox{$\SavedStyle\sqcap$}\cr \raise-0.2\LMex\hbox{$\SavedStyle\sqcap$}}}{\sqcap}}}}
\def\mmmerge{\mathrel{\ThisStyle{\stretchrel*{\ooalign{\raise0.6\LMex\hbox{$\SavedStyle\sqcap$}\cr \raise0.2\LMex\hbox{$\SavedStyle\sqcap$}\cr \raise-0.2\LMex\hbox{$\SavedStyle\sqcap$}}}{\sqcap}}}}

\newcommand{\blockedset}{\ensuremath{\mathcal{B}}}
\newcommand{\semavail}{\ensuremath{M}}

\newcommand{\semavaildef}[3]{\ensuremath{\semavail^{#1, #2}_{(#3\ldots)}}}

\newcommand{\union}{\cup}
\newcommand{\inters}{\cap}
\newcommand{\Union}{\bigcup}
\newcommand{\Inters}{\bigcap}
\newcommand{\dunion}{\uplus}

\DeclarePairedDelimiter\card{\lvert}{\rvert}

\providecommand{\implies}{\Rightarrow}

\providecommand{\coloneq}{\mathrel{\mathop:}=} \providecommand{\Coloneqq}{\mathrel{\mathop{::}}=} \newcommand{\is}{\coloneq}

 \def\st{\colon}

\def\grammOr{\hspace{3pt}\mid\hspace{3pt}}
\def\grammIs{\Coloneqq}

\begingroup
\catcode`\|=\active \gdef\@grammar@bar{\catcode`\|=\active \def|{\grammOr}}
\endgroup

\newcommand{\gramm}[1]{\begingroup
  \def\is{\grammIs}\@grammar@bar #1\endgroup }

\newenvironment{grammar}{\begin{equation*}\def\is{& \grammIs }\@grammar@bar \aligned }
{\endaligned \end{equation*}\aftergroup\ignorespaces }

\newcommand{\yessymbol}{\textcolor{teal}{$\pmb{\checkmark}$}}
\newcommand{\nosymbol}{\textcolor{red}{$\times$}}

\newcommand{\hole}{\hbox{-}}

\newcommand{\subsetproj}[2]{\ensuremath{\mathscr{P}(#1,#2)}}
\newcommand{\subsetcons}[2]{\ensuremath{\mathscr{C}(#1,#2)}}

\newcommand{\transAnnoFunc}{\ensuremath{\operatorname{tr-orig}}}
\newcommand{\transAnnoFunDest}{\ensuremath{\operatorname{tr-dest}}}

\newcommand{\globcomplocal}[3]{\ensuremath{\operatorname{R}^#1_{#2}(#3)}}
 \usepackage{newunicodechar}
\newunicodechar{∃}{\ensuremath{\exists}}
\newunicodechar{∀}{\ensuremath{\forall}}
\newunicodechar{θ}{\ensuremath{\theta}}
\newunicodechar{τ}{\ensuremath{\tau}}
\newunicodechar{φ}{\ensuremath{\varphi}}
\newunicodechar{ξ}{\ensuremath{\xi}}
\newunicodechar{ζ}{\ensuremath{\zeta}}
\newunicodechar{ψ}{\ensuremath{\psi}}
\newunicodechar{π}{\ensuremath{\pi}}
\newunicodechar{α}{\ensuremath{\alpha}}
\newunicodechar{β}{\ensuremath{\beta}}
\newunicodechar{γ}{\ensuremath{\gamma}}
\newunicodechar{δ}{\ensuremath{\delta}}
\newunicodechar{ε}{\ensuremath{\varepsilon}}
\newunicodechar{κ}{\ensuremath{\kappa}}
\newunicodechar{λ}{\ensuremath{\lambda}}
\newunicodechar{μ}{\ensuremath{\mu}}
\newunicodechar{ρ}{\ensuremath{\rho}}
\newunicodechar{σ}{\ensuremath{\sigma}}
\newunicodechar{ω}{\ensuremath{\omega}}
\newunicodechar{Γ}{\ensuremath{\Gamma}}
\newunicodechar{Φ}{\ensuremath{\Phi}}
\newunicodechar{Δ}{\ensuremath{\Delta}}
\newunicodechar{Σ}{\ensuremath{\Sigma}}
\newunicodechar{Π}{\ensuremath{\Pi}}
\newunicodechar{∑}{\ensuremath{\Sigma}}
\newunicodechar{∏}{\ensuremath{\Pi}}
\newunicodechar{Θ}{\ensuremath{\Theta}}
\newunicodechar{Ω}{\ensuremath{\Omega}}
\newunicodechar{⇒}{\ensuremath{\Rightarrow}}
\newunicodechar{⇐}{\ensuremath{\Leftarrow}}
\newunicodechar{⇔}{\ensuremath{\Leftrightarrow}}
\newunicodechar{→}{\ensuremath{\rightarrow}}
\newunicodechar{←}{\ensuremath{\leftarrow}}
\newunicodechar{↔}{\ensuremath{\leftrightarrow}}
\newunicodechar{¬}{\ensuremath{\neg}}
\newunicodechar{∧}{\ensuremath{\land}}
\newunicodechar{∨}{\ensuremath{\lor}}
\newunicodechar{≠}{\ensuremath{\neq}}
\newunicodechar{≡}{\ensuremath{\equiv}}
\newunicodechar{∼}{\ensuremath{\sim}}
\newunicodechar{≈}{\ensuremath{\approx}}
\newunicodechar{≥}{\ensuremath{\geq}}
\newunicodechar{≤}{\ensuremath{\leq}}
\newunicodechar{≫}{\ensuremath{\gg}}
\newunicodechar{≪}{\ensuremath{\ll}}
\newunicodechar{∅}{\ensuremath{\emptyset}}
\newunicodechar{⊆}{\ensuremath{\subseteq}}
\newunicodechar{⊂}{\ensuremath{\subset}}
\newunicodechar{∩}{\ensuremath{\cap}}
\newunicodechar{⋂}{\ensuremath{\cap}}
\newunicodechar{∪}{\ensuremath{\cup}}
\newunicodechar{⋃}{\ensuremath{\cup}}
\newunicodechar{⊎}{\ensuremath{\uplus}}
\newunicodechar{∈}{\ensuremath{\in}}
\newunicodechar{∉}{\ensuremath{\not\in}}
\newunicodechar{⊤}{\ensuremath{\top}}
\newunicodechar{⊥}{\ensuremath{\bot}}
\newunicodechar{₀}{\ensuremath{_0}}
\newunicodechar{₁}{\ensuremath{_1}}
\newunicodechar{₂}{\ensuremath{_2}}
\newunicodechar{₃}{\ensuremath{_3}}
\newunicodechar{₄}{\ensuremath{_4}}
\newunicodechar{₅}{\ensuremath{_5}}
\newunicodechar{₆}{\ensuremath{_6}}
\newunicodechar{₇}{\ensuremath{_7}}
\newunicodechar{₈}{\ensuremath{_8}}
\newunicodechar{₉}{\ensuremath{_9}}
\newunicodechar{⁰}{\ensuremath{^0}}
\newunicodechar{¹}{\ensuremath{^1}}
\newunicodechar{²}{\ensuremath{^2}}
\newunicodechar{³}{\ensuremath{^3}}
\newunicodechar{⁴}{\ensuremath{^4}}
\newunicodechar{⁵}{\ensuremath{^5}}
\newunicodechar{⁶}{\ensuremath{^6}}
\newunicodechar{⁷}{\ensuremath{^7}}
\newunicodechar{⁸}{\ensuremath{^8}}
\newunicodechar{⁹}{\ensuremath{^9}}
\newunicodechar{𝔹}{\ensuremath{\mathbb{B}}}
\newunicodechar{ℝ}{\ensuremath{\mathbb{R}}}
\newunicodechar{ℕ}{\ensuremath{\mathbb{N}}}
\newunicodechar{ℂ}{\ensuremath{\mathbb{C}}}
\newunicodechar{ℚ}{\ensuremath{\mathbb{Q}}}
\newunicodechar{𝕋}{\ensuremath{\mathbb{T}}}
\newunicodechar{𝕏}{\ensuremath{\mathbb{X}}}
\newunicodechar{ℤ}{\ensuremath{\mathbb{Z}}}
\newunicodechar{✓}{\ding{51}}
\newunicodechar{✗}{\ding{55}}
\newunicodechar{◊}{\ensuremath{\lozenge}}
\newunicodechar{□}{\ensuremath{\square}}
\newunicodechar{𝓐}{\ensuremath{\mathcal{A}}}
\newunicodechar{𝓑}{\ensuremath{\mathcal{B}}}
\newunicodechar{𝓒}{\ensuremath{\mathcal{C}}}
\newunicodechar{𝓓}{\ensuremath{\mathcal{D}}}
\newunicodechar{𝓔}{\ensuremath{\mathcal{E}}}
\newunicodechar{𝓕}{\ensuremath{\mathcal{F}}}
\newunicodechar{𝓖}{\ensuremath{\mathcal{G}}}
\newunicodechar{𝓗}{\ensuremath{\mathcal{H}}}
\newunicodechar{𝓘}{\ensuremath{\mathcal{I}}}
\newunicodechar{𝓙}{\ensuremath{\mathcal{J}}}
\newunicodechar{𝓚}{\ensuremath{\mathcal{K}}}
\newunicodechar{𝓛}{\ensuremath{\mathcal{L}}}
\newunicodechar{𝓜}{\ensuremath{\mathcal{M}}}
\newunicodechar{𝓝}{\ensuremath{\mathcal{N}}}
\newunicodechar{𝓞}{\ensuremath{\mathcal{O}}}
\newunicodechar{𝓟}{\ensuremath{\mathcal{P}}}
\newunicodechar{𝓠}{\ensuremath{\mathcal{Q}}}
\newunicodechar{𝓡}{\ensuremath{\mathcal{R}}}
\newunicodechar{𝓢}{\ensuremath{\mathcal{S}}}
\newunicodechar{𝓣}{\ensuremath{\mathcal{T}}}
\newunicodechar{𝓤}{\ensuremath{\mathcal{U}}}
\newunicodechar{𝓥}{\ensuremath{\mathcal{V}}}
\newunicodechar{𝓦}{\ensuremath{\mathcal{W}}}
\newunicodechar{𝓧}{\ensuremath{\mathcal{X}}}
\newunicodechar{𝓨}{\ensuremath{\mathcal{Y}}}
\newunicodechar{𝓩}{\ensuremath{\mathcal{Z}}}
\newunicodechar{…}{\ensuremath{\ldots}}
\newunicodechar{∗}{\ensuremath{\ast}}
\newunicodechar{⊢}{\ensuremath{\vdash}}
\newunicodechar{⊧}{\ensuremath{\models}}
\newunicodechar{′}{\ensuremath{'}}
\newunicodechar{″}{\ensuremath{''}}
\newunicodechar{‴}{\ensuremath{'''}}
\newunicodechar{∥}{\ensuremath{\|}}
\newunicodechar{⊕}{\ensuremath{\oplus}}
\newunicodechar{⁺}{\ensuremath{^+}}
\newunicodechar{⊇}{\ensuremath{\supseteq}}
\newunicodechar{∘}{\ensuremath{\circ}}
\newunicodechar{∙}{\ensuremath{\cdot}}
\newunicodechar{⋅}{\ensuremath{\cdot}}
\newunicodechar{≈}{\ensuremath{\approx}}
\newunicodechar{×}{\ensuremath{\times}}
\newunicodechar{∞}{\ensuremath{\infty}}
\newunicodechar{⊑}{\ensuremath{\sqsubseteq}}
 
\newcounter{symbolfootnote}
\newcommand*{\symbolfootnote}[2]{\setcounter{symbolfootnote}{\value{footnote}}\renewcommand*{\thefootnote}{\fnsymbol{footnote}}\footnote[#1]{#2}\setcounter{footnote}{\value{symbolfootnote}}\renewcommand*{\thefootnote}{\arabic{footnote}}}

\usepackage[disable]{todonotes}
\usepackage[firstpage,stamp=false]{draftwatermark}

\usepackage{hyperref}
\usepackage[capitalise]{cleveref}
\usepackage[shortlabels]{enumitem}

\crefformat{section}{\S#2#1#3} \crefmultiformat{section}{\S#2#1#3}{ and~\S#2#1#3}{, \S#2#1#3}{, and~\S#2#1#3} \crefrangeformat{section}{\S#3#1#4 to~\S#5#2#6}

\newtoggle{techrep}
\toggletrue{techrep}
\iftoggle{techrep}{
    \newcommand{\appendixRef}[1]{\cref{#1}}
}{
    \newcommand{\appendixRef}[1]{the extended version~\cite{li2023complete}}
}

\makeatletter
\newcommand{\printfnsymbol}[1]{\textsuperscript{\@fnsymbol{#1}}}
\makeatother

\begin{document}
\title{Complete Multiparty Session Type Projection with Automata}
\titlerunning{Complete Multiparty Session Type Projection with Automata}

\author{
Elaine Li\protect\symbolfootnote{1}{equal contribution}\inst{1}\orcidlink{0000-0003-0173-4498} \and
Felix Stutz\printfnsymbol{1}\protect\symbolfootnote{2}{corresponding author}\inst{2}\orcidlink{0000-0003-3638-4096} \and
Thomas Wies\inst{1}\orcidlink{0000-0003-4051-5968} \and
Damien Zufferey\inst{3}\orcidlink{0000-0002-3197-8736}
}
\institute{New York University \email{efl9013@nyu.edu, wies@cs.nyu.edu} \and
Max Planck Institute for Software Systems \email{fstutz@mpi-sws.org} \and
SonarSource \email{damien.zufferey@sonarsource.com}}

\authorrunning{E.\ Li, F.\ Stutz, T.\ Wies, D.\ Zufferey}

\maketitle              \begin{abstract}

Multiparty session types (MSTs) are a type-based approach to verifying communication protocols. 
Central to MSTs is a \emph{projection operator}: a partial function that maps protocols represented as global types to correct-by-construction implementations for each participant, represented as a communicating state machine.
Existing \mbox{projection} operators are syntactic in nature, and trade efficiency for completeness.
We present the first projection operator that is sound, complete, and efficient.
Our projection separates synthesis from checking implementability. 
For synthesis, we use a simple automata-theoretic construction; for checking implementability, we present succinct conditions that summarize insights into the property of implementability. 
We use these conditions to show that MST implementability is in PSPACE.
This improves upon a previous decision procedure that is in EXPSPACE and applies to a smaller class of MSTs. We demonstrate the effectiveness of our approach using a prototype implementation, which handles global types not supported by previous work without sacrificing performance. 

\end{abstract}
\section{Introduction}
\label{sec:intro}

Communication protocols are key components in many safety and operation critical systems, making them prime targets for formal verification. 
Unfortunately, most verification problems for such protocols (e.g. deadlock freedom) are undecidable~\cite{DBLP:journals/jacm/BrandZ83}.
To make verification computationally tractable, several restrictions have been proposed~\cite{DBLP:journals/iandc/CeceF05,DBLP:conf/concur/BolligFS20,DBLP:conf/cav/AbdullaBJ98,DBLP:conf/lics/AbdullaAA16,DBLP:journals/acta/PengP92,DBLP:conf/tacas/TorreMP08}. 
In particular, multiparty session types (MSTs)~\cite{DBLP:conf/popl/HondaYC08} have garnered a lot of attention in recent years (see, e.g., the survey by Ancona et al.~\cite{DBLP:journals/ftpl/AnconaBB0CDGGGH16}).
In the MST setting, a protocol is specified as a global type, which describes the desired interactions of all roles involved in the protocol. 
Local implementations describe behaviors for each individual role.
The implementability problem for a global type asks whether there exists a collection of local implementations whose composite behavior when viewed as a communicating state machine (CSM) matches that of the global type and is deadlock-free.
The synthesis problem is to compute such an implementation from an implementable global type.

MST-based approaches typically solve synthesis and implementability simultaneously via an efficient syntactic \emph{projection operator} \cite{DBLP:conf/popl/HondaYC08,DBLP:conf/sfm/CoppoDPY15,DBLP:journals/jlp/ToninhoY17,DBLP:conf/ecoop/ScalasDHY17}.
Abstractly, a
projection operator is a partial map from global types to collections of implementations. 
A projection operator $\texttt{proj}$ is sound when every global type $\GG$ in its domain is implemented by $\texttt{proj}(\GG)$, and complete when every implementable global type is in its domain.
Existing practical projection operators for MSTs are all incomplete (or unsound). 
Recently, the implementability problem was shown to be decidable for a class of MSTs via a reduction to safe realizability of globally cooperative high-level message sequence charts (HMSCs)~\cite{ecoop-draft}. In principle, this result yields a complete and sound projection operator for the considered class. However, this operator would not be practical. In particular, the proposed implementability check is in \mbox{EXPSPACE}.

\myparagraph{Contributions.}
In this paper, we present the first practical sound and complete projection operator for general MSTs.
The synthesis problem for implementable global types is conceptually easy~\cite{ecoop-draft}
-- the challenge lies in determining whether a global type \emph{is} implementable.
We thus separate synthesis from checking implementability. 
We first use a standard automata-theoretic construction to obtain a candidate implementation for a potentially non-implementable global type. 
However, unlike~\cite{ecoop-draft}, we then verify the correctness of this implementation directly using efficiently checkable conditions derived from the global type. 
When a global type is not implementable, our constructive completeness proof provides a counterexample trace.

The resulting projection operator yields a PSPACE decision procedure for implementa\-bility.\footnote{A previous version of this paper claimed that the implementability problem for MSTs is PSPACE-hard, but there was an error in the proof. This version has since been edited to remove the claim. Instead, we show that there exist global types for which the construction of an implementation requires exponential time.}
This result both generalizes and tightens the decidability and complexity results obtained in~\cite{ecoop-draft}.

We evaluate a prototype of our projection algorithm on benchmarks taken from the literature.
Our prototype benefits from both the efficiency of existing lightweight but incomplete syntactic projection operators\cite{DBLP:conf/popl/HondaYC08,DBLP:conf/sfm/CoppoDPY15,DBLP:journals/jlp/ToninhoY17,DBLP:conf/ecoop/ScalasDHY17}, and the generality of heavyweight automata-based model checking techniques~\cite{DBLP:journals/pacmpl/ScalasY19,DBLP:conf/cav/LangeY19}: it handles protocols rejected by previous practical approaches while preserving the efficiency that makes MST-based techniques so attractive.

 \section{Motivation and Overview}
\label{sec:motivation}

\begin{figure}[t]
\begin{subfigure}[b]{0.28\textwidth}
\centering
\includegraphics[width=0.95\textwidth]{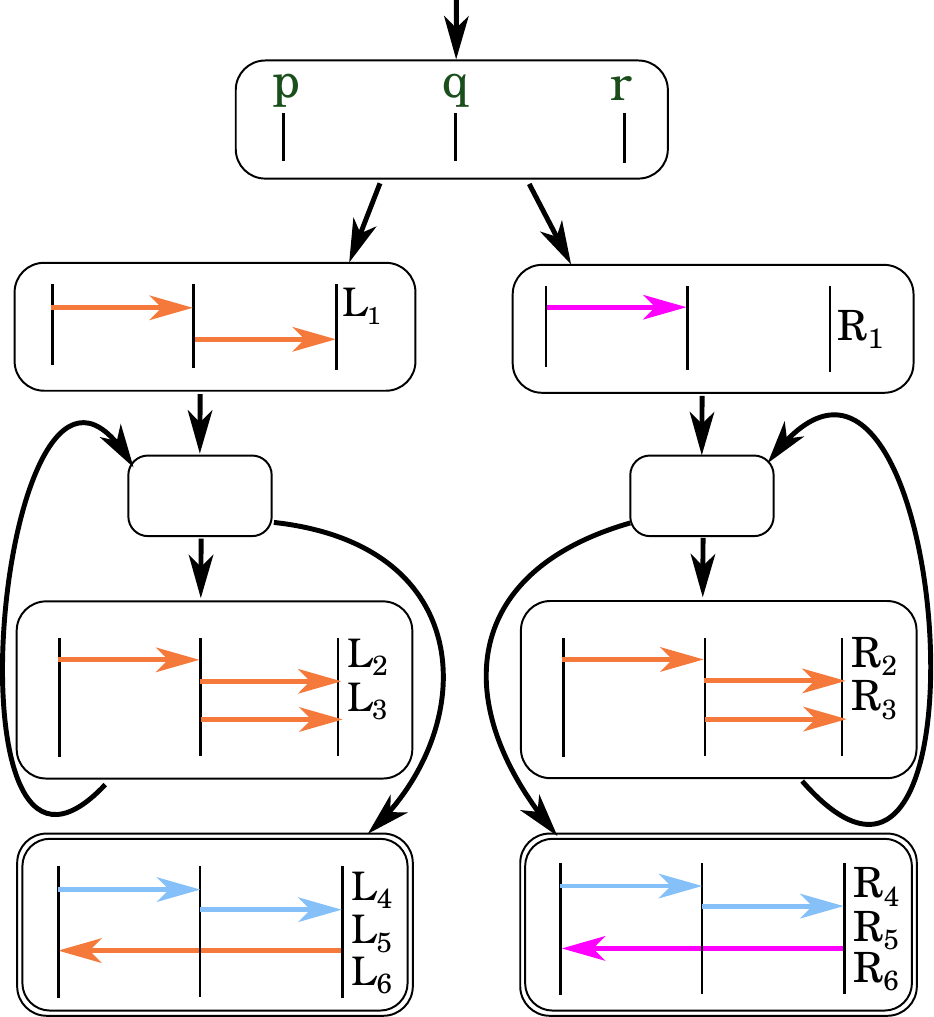}
\caption{Odd-even protocol}
    \label{fig:hmsc-odd-even}
\end{subfigure}
\label{fig:intro}
\begin{subfigure}[b]{0.23\textwidth}
\centering
\resizebox{1.10\textwidth}{!}{
    \begin{tikzpicture}[node distance=0.8cm and 1.2cm]
    \node[draw, circle, minimum size=0.3cm, initial above, initial text = ] (q0) {};

    \node[draw, circle, minimum size=0.3cm, below = of q0, xshift=-6mm] (q1l) {};
    \node[draw, circle, minimum size=0.3cm, below = of q1l] (q2l) {};
    \node[draw, circle, minimum size=0.3cm, accepting, below = of q2l] (q3l) {};

    \node[draw, circle, minimum size=0.3cm, below = of q0, xshift=6mm] (q1r) {};
    \node[draw, circle, minimum size=0.3cm, below = of q1r] (q2r) {};
    \node[draw, circle, minimum size=0.3cm, accepting, below = of q2r] (q3r) {};
\path[->] (q0) edge node[left] {$\ssnd{\procB}{\textcolor{orange}{o}}$} (q1l);
    \path[->] (q1l) edge [loop left] node[left] {$\ssnd{\procB}{\textcolor{orange}{o}}$} (q1l);
    \path[->] (q1l) edge node[left] {$\ssnd{\procB}{\textcolor{blue}{b}}$} (q2l);
    \path[->] (q2l) edge node[left] {$\srcv{\procC}{\textcolor{orange}{o}}$} (q3l);

    \path[->] (q0) edge node[right] {$\ssnd{\procB}{\textcolor{magenta}{m}}$} (q1r);
    \path[->] (q1r) edge [loop right] node[right] {$\ssnd{\procB}{\textcolor{orange}{o}}$} (q1r);
    \path[->] (q1r) edge node[right] {$\ssnd{\procB}{\textcolor{blue}{b}}$} (q2r);
    \path[->] (q2r) edge node[right] {$\srcv{\procC}{\textcolor{magenta}{m}}$} (q3r);

\end{tikzpicture}
 }
\caption{Local impl. \\ for role $\procA$ \hspace{1cm}}
\label{fig:hmsc-odd-even-p}
\end{subfigure}
\begin{subfigure}[b]{0.19\textwidth}
\centering
\resizebox{1.08\textwidth}{!}{
    \begin{tikzpicture}[node distance=0.8cm and 1.2cm]
    \node[draw, circle, minimum size=0.3cm, initial above, initial text = ] (q0) {};

    \node[draw, circle, minimum size=0.3cm, left = of q0, yshift=-5mm] (q1l) {};
    \node[draw, circle, minimum size=0.3cm, right = of q1l, yshift=-5mm] (q2) {};
    \node[draw, circle, minimum size=0.3cm, below left = of q2] (q3) {};
    \node[draw, circle, minimum size=0.3cm, below = of q3] (q4) {};
    \node[draw, circle, minimum size=0.3cm, below = of q2] (q5) {};
    \node[draw, circle, minimum size=0.3cm, accepting, below = of q5] (q6) {};
\path[->] (q0) edge[sloped] node[above] {$\srcv{\procA}{\textcolor{orange}{o}}$} (q1l);
    \path[->] (q1l) edge[sloped] node[above,xshift=1.5mm] {$\ssnd{\procC}{\textcolor{orange}{o}}$} (q2);
    \path[->] (q0) edge node[right] {$\srcv{\procA}{\textcolor{magenta}{m}}$} (q2);
    \path[->] (q2) edge[sloped] node[above] {$\srcv{\procA}{\textcolor{orange}{o}}$} (q3);
    \path[->] (q3) edge node[left] {$\ssnd{\procC}{\textcolor{orange}{o}}$} (q4);
    \path[->] (q4) edge[sloped] node[below] {$\ssnd{\procC}{\textcolor{orange}{o}}$} (q2);
    \path[->] (q2) edge node[right] {$\srcv{\procA}{\textcolor{blue}{b}}$} (q5);
    \path[->] (q5) edge node[right] {$\ssnd{\procC}{\textcolor{blue}{b}}$} (q6);

\end{tikzpicture}
 }
\caption{Local impl. \\ for role $\procB$ \hspace{1cm}}
\label{fig:hmsc-odd-even-q}
\end{subfigure}
\begin{subfigure}[b]{0.26\textwidth}
\centering
\resizebox{1.10\textwidth}{!}{
    \begin{tikzpicture}[node distance=0.8cm and 1.2cm]
    \node[draw, rectangle, rounded corners, initial above, initial text = ] (q0) {$\set{\textnormal{L}_1, \textnormal{R}_1, \textnormal{R}_2, \textnormal{R}_4}$};
    \node[draw, rectangle, rounded corners, below left=0.8cm and -0.8cm of q0] (q1) {$\set{\textnormal{L}_2, \textnormal{L}_4, \textnormal{R}_3}$};
    \node[draw, rectangle, rounded corners, below right=0.8cm and -0.8cm of q0] (q2) {$\set{\textnormal{L}_3, \textnormal{R}_2, \textnormal{R}_4}$};
    \node[draw, rectangle, rounded corners, below= of q1] (q3) {$\set{\textnormal{L}_5}$};
    \node[draw, rectangle, rounded corners, accepting, below = of q3] (q4) {$\set{\textnormal{L}_6}$};
    \node[draw, rectangle, rounded corners, below= of q2] (q5) {$\set{\textnormal{R}_5}$};
    \node[draw, rectangle, rounded corners, accepting, below = of q5] (q6) {$\set{\textnormal{R}_6}$};
\path[->] (q0) edge node[left] {$\srcv{\procB}{\textcolor{orange}{o}}$} (q1);
    \path[->] (q1) edge [bend left=8] node[above] {$\srcv{\procB}{\textcolor{orange}{o}}$} (q2);
    \path[->] (q2) edge [bend left=8] node[below] {$\srcv{\procB}{\textcolor{orange}{o}}$} (q1);

    \path[->] (q1) edge node[right] {$\srcv{\procB}{\textcolor{blue}{b}}$} (q3);
    \path[->] (q3) edge node[right] {$\ssnd{\procA}{\textcolor{orange}{o}}$} (q4);

    \path[->] (q0) edge[bend left,out=65,in=90,looseness=1.3] node[sloped,above] {$\srcv{\procB}{\textcolor{blue}{b}}$} (q5);
    \path[->] (q2) edge node[right] {$\srcv{\procB}{\textcolor{blue}{b}}$} (q5);
    \path[->] (q5) edge node[right] {$\ssnd{\procA}{\textcolor{magenta}{m}}$} (q6);

\end{tikzpicture}
 }
\caption{Local impl. \\ for role $\procC$ \hspace{1cm}}
\label{fig:hmsc-odd-even-r}
\end{subfigure}

\caption{Odd-even: An implementable but not (yet) projectable protocol and its local implementations}
\label{fig:odd-even}
\end{figure}
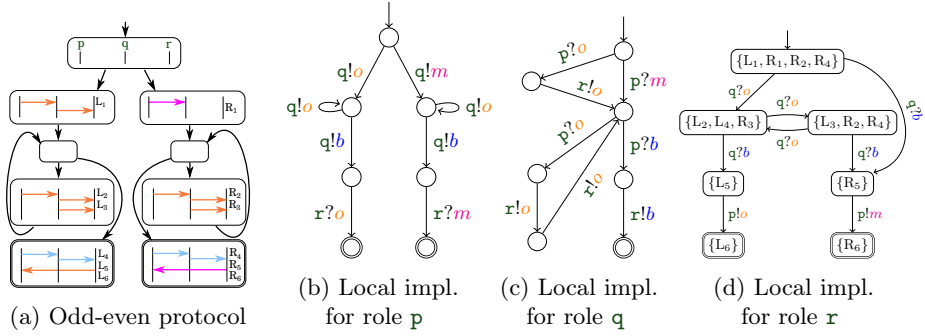

\myparagraph{Incompleteness of existing projection operators.}
A key limitation of existing projection operators is that the implementation for each role is obtained via a linear traversal of the global type, and thus shares its structure. 
The following example, which is not projectable by any existing approach, demonstrates how enforcing structural similarity can lead to incompleteness.
\begin{example}[Odd-even]
\label{ex:odd-even}
Consider the following global type $\GG_{oe}$:
{\scriptsize\[
    + \;
    \begin{cases}
    \msgFromTo{\procA}{\procB}{\msgO}. \,
    \msgFromTo{\procB}{\procC}{\msgO}. \,
    \mu t_1. \,
    (
    \msgFromTo{\procA}{\procB}{\msgO}. \,
    \msgFromTo{\procB}{\procC}{\msgO}. \,
    \msgFromTo{\procB}{\procC}{\msgO}. \,
    t_1
    \; + \;
    \msgFromTo{\procA}{\procB}{\msgB}. \,
    \msgFromTo{\procB}{\procC}{\msgB}. \,
    \msgFromTo{\procC}{\procA}{\msgO}. \,
    0
    )
    \\
    \msgFromTo{\procA}{\procB}{\msgM}. \,
    \mu t_2. \,
    (
    \msgFromTo{\procA}{\procB}{\msgO}. \,
    \msgFromTo{\procB}{\procC}{\msgO}. \,
    \msgFromTo{\procB}{\procC}{\msgO}. \,
    t_2
    \; + \;
    \msgFromTo{\procA}{\procB}{\msgB}. \,
    \msgFromTo{\procB}{\procC}{\msgB}. \,
    \msgFromTo{\procC}{\procA}{\msgM}. \,
    0
    )
    \end{cases}
\]}
A term $\msgFromTo{\procA}{\procB}{\val}$ specifies the exchange of message $\val$ between sender $\procA$ and receiver $\procB$. 
The term represents two local events observed separately due to asynchrony: a send event $\snd{\procA}{\procB}{\val}$ observed by role $\procA$, and a receive event $\rcv{\procA}{\procB}{\val}$ observed by role $\procB$. 
The $+$ operator denotes choice, $\mu t.\, G$ denotes recursion, and $0$ denotes protocol termination.

\cref{fig:hmsc-odd-even} visualizes $\GG_{oe}$ as an HMSC.
The left and right sub-protocols respectively correspond to the top and bottom branches of the protocol.
Role $\procA$ chooses a branch by sending either $\msgO$ or $\msgM$ to~$\procB$. On the left, $\procB$ echoes this message to~$\procC$. Both branches continue in the same way: $\procA$ sends an arbitrary number of $\msgO$ messages to $\procB$, each of which is forwarded twice from $\procB$ to $\procC$. Role $\procA$ signals the end of the loop by sending $\msgB$ to~$\procB$, which $\procB$ forwards to $\procC$. Finally, depending on the branch, $\procC$ must send $\msgO$ or $\msgM$ to~$\procA$.

\cref{fig:hmsc-odd-even-p,fig:hmsc-odd-even-q} depict the structural similarity between the global type $\GG_{oe}$ and the implementations for $\procA$ and $\procB$.  
For the ``choicemaker'' role $\procA$, the reason is evident. 
Role $\procB$'s implementation collapses the continuations of both branches in the protocol into a single sub-component. 
For $\procC$ (\cref{fig:hmsc-odd-even-r}), the situation is more complicated.
Role $\procC$ does not decide on or learn directly which branch is taken, but can deduce it from the parity of the number of $\msgO$ messages received from $\procB$: odd means left and even means right.
The resulting local implementation features transitions going back and forth between the two branches that do not exist in the global type.
Syntactic projection operators fail to create such transitions.
\exampleend
\end{example}
One response to the brittleness of existing projection operators has been to give up on global type specifications altogether and instead revert to model checking user-provided implementations~\cite{DBLP:journals/pacmpl/ScalasY19,DBLP:conf/cav/LangeY19}.
We posit that what needs rethinking is not the concept of global types, but rather how projections are computed and how implementability is checked. 

\myparagraph{Our automata-theoretic approach.}
The synthesis step in our projection operator uses textbook automata-theoretic constructions. 
From a given global type, we derive a finite state machine, and use it to define a homomorphism automaton for each role.
We then determinize
this homomorphism automaton via subset construction to obtain a local candidate implementation for each role.
If the global type is implementable, this construction always yields an implementation.
The implementations shown in \cref{fig:hmsc-odd-even-p,fig:hmsc-odd-even-q,fig:hmsc-odd-even-r} are the result of applying this construction to $\GG_{oe}$ from \cref{ex:odd-even}.
Notice that the state labels in \cref{fig:hmsc-odd-even-r} correspond to sets of labels in the global protocol.

Unfortunately, not all global types are implementable.

\begin{figure}[t]
\begin{subfigure}[b]{0.24\textwidth}
    \centering
    \includegraphics[width=0.97\textwidth]{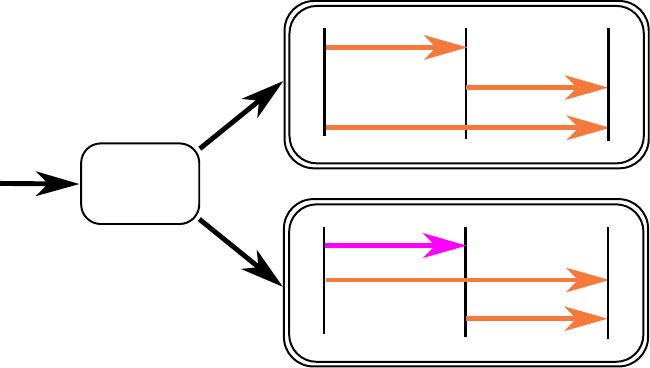}
    \caption{$\GG_r$}
    \label{fig:rcv-cond-negative}
\end{subfigure}
\begin{subfigure}[b]{0.24\textwidth}
    \centering
    \includegraphics[width=0.97\textwidth]{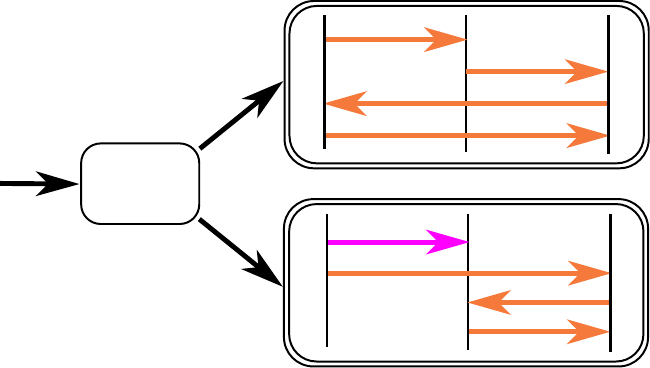}
    \caption{$\GG_r'$}
    \label{fig:rcv-cond-positive}
\end{subfigure}
\begin{subfigure}[b]{0.24\textwidth}
    \centering
    \includegraphics[width=0.97\textwidth]{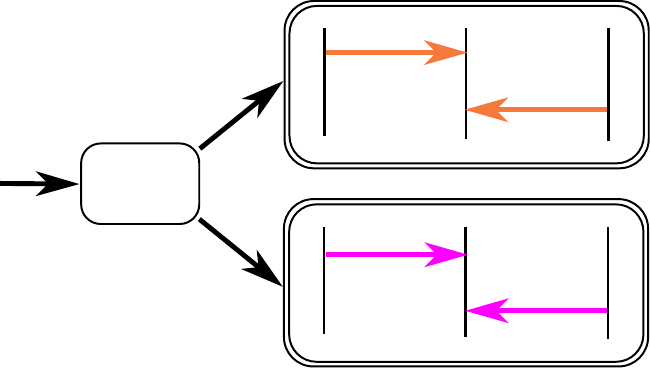}
    \caption{$\GG_s$}
    \label{fig:snd-cond-negative}
\end{subfigure}
\begin{subfigure}[b]{0.24\textwidth}
    \centering
    \includegraphics[width=0.97\textwidth]{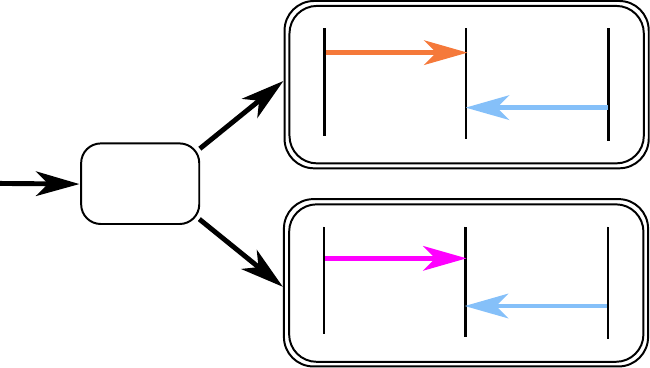}
    \caption{$\GG_s'$}
    \label{fig:snd-cond-positive}
\end{subfigure}
\caption{High-level message sequence charts for the global types of \cref{ex:implementability}.}
\label{fig:implementability}
\end{figure}

\begin{example}\label{ex:implementability}
 Consider the following four global types also depicted in \cref{fig:implementability}:
    { \scriptsize
      \begin{align*}
		\GG_r = {} & + \;
		\begin{cases}
			\msgFromTo{\procA}{\procB}{\msgO}. \,
			\msgFromTo{\procB}{\procC}{\msgO}. \,
			\msgFromTo{\procA}{\procC}{\msgO}. \,
			0
			\\
			\msgFromTo{\procA}{\procB}{\msgM}. \,
			\msgFromTo{\procA}{\procC}{\msgO}. \,
			\msgFromTo{\procB}{\procC}{\msgO}. \,
			0
		\end{cases}
		&
		\qquad
\GG_s = {} & + \;
		\begin{cases}
                  \msgFromTo{\procA}{\procB}{\msgO}. \,
                  \msgFromTo{\procC}{\procB}{\msgO}. \,
                  0
                  \\
                  \msgFromTo{\procA}{\procB}{\msgM}. \,
                  \msgFromTo{\procC}{\procB}{\msgM}. \,
                  0
                \end{cases}\\
                \GG_r' = {} & + \;
	        \begin{cases}
                  \msgFromTo{\procA}{\procB}{\msgO}. \,
                  \msgFromTo{\procB}{\procC}{\msgO}. \,
                  \msgFromTo{\procC}{\procA}{\msgO}. \,
                  \msgFromTo{\procA}{\procC}{\msgO}. \,
                  0
                  \\
                  \msgFromTo{\procA}{\procB}{\msgM}. \,
                  \msgFromTo{\procA}{\procC}{\msgO}. \,
                  \msgFromTo{\procC}{\procB}{\msgO}. \,
                  \msgFromTo{\procB}{\procC}{\msgO}. \,
                  0
                \end{cases}
                &
		\qquad
\GG_s' = {} & + \;
		\begin{cases}
                  \msgFromTo{\procA}{\procB}{\msgO}. \,
                  \msgFromTo{\procC}{\procB}{\msgB}. \,
                  0
                  \\
                  \msgFromTo{\procA}{\procB}{\msgM}. \,
                  \msgFromTo{\procC}{\procB}{\msgB}. \,
                  0
                \end{cases}
	\end{align*}
      }

  \noindent Similar to $\GG_{oe}$, in all four examples, $\procA$ chooses a branch by sending either $\msgO$ or $\msgM$ to $\procB$. The global type $\GG_r$ is not implementable because $\procC$ cannot learn which branch was chosen by $\procA$. For any local implementation of $\procC$ to be able to execute both branches, it must be able to receive $\msgO$ from $\procA$ and $\procB$ in any order. Because the two send events $\snd{\procA}{\procC}{\msgO}$ and $\snd{\procB}{\procC}{\msgO}$ are independent of each other, they may be reordered. Consequently, any implementation of $\GG_r$ would have to permit executions that are consistent with global behaviors not described by $\GG_r$, such as $\msgFromTo{\procA}{\procB}{\msgM}.\,\msgFromTo{\procB}{\procC}{\msgO}.\,\msgFromTo{\procA}{\procC}{\msgO}$. Contrast this with $\GG_r'$, which is implementable. In the top branch of $\GG_r'$, role $\procA$ can only send to $\procC$ after it has received from $\procC$, which prevents the reordering of the send events $\snd{\procA}{\procC}{\msgO}$ and $\snd{\procB}{\procC}{\msgO}$. The bottom branch is symmetric. Hence, $\procC$ learns $\procA$'s choice based on which message it receives first.

    For the global type $\GG_s$, role $\procC$ again cannot learn the branch chosen by~$\procA$. That is, $\procC$ cannot know whether to send $\msgO$ or $\msgM$ to $\procB$, leading inevitably to deadlocking executions. In contrast, $\GG_s'$ is again implementable because the expected behavior of $\procC$ is independent of the choice by $\procA$.
\exampleend
\end{example}

These examples show that the implementability question is non-trivial. 
To check implementability, we present conditions that precisely characterize when the subset construction for $\GG$ yields an implementation.

\myparagraph{Overview.}
The rest of the paper is organized as follows. 
\S\ref{sec:prelim} contains relevant definitions for our work. 
\S\ref{sec:constructing-implementations} describes the synthesis step of our projection. 
\S\ref{sec:checking-conditions} presents the two conditions that characterize implementability of a given global type.
In \S\ref{sec:soundness}, we prove soundness of our projection via a stronger inductive invariant guaranteeing per-role agreement on a global run of the protocol.
In \S\ref{sec:completeness}, we prove completeness by showing that our two conditions hold if a global type is implementable. In \S\ref{sec:complexity}, we discuss the complexity of our construction and condition checks.
\S\ref{sec:eval} presents our artifact and evaluation, and \S\ref{sec:discussion} as well as \S\ref{sec:related} discuss related work.
\iftoggle{techrep}{}{
Additional details including omitted proofs can be found in the extended version of the paper~\cite{li2023complete}.
}
 \section{Preliminaries}
\label{sec:prelim}

\paragraph{Words.}
Let $\Alphabet$ be a finite alphabet.
$\Alphabet^*$ denotes the set of finite words over $\Alphabet$, $\Alphabet^\omega$ the set of infinite words, and $\Alphabet^\infty\negthinspace$ their union $\Alphabet^* \cup \Alphabet^\omega$.
A word $u \in \Alphabet^*$ is a \emph{prefix} of word $v \in \Alphabet^\infty$, denoted $u \leq v$, if there exists $w \in \Alphabet^\infty$ with $u \cdot w = v$.

\paragraph{Message Alphabet.}
Let $\Procs$ be a set of roles and $\MsgVals$ be a set of messages. We define the set of \emph{synchronous events} $\AlphSync \is \set{ \msgFromTo{\procA}{\procB}{\val} \mid \procA,\procB ∈ \Procs \text{ and } \val ∈ \MsgVals}$ where $\msgFromTo{\procA}{\procB}{\val}$ denotes that message $\val$ is sent by $\procA$ to $\procB$ atomically.
This is split for \emph{asynchronous events}.
For a role $\procA \in\Procs$, we define the alphabet
    $\Alphabet_{\procA,!} = \set{\snd{\procA}{\procB}{\val} \mid \procB \in \Procs,\; \val \in \MsgVals }$ of \emph{send} events
and the alphabet
    $\Alphabet_{\procA,?} = \set{\rcv{\procB}{\procA}{\val} \mid \procB \in \Procs,\; \val \in \MsgVals }$ of \emph{receive} events.
The event $\snd{\procA}{\procB}{\val}$ denotes role $\procA$ sending a message $\val$ to $\procB$,
and $\rcv{\procB}{\procA}{\val}$ denotes role $\procA$ receiving a message $\val$ from $\procB$.
We write $\Alphabet_{\procA} = \Alphabet_{\procA,!} \union \Alphabet_{\procA,?}$,
$\Alphabet_! = \Union_{\procA \in \Procs} \Alphabet_{\procA,!}$, and
$\Alphabet_? = \Union_{\procA \in \Procs} \Alphabet_{\procA,?}$.
Finally, $\AlphAsync = \Alphabet_! \cup \Alphabet_?$.
We say that $\procA$ is \emph{active} in $x \in \AlphAsync$ if $x \in \Alphabet_{\procA}$.
For each role $\procA\in \Procs$, we define a homomorphism~$\wproj_{\Alphabet_\procA}$, where $x \wproj_{\Alphabet_\procA} = x$ if $x \in \Alphabet_\procA$ and $\emptystring$ otherwise.
We write $\MsgVals(w)$ to project the send and receive events in $w$ onto their messages.
We fix $\Procs$ and~$\MsgVals$ in the rest of the paper.

\paragraph{Global Types -- Syntax.}
Global types for MSTs \cite{DBLP:conf/concur/MajumdarMSZ21} are defined by the grammar:
    \begin{grammar}
     G \is
       0
     | \sum_{i ∈ I} \msgFromTo{\procA}{\procB_{i}}{\val_i.G_i}
     | \mu t. \; G
     | t
    \end{grammar}\\[-3ex]
where $\procA, \procB_i$ range over $\Procs$, $\val_i$ over $\MsgVals$, and $t$ over a set of recursion variables.

We require each branch of a choice to be distinct: 
$∀ i,j ∈ I.\, i≠j ⇒ (\procB_{i},\val_i) ≠ (\procB_{j},\val_j)$, 
the sender and receiver of an atomic action to be distinct: 
$∀ i ∈ I.\, \procA ≠ \procB_i$,
and recursion to be guarded:
in $μ t. \, G$, there is at least one message between $μt$ and each $t$ in $G$.
When $|I| = 1$, we omit~$\sum$.
For readability, we sometimes use the infix operator $+$ for choice, instead of~$\sum$.
When working with a protocol described by a global type, we write $\GG$ to refer to the top-level type, and we use~$G$ to refer to its subterms.
For the size of a global type, we disregard multiple occurrences of the same subterm. 

We use the extended definition of global types from \cite{DBLP:conf/concur/MajumdarMSZ21} that allows a sender to send messages to different roles in a choice.
We call this \emph{sender-driven choice}, as in \cite{ecoop-draft}, while it was called generalized choice in \cite{DBLP:conf/concur/MajumdarMSZ21}.
This definition subsumes classical MSTs that only allow \emph{directed choice} \cite{DBLP:conf/popl/HondaYC08}.
The types we use focus on communication primitives and omit features like delegation or parametrization.
We defer a detailed discussion of different MST frameworks to \cref{sec:related}.

\paragraph{Global Types -- Semantics.}
As a basis for the semantics of a global type $\GG$, we construct a finite state machine
$
\semglobalsync(\GG) = (Q_{\GG}, \AlphSync, δ_{\GG}, q_{0, \GG}, F_{\GG})
$
where
\vspace{-1.5ex}
\begin{itemize}
	\item $Q_{\GG}$ is the set of all syntactic subterms in $\GG$ together with the term $0$,
	\item $δ_{\GG}$ is the smallest set containing
	$(\sum_{i ∈ I} \msgFromTo{\procA}{\procB_i}{\val_i.G_i}, \msgFromTo{\procA}{\procB_i}{\val_i}, G_i)$ for each $i ∈ I$,
	as well as $(μ t. G', ε, G')$ and $(t, ε, μ t. G')$ for each subterm~$\mu t.G'$, \item $q_{0, \GG} = \GG$ and
	$F_{\GG} = \set{0}$.
\end{itemize}
\vspace{-1.5ex}
We define a homomorphism 
$\SyncToAsync$ onto the asynchronous alphabet:
\[
\SyncToAsync(\msgFromTo{\procA}{\procB}{\val})
\is
\snd{\procA}{\procB}{\val}. \,
\rcv{\procA}{\procB}{\val}\enspace.
\]
The semantics $\lang(\GG)$ of a global type $\GG$ is given by
$\interswaplang(\SyncToAsync(\lang(\semglobalsync(\GG))))$ where $\interswaplang$ is the closure under the indistinguishability relation $\interswap$ \cite{DBLP:conf/concur/MajumdarMSZ21}.
Two events are independent if they are not related by the \emph{happened-before} relation \cite{DBLP:journals/cacm/Lamport78}. For instance, any two send events from distinct senders are independent.
Two words are indistinguishable if one can be reordered into the other by repeatedly swapping consecutive independent events. The full definition is in \appendixRef{app:indist-rel}.

\paragraph{Communicating State Machine \cite{DBLP:journals/jacm/BrandZ83}.}
\label{def:csm-formalisation}
$\mathcal{A} = \CSM{A}$ is a CSM over $\Procs$ and~$\MsgVals$ if
${A}_\procA$
is a finite state machine
over~$\Alphabet_\procA$ for every $\procA\in\Procs$, denoted by 
$(Q_\procA, \Alphabet_\procA, \delta_\procA, q_{0, \procA}, F_\procA)$.
Let 
$\prod_{\procA \in \Procs} s_\procA$ 
denote the set of global states and
\mbox{$\channels = \set{(\channel{\procA}{\procB}) \mid \procA,\procB\in \Procs, \procA\neq \procB}$}
denote the set of channels. 
A~\emph{configuration} of $\mathcal{A}$ is a pair $(\vec{s}, \xi)$, where $\vec{s}\,$ is a global state and
$\xi : \channels \rightarrow \MsgVals^*$
is a mapping from each channel to a sequence of messages.
We use $\vec{s}_\procA$ to denote the state of $\procA$ in $\vec{s}$.
The CSM transition relation, denoted $\rightarrow$, is defined as follows. 
\begin{itemize}
	\item
	$(\vec{s},\xi) \xrightarrow{\snd{\procA}{\procB}{\val}} (\pvec{s}',\xi')$ if
	$(\vec{s}_\procA, \snd{\procA}{\procB}{\val}, \pvec{s}'_\procA)\in\delta_\procA$,
	$\vec{s}_\procC = \pvec{s}'_\procC$ for every role $\procC \neq \procA$,
	$\xi'(\channel{\procA}{\procB}) =  \xi(\channel{\procA}{\procB})\cdot\val$ and $\xi'(c) = \xi(c)$ for every other channel $c\in \channels$.
	
	\item
	$(\vec{s},\xi) \xrightarrow{\rcv{\procA}{\procB}{\val}} (\pvec{s}',\xi')$ if
	$(\vec{s}_\procB, \rcv{\procA}{\procB}{\val}, \pvec{s}'_\procB)\in\delta_\procB$,
	$\vec{s}_\procC = \pvec{s}'_\procC$ for every role $\procC \neq \procB$,
	$\xi(\channel{\procA}{\procB}) = \val\cdot \xi'(\channel{\procA}{\procB})$
	and $\xi'(c) = \xi(c)$ for every other channel $c\in \channels$.
\end{itemize}
In the initial configuration $(\vec{s}_0, \xi_0)$, each role's state in $\vec{s}_0$ is the initial state $q_{0,\procA}$ of $A_\procA$, and $\xi_0$ maps each channel to $\emptystring$.
A configuration $(\vec{s}, \xi)$ is said to be \emph{final} iff $\vec{s}_\procA$ is final for every $\procA$ and $\xi$ maps each channel to~$\emptystring$.
Runs and traces are defined in the expected way. 
A run is \emph{maximal} if either it is finite and ends in a final configuration, or it is infinite. 
The language $\lang(\mathcal{A})$ of the CSM $\mathcal{A}$ is defined as the set of maximal traces.
A configuration $(\vec{s}, \xi)$ is a \emph{deadlock} if it is not final and has no outgoing transitions.
A CSM is \emph{deadlock-free} if no reachable configuration is a deadlock. 

\medskip
Finally, implementability is formalized as follows. 
\begin{definition}[Implementability~\cite{DBLP:conf/concur/MajumdarMSZ21}]
\label{def:implementability}
A global type $\GG$ is \emph{implementable} if there exists a CSM $\CSM{A}$ such that the following two properties hold: \newline
\begin{inparaenum}[(i)]
 \item \label{def:implementability-protocol-fidelity}
 \emph{protocol fidelity}: $\lang(\CSM{A}) = \lang(\GG)$, and
 \item \label{def:implementability-deadlock-freedom}
 \emph{deadlock freedom}: $\CSM{A}$ is deadlock-free.
\end{inparaenum}
We say that $\CSM{A}$ implements $\GG$.
\end{definition}

 \section{Synthesizing Implementations}
\label{sec:constructing-implementations}
\label{sec:synthesizing-implementations}

The construction is carried out in two steps.
First, for each role $\procA \in \Procs$, we define an intermediate state machine $\projerasure{\GG}{\procA}$ that is a homomorphism of $\semglobal(\GG)$.
We call $\projerasure{\GG}{\procA}$ the \emph{projection by erasure} for $\procA$, defined below.

\begin{definition}[Projection by Erasure]
	\label{def:projection-by-erasure}
	Let $\GG$ be some global type
	with its state machine
	$
	\semglobalsync(\GG) =
	(Q_{\GG},
	\AlphSync,
	\delta_{\GG},
	q_{0, \GG},
	F_{\GG})
	$.
	For each role $\procA \in \Procs$, we define the state machine
	$
	\projerasure{\GG}{\procA} \,=
	(Q_{\GG},
	\Alphabet_\procA \dunion \set{\emptystring},
	\projerasuretrans,
	q_{0, \GG},
	F_{\GG})
	$
	where
	$\projerasuretrans \is
	\set{q \xrightarrow{\SyncToAsync(a) \wproj_{\Alphabet_\procA}} q'
		\mid q \xrightarrow{a} q' \in \delta_{\GG}}$.
	By definition of $\SyncToAsync(\hole)$, it holds that $\SyncToAsync(a) \wproj_{\Alphabet_\procA} \in \Alphabet_\procA \dunion \set{\emptystring}$.
\end{definition}

\noindent
Then, we determinize $\projerasure{\GG}{\procA}$ via a standard subset construction to obtain a deterministic local state machine for $\procA$. 

\begin{definition}[Subset Construction]
\label{def:subset-construction}
Let $\GG$ be a global type and $\procA$ be a role. Then, the \emph{subset construction} for $\procA$ is defined as
\[
\subsetcons{\GG}{\procA} =
\bigl(
Q_{\procA},
\Alphabet_\procA,
\delta_{\procA},
s_{0, \procA},
F_{\procA}
\bigr)
\text{ where }
\]
\begin{itemize}
\item $ \delta(s, a) \is
	\set{q' \in Q_{\GG}
		\mid
		\exists q \in s,
		q \xrightarrow{a} \xrightarrow{\emptystring}\mathrel{\vphantom{\to}^*} q' \in \projerasuretrans
	},
	$
	for
	every
	$s \subseteq Q_{\GG}$ and
	$a \in \Alphabet_{\procA}$\item $s_{0, \procA} \is
	\set{q \in Q_{\GG} \mid
		q_{0, \GG} \xrightarrow{\emptystring} \mathrel{\vphantom{\to}^*} q \in \projerasuretrans}$,
	\item $Q_{\procA} \is \lfp_{\set{s_{0,\procA}}}^\subseteq \lambda Q.\, Q \union \set{ \delta(s,a) \mid s \in Q \land a \in \Alphabet_{\procA}} \setminus \set{\emptyset}$
	, and
\item $\delta_{\procA} \is \restrict{\delta}{Q_{\procA} \times \Alphabet_{\procA}}$
	\item $F_{\procA} \is
	\set{s \in Q_{\procA}
		\mid s \inters F_{\GG} \neq \emptyset}$
\end{itemize}
\end{definition}

Note that the construction ensures that $Q_\procA$ only contains subsets of $Q_{\GG}$ whose states are reachable via the same traces, i.e. we typically have $|Q_{\procA}| \ll 2^{|Q_{\GG}|}$.

The following characterization is immediate from the subset construction; the proof can be found in \appendixRef{proof:synthesizing-implementations}.

\begin{restatable}{lemma}{constructionProperties}
	\label{lm:projection-preserves-per-process-runs}
	\label{lm:languages-of-roles}
	Let $\GG$ be a global type, $\procC$ be a role, and
	$\subsetcons{\GG}{\procC}$
	be its \emph{subset construction}.
	If $w$ is a trace of $\semglobalsync(\GG)$, $\SyncToAsync(w) \wproj_{\Alphabet_\procC}$ is a trace of $\subsetcons{\GG}{\procC}$.
	If $u$ is a trace of $\subsetcons{\GG}{\procC}$, there is a trace $w$ of $\semglobalsync(\GG)$ such that $\SyncToAsync(w) \wproj_{\Alphabet_\procC} = u$.
	It holds that
	$\lang(\GG) \wproj_{\Alphabet_\procC} = \lang(\subsetcons{\GG}{\procC})$.
\end{restatable}

Using this lemma, we show that the CSM $\CSMl{\subsetcons{\GG}{\procA}}$ preserves all behaviors of $\GG$.

\begin{restatable}{lemma}{constructionPreservesBehaviors}
	\label{lm:constructionPreservesBehaviors}
	For all global types $\GG$, $\lang(\GG) \subseteq \lang(\CSMl{\subsetcons{\GG}{\procA}})$.
\end{restatable}
We briefly sketch the proof here.
	 	Given that $\CSMl{\subsetcons{\GG}{\procA}}$ is deterministic, to prove language inclusion it suffices to prove the inclusion of the respective prefix sets:
	 	\[
	 		\text{pref}(\lang(\GG)) \subseteq \text{pref}(\lang\CSMl{\subsetcons{\GG}{\procA}})
	 	\]
	 	Let $w$ be a word in $\lang(\GG)$. 
	 	If $w$ is finite, membership in $\lang(\CSMl{\subsetcons{\GG}{\procA}})$ is immediate from the claim above.
	 	If $w$ is infinite, we show that $w$ has an infinite run in $\CSMl{\subsetcons{\GG}{\procA}}$ using König's Lemma.
	 	We construct an infinite graph $\mathcal{G}_w(V, E)$ with 
	 	$V \is \{v_{\run} \mid \trace(\run) \leq w\}$ and 
	 	$E \is \{(v_{\run_1}, v_{\run_2}) \mid \exists~x \in \AlphAsync.~\trace(\run_2) = \trace(\run_1)\cdot x\}$. 
	 	Because $\CSMl{\subsetcons{\GG}{\procA}}$ is deterministic, $\mathcal{G}_w$ is a tree rooted at $v_\emptystring$, the vertex corresponding to the empty run.
	 	By König's Lemma, every infinite tree contains either a vertex of infinite degree or an infinite path.
	 	Because $\CSMl{\subsetcons{\GG}{\procA}}$ consists of a finite number of communicating state machines, the last configuration of any run has a finite number of next configurations, and $\mathcal{G}_w$ is finitely branching.
	 	Therefore, there must exist an infinite path in $\mathcal{G}_w$ representing an infinite run for $w$, and thus
	 	$w \in \lang(\CSMl{\subsetcons{\GG}{\procA}})$.
	
	     The proof of the inclusion of prefix sets proceeds by structural induction and primarily relies on Lemma \ref{lm:languages-of-roles} and the fact that all prefixes in $\lang(\GG)$ respect the order of send before receive events.

\section{Checking Implementability}
\label{sec:checking-conditions}

We now turn our attention to checking implementability of a CSM produced by the subset construction. 
We revisit the global types from \cref{ex:implementability} (also shown in \cref{fig:implementability}), which demonstrate that the naive subset construction does not always yield a sound implementation. 
From these examples, we distill our conditions that precisely identify the implementable global types.

In general, a global type $\GG$ is not implementable when the agreement on a global run of $\semglobal(\GG)$ among all participating roles cannot be conveyed via sending and receiving messages alone.
When this happens, roles can take locally permitted transitions that commit to incompatible global runs, resulting in a trace that is not specified by $\GG$. Consequently, our conditions need to ensure that when a role $\procA$ takes a transition in $\subsetcons{\GG}{\procA}$, it~only commits to global runs that are consistent with the local views of all other roles.
We discuss the relevant conditions imposed on send and receive transitions separately. 

\myparagraph{Send Validity.} 
Consider $\GG_s$ from \cref{ex:implementability}. The CSM
$\CSMl{\subsetcons{\GG_s}{\procA}}$ has an execution with the trace
$
\snd{\procA}{\procB}{\msgO}
\cdot 
\rcv{\procA}{\procB}{\msgO}
\cdot 
\snd{\procC}{\procB}{\msgM}
$.
This trace is possible because the initial state of $\subsetcons{\GG_s}{\procC}$, $s_{0,\procC}$, contains two states of $\projerasure{\GG_s}{\procC}$, each of which has a single outgoing send transition labeled with $\snd{\procC}{\procB}{\msgO}$ and $\snd{\procC}{\procB}{\msgM}$ respectively.
Both of these transitions are always enabled in~$s_{0,\procC}$, meaning that $\procC$ can send $\snd{\procC}{\procB}{\msgM}$ even when $\procA$ has chosen the top branch and $\procB$ expects to receive~$\msgO$ instead of $\msgM$ from $\procC$. This results in a deadlock.
In contrast, while the state $s_{0,\procC}$ in $\subsetcons{\GG_s'}{\procC}$ likewise contains two states of $\projerasure{\GG_s'}{\procC}$, each with a single outgoing send transition, now both transitions are labeled with $\snd{\procC}{\procB}{\msgB}$.
These two transitions collapse to a single one in $\subsetcons{\GG_s'}{\procC}$. This transition is consistent with both possible local views that $\procA$ and~$\procB$ might hold on the global~run.

Intuitively, to prevent the emergence of inconsistent local views from send transitions of $\subsetcons{\GG}{\procA}$, we must enforce that for every state $s \in Q_{\procA}$ with an outgoing send transition labeled $x$, a transition labeled $x$ must be enabled in all states of $\projerasure{\GG}{\procA}$ represented by $s$.
We use the following auxiliary definition to formalize this intuition subsequently.

\begin{definition}[Transition Origin and Destination]
	Let $s \xrightarrow{x} s' \in \delta_\procA$ be a transition in $\subsetcons{\GG}{\procA}$ and $\projerasuretrans$ be the transition relation of $\projerasure{\GG}{\procA}$.
	We define the set of \emph{transition origins} $\transAnnoFunc(s \xrightarrow{x} s')$ and \emph{transition destinations} $\transAnnoFunDest(s \xrightarrow{x} s')$ as follows:
	\begin{align*}
		\transAnnoFunc(s \xrightarrow{x} s')
		\is {} &
		\set{G \in s
			\mid
			\exists G' \in s'. \,
			G \xrightarrow{x}\mathrel{\vphantom{\to}^*} G' \in \projerasuretrans} \text{ and }\\
		\transAnnoFunDest(s \xrightarrow{x} s')
		\is {} &
		\set{G' \in s'
			\mid
			\exists G \in s. \,
			G \xrightarrow{x}\mathrel{\vphantom{\to}^*} G' \in \projerasuretrans} \enspace.
	\end{align*}
\end{definition}

Our condition on send transitions is then stated below. 

\begin{definition}[Send Validity]
	\label{cond:send-state-validity-transition-origins}
	$\subsetcons{\GG}{\procA}$ satisfies \emph{Send Validity} iff every send transition $s  \xrightarrow{x} s' \in \delta_\procA$ is enabled in all states contained in $s$: 
	\[
	\forall s  \xrightarrow{x} s' \in \delta_\procA.
	~x \in \Alphabet_{\procA,!} \implies
	\transAnnoFunc(s  \xrightarrow{x} s') = s \enspace.
	\]
\end{definition}

\myparagraph{Receive Validity.}
To motivate our condition on receive transitions, let us revisit $\GG_r$ from \cref{ex:implementability}. The CSM $\CSMl{\subsetcons{\GG_r}{\procA}}$ recognizes the following trace not in the global type language $\lang(\GG_r)$:
\vspace{-0.7ex}
\begin{center}
	$
	\snd{\procA}{\procB}{\msgO}
	\cdot 
	\rcv{\procA}{\procB}{\msgO}
	\cdot 
	\snd{\procB}{\procC}{\msgO}
	\cdot 
	\snd{\procA}{\procC}{\msgO}
	\cdot
	\rcv{\procA}{\procC}{\msgO}
	\cdot
	\rcv{\procB}{\procC}{\msgO} \enspace.
	$
\end{center}
\vspace{-0.7ex}
The issue lies with $\procC$ which cannot distinguish between the two branches in $\GG_r$.
The initial state $s_{0,\procC}$ of $\subsetcons{\GG_r}{\procC}$ has two states of $\semglobal(\GG_r)$ corresponding to the subterms
$G_t \is \msgFromTo{\procB}{\procC}{\msgO}. \, \msgFromTo{\procA}{\procC}{\msgO}. \,0$
and
$G_b \is \msgFromTo{\procA}{\procC}{\msgO}. \, \msgFromTo{\procB}{\procC}{\msgO}. \,0\,$.
Here, $G_t$ and $G_b$ are the top and bottom branch of $\GG_r$ respectively. This means that there are outgoing transitions in $s_{0,\procC}$ labeled with $\rcv{\procA}{\procC}{\msgO}$ and $\rcv{\procB}{\procC}{\msgO}$.
If $\procC$ takes the transition labeled $\rcv{\procA}{\procC}{\msgO}$, it commits to the bottom branch $G_b$.
However, observe that the message $\msgO$ from $\procA$ can also be available at this time point if the other roles follow the top branch $G_t$.
This is because $\procA$ can send $\msgO$ to~$\procC$ without waiting for $\procC$ to first receive from $\procB$.
In this scenario, the roles disagree on which global run of $\semglobal(\GG_r)$ to follow, resulting in the violating trace above.

Contrast this with $\GG_r'$. Here, $s_{0,\procC}$ again has outgoing transitions labeled with $\rcv{\procA}{\procC}{\msgO}$ and $\rcv{\procB}{\procC}{\msgO}$. However, if $\procC$ takes the transition labeled $\rcv{\procA}{\procC}{\msgO}$, committing to the bottom branch, no disagreement occurs. This is because if the other roles are following the top branch, then $\procA$ is blocked from sending to $\procC$ until after it has received confirmation that $\procC$ has received its first message from $\procB$.

For a receive transition $s \xrightarrow{x} s_1$ in $\subsetcons{\GG}{\procA}$ to be safe, we must enforce that the receive event $x$ cannot also be available due to reordered sent messages in the continuation $G_2 \in s_2$ of another outgoing receive transition $s \xrightarrow{y} s_2$. To formalize this condition, we use the set $\semavail^\blockedset_{(G \ldots)}$ of \emph{available messages} for a syntactic subterm $G$ of $\GG$ and a set of \emph{blocked} roles $\blockedset$. This notion was already defined in \cite[Sec.\,2.2]{DBLP:conf/concur/MajumdarMSZ21}. Intuitively, $\semavail^\blockedset_{(G \ldots)}$ consists of all send events $\snd{\procB}{\procC}{\val}$ that can occur on the traces of $G$ such that
$\val$ will be the first message added to channel~$(\channel{\procB}{\procC})$ before any of the roles in $\blockedset$ takes a step.
\paragraph{Available messages.}
The set of available messages is recursively defined on the structure of the global type.
To obtain all possible messages, we need to unfold the distinct recursion variables once.
For this, we define a map $\getMu$ from variable to subterms and write $\getMuG$ for $\getMu(\GG)$:

\vspace{0.7ex}
\begin{footnotesize}
	{} \hfill
	$
	\getMu(0) \is [\,]
	$ \hfill $
	\getMu(t) \is [\,]
	$ \hfill $
	\getMu(μt.G) \is [t \mapsto G] ∪ \getMu(G)
	$  \hfill {}
	
	{} \hfill $
	\getMu(\sum_{i ∈ I} \msgFromTo{\procA}{\procB_i}{\val_i.G_i}) \is \bigcup_{i∈I} \getMu(G_i)
	$  \hfill {}
\end{footnotesize}
\vspace{0.5ex}

\noindent The function $\semavaildef{\blockedset}{T}{\hole}$ keeps a set of unfolded variables $T$, which is empty initially.

\begin{center}
	\begin{minipage}{0.98\textwidth}
		\begin{small}
			\noindent
			$
			\semavaildef{\blockedset}{T}{0} \is ∅ \hfill
			\semavaildef{\blockedset}{T}{μt.G} \is \semavaildef{\blockedset}{T ∪ \set{t}}{G} \hfill
			\semavaildef{\blockedset}{T}{t} \is \begin{cases}
				∅ & \text{if} ~ t ∈ T\\
				\semavaildef{\blockedset}{T ∪ \set{t}}{\getMuG(t)} & \text{if} ~ t ∉ T
			\end{cases}\\[2mm]
			\semavaildef{\blockedset}{T}{\sum_{i ∈ I} \msgFromTo{\procA}{\procB_i}{\val_i.G_i}}
			\is
			\begin{cases}
				\bigcup_{i∈I,m∈\MsgVals} (\semavaildef{\blockedset}{T}{G_i} \setminus \set{ \snd{\procA}{\procB_i}{\val} }) ∪  \set{ \snd{\procA}{\procB_i}{\val_i} } \quad \hfill \text{if} ~ \procA ∉ \blockedset \\
				\bigcup_{i∈I} \semavaildef{\blockedset ∪ \set{ \procB_i }}{T}{G_i} \quad \hfill \text{if} ~ \procA ∈ \blockedset
			\end{cases}
			$
		\end{small}
	\end{minipage}
\end{center}

\noindent We write $\semavail^{\blockedset}_{(G \ldots)}$ for $\semavaildef{\blockedset}{\emptyset}{G}$.
If $\blockedset$ is a singleton set, we omit set notation and write $\semavail^{\procA}_{(G \ldots)}$
for $\semavail^{\set{\procA}}_{(G \ldots)}$.
The set of available messages captures the possible states of all channels before a given receive transition is taken.

\begin{definition}[Receive Validity]
\label{cond:rcv-state-validity}
$\subsetcons{\GG}{\procA}$ satisfies \emph{Receive Validity} iff no receive transition is enabled in an alternative continuation that originates from the same source state:
\[
\begin{array}{l}
	\forall 
	s \xrightarrow{\rcv{\procB_1}{\procA}{\val_1}} s_1,\,
	s \xrightarrow{\rcv{\procB_2}{\procA}{\val_2}} s_2 \in \delta_\procA.\, \\
	\qquad  \procB_1 \neq \procB_2
	\; \implies \;
	\forall~G_2 \in \transAnnoFunDest(s \xrightarrow{\rcv{\procB_2}{\procA}{\val_2}} s_2). \;
	\snd{\procB_1}{\procA}{\val_1} \notin \semavail^{\procA}_{(G_2 \ldots)} \enspace.
\end{array}
\]
\end{definition}

\myparagraph{Subset Projection.}
We are now ready to define our projection operator.
\begin{definition}[Subset Projection of $\GG$]
\label{def:conditions-powerset-proj}
	The \emph{subset projection} $\subsetproj{\GG}{\procA}$ of $\GG$ onto $\procA$ is $\subsetcons{\GG}{\procA}$ if it satisfies Send Validity and Receive Validity. We lift this operation to a partial function from global types to CSMs in the expected~way.\end{definition}

We conclude our discussion with an observation about the syntactic structure of the subset projection: \noindent Send Validity implies that no state has both outgoing send and receive transitions (also known as mixed choice).

\begin{corollary}[No Mixed Choice]
	\label{cor:no-mixed-choice}If $\subsetproj{\GG}{\procA}$ satisfies Send Validity, then for all $s  \xrightarrow{x_1} s_1, s \xrightarrow{x_2} s_2 \in \delta_\procA$, $x_1 \in \Alphabet_!$ iff $x_2 \in \Alphabet_!$.
\end{corollary}

 \section{Soundness}
\label{sec:soundness}

In this section, we prove the soundness of our subset projection, stated as follows. 

\begin{restatable}{theorem}{soundnessTheorem}
	\label{thm:soundness}
	Let $\GG$ be a global type and $\CSMl{\subsetproj{\GG}{\procA}}$ be the subset projection.
	Then, $\CSMl{\subsetproj{\GG}{\procA}}$
	implements $\GG$. 
\end{restatable}

Recall that implementability is defined as protocol fidelity and deadlock freedom. 
Protocol fidelity consists of two language inclusions. 
The first inclusion, 
$\lang(\GG) \subseteq \lang(\CSMl{\subsetproj{\GG}{\procA}})$, 
enforces that the subset projection generates at least all behaviors of the global type.
We showed in \cref{lm:constructionPreservesBehaviors} that this holds for the subset construction alone (without Send and Receive Validity).

The second inclusion, 
$\lang(\CSMl{\subsetproj{\GG}{\procA}}) \subseteq \lang(\GG)$, 
enforces that no new behaviors are introduced. 
The proof of this direction relies on a stronger inductive invariant that we show for all traces of the subset projection. 
As discussed in \S\ref{sec:checking-conditions}, violations of implementability occur when roles commit to global runs that are inconsistent with the local views of other roles.  
Our inductive invariant states the exact opposite: that all local views are consistent with one another.
First, we formalize the local view of a role.

\begin{definition}[Possible run sets]
	\label{def:possible-run-sets}
	Let $\GG$ be a global type and $\semglobalsync(\GG)$ be the corresponding state machine.
	Let $\procA$ be a role and
	$w \in \AlphAsync^*$ be a word.
	We define the set of possible runs $\globcomplocal{\GG}{\procA}{w}$
	as all maximal runs of $\semglobalsync(\GG)$ that are consistent with $\procA$'s local view of $w$:
	\[
	\globcomplocal{\GG}{\procA}{w}
	\is
	\set{
		\run
		\text{ is a maximal run of }
		\semglobalsync(\GG)
		\mid
		w \wproj_{\Alphabet_\procA} \preforder \SyncToAsync(\trace(\run)) \wproj _{\Alphabet_\procA}
	}
    \enspace .
	\]
\end{definition}

While \cref{def:possible-run-sets} captures the set of maximal runs that are consistent
with the local view of a single role, we would like to refer to the set of runs that is consistent with the local view of all roles. We formalize this as the intersection of the possible run sets for all roles, which we denote as 
\[
	I(w) 
	\is 
	\Inters_{\procA \in \Procs} \globcomplocal{\GG}{\procA}{w}
	\enspace.
\]
With these definitions in hand, we can now formulate our inductive invariant: 
\begin{restatable}{lemma}{intersNonempty}
	\label{lm:intersNonempty}
	Let $\GG$ be a global type and $\CSMl{\subsetproj{\GG}{\procA}}$ be the subset projection.
	Let $w$ be a trace of $\CSMl{\subsetproj{\GG}{\procA}}$.
	It holds that
	$I(w)$ is non-empty.
\end{restatable}

The reasoning for the sufficiency of \cref{lm:intersNonempty} is included in the proof of \cref{thm:soundness}, found in \appendixRef{app:soundness}.
In the rest of this section, we focus our efforts on how to show this inductive invariant, namely that the intersection of all roles' possible run sets is non-empty.

We begin with the observation that the empty trace $\emptystring$ is consistent with all runs.
As a result, 
$I(\emptystring) =
\Inters_{\procA \in \Procs} \globcomplocal{\GG}{\procA}{\emptystring}$
contains all maximal runs in $\semglobal(\GG)$.
By definition, state machines for global types include at least one run, and the base case is trivially discharged. Intuitively, $I(w)$ shrinks as more events are appended to $w$, but we show that at no point does it shrink to $\emptyset$.
We consider the cases where a send or receive event is appended to the trace separately, and show that the intersection set shrinks in a principled way that preserves non-emptiness. 
In fact, when a trace is extended with a receive event, Receive Validity guarantees that the intersection set does not shrink at all. 

\begin{figure}[t]
    \centering
    \def\boundingbox{(-3,-1.2) rectangle (3,1.2)}
\def\preRall{(0,0) circle (0.5)}
\def\preRp{(-1 cm,0) circle[x radius = 1.8 cm, y radius = 1 cm]}
\def\preRq{(1 cm,0) circle[x radius = 1.8 cm, y radius = 1 cm]}
\def\bottomHalf{(-3,-1.2) rectangle (3,0)}
\def\topHalf{(-3,0) rectangle (3,1.2)}
\def\leftHalf{(-3,-1.2) rectangle (0,1.2)}
\def\rightHalf{(0,-1.2) rectangle (3,1.2)}
\def\postRq{}

\tikzset{
    myoutline/.style={draw=blue!80, thick}
}

\pgfdeclarepattern{
  name=stripes,
  parameters={\stripesoffset,\stripeslinewidth},
  bottom left={\pgfpoint{0pt}{0pt}},
  top right={\pgfpoint{10pt}{10pt}},
  tile size={\pgfpoint{10pt}{10pt}},
  code={
    \pgfsetlinewidth{\stripeslinewidth}
    \pgfpathmoveto{\pgfpoint{\stripesoffset - 5pt}{- 5pt}}
    \pgfpathlineto{\pgfpoint{\stripesoffset +15pt}{ 15pt}}
    \pgfpathmoveto{\pgfpoint{\stripesoffset + 5pt}{- 5pt}}
    \pgfpathlineto{\pgfpoint{\stripesoffset +25pt}{ 15pt}}
    \pgfpathmoveto{\pgfpoint{\stripesoffset -15pt}{- 5pt}}
    \pgfpathlineto{\pgfpoint{\stripesoffset + 5pt}{ 15pt}}
    \pgfusepath{stroke}
  }
}

\tikzset{
  stripes offset/.store in=\stripesoffset,
  stripes line width/.store in=\stripeslinewidth,
  stripes offset=0pt,
  stripes line width=3.5pt,
}

\begin{tikzpicture}
    \begin{scope}[xshift=-3.2cm]
\draw \boundingbox ;
        \node[anchor=south] at (0,1.3 cm) {$x = \snd{\procA}{\procB}{\val}$, $w \in \AlphAsync^*$};
\draw[pattern=stripes, pattern color=blue!20, stripes offset=5.0pt] \preRp;
\begin{scope}
            \clip \preRp;
            \fill[pattern=stripes, pattern color=red!20] \bottomHalf;
            \draw[myoutline,dashed] \bottomHalf;
        \end{scope}
\draw[myoutline] \preRp node[xshift=-0.5 cm,yshift=-0.5 cm] {$\globcomplocal{\GG}{\procA}{wx}$};
\draw[myoutline] \preRall;
        \node at (1.8cm,0.8cm) {$\Inters_{\procC \in \Procs} \globcomplocal{\GG}{\procC}{w}$};
        \draw[semithick] (0.8cm,0.6cm) -- (0.2cm,0.2cm);
\draw[myoutline] \preRp node[xshift=-0.5 cm,yshift=0.5 cm] {$\globcomplocal{\GG}{\procA}{w}$};
\end{scope}
    \begin{scope}[xshift=3.2cm]
\draw \boundingbox ;
        \node[anchor=south] at (0,1.3 cm) {$y = \rcv{\procA}{\procB}{\val}$, $w' = wxu$ with $u \in \AlphAsync^*$};
\begin{scope}
            \clip \preRq;
            \fill[pattern=stripes, pattern color=red!20] \bottomHalf;
        \end{scope}
\fill[pattern=stripes, pattern color=blue!20, stripes offset=5.0pt] \preRq;
\begin{scope}
            \clip \preRq;
            \draw[myoutline,dashed] \bottomHalf;
        \end{scope}
\begin{scope}
            \clip \preRp;
            \draw[myoutline] \bottomHalf;
        \end{scope}
        \begin{scope}
            \clip \bottomHalf;
\draw[myoutline] \preRp node[xshift=-0.5 cm,yshift=-0.5 cm] {$\globcomplocal{\GG}{\procA}{w'}$};
\draw[myoutline] \preRall;
\draw[myoutline] \preRq node[xshift=0.5 cm,yshift=-0.5 cm] {$\globcomplocal{\GG}{\procB}{w'y}$};
        \end{scope}
        \draw[myoutline] \preRq node[xshift=0.5 cm,yshift=0.5 cm] {$\globcomplocal{\GG}{\procB}{w'}$};
\node at (-1.7cm,0.7cm) {$\Inters_{\procC \in \Procs} \globcomplocal{\GG}{\procC}{w'}$};
        \draw[semithick] (-0.8cm,0.5cm) -- (-0.2cm,-0.1cm);
    \end{scope}
\end{tikzpicture}
     \vspace{-3mm}
    \caption{Evolution of $\globcomplocal{\GG}{\hole}{\hole}$ sets when $\procA$ sends a message $\val$ and $\procB$ receives it.}
    \label{fig:venn-diagram-R-shrink}
\end{figure}
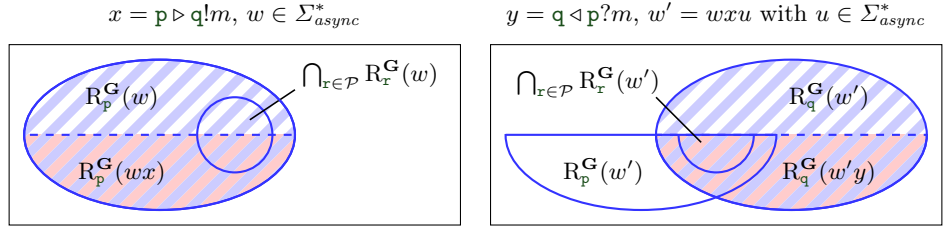

\begin{restatable}{lemma}{rcvIntersectionSetEquality}
	\label{lm:rcvIntersectionSetEquality}
	Let $\GG$ be a global type and $\CSMl{\subsetproj{\GG}{\procA}}$ be the subset projection.
	Let $wx$ be a trace of $\CSMl{\subsetproj{\GG}{\procA}}$ such that $x \in \Alphabet_?$. 
	Then, $I(w) = I(wx)$. 
\end{restatable}

To prove this equality, we further refine our characterization of intersection sets. 
In particular, we show that in the receive case, the intersection between the sender and receiver's possible run sets stays the same, i.e.\
\[
\globcomplocal{\GG}{\procA}{w} \inters \globcomplocal{\GG}{\procB}{w} = 										
\globcomplocal{\GG}{\procA}{wx} \inters \globcomplocal{\GG}{\procB}{wx}
\enspace.
\]
Note that it is not the case that the receiver only follows a subset of the sender's possible runs.
In other words, 
$\globcomplocal{\GG}{\procB}{w} \subseteq \globcomplocal{\GG}{\procA}{w}$ 
is not inductive. 
The equality above simply states that a receive action can only eliminate runs that have already been eliminated by its sender.
\cref{fig:venn-diagram-R-shrink} depicts this relation.

Given that the intersection set strictly shrinks, the burden of eliminating runs must then fall upon send events. 
We show that send transitions shrink the possible run set of the sender in a way that is \emph{prefix-preserving}. 
To make this more precise, we introduce the following definition on runs.

\begin{definition}[Unique splitting of a possible run]
	Let $\GG$ be a global type, $\procA$ a role, and $w \in \AlphAsync^*$ a word. Let $\run$ be a possible run in $\globcomplocal{\GG}{\procA}{w}$. 
	We define the longest prefix of $\run$ matching $w$:
	\[
	\alpha'
	\is
	\max
	\set{
		\run'
		\mid
		\run' \leq \run ~\wedge~
		\SyncToAsync(\trace(\run')) \wproj _{\Alphabet_\procA} \preforder w \wproj_{\Alphabet_\procA}
	}
	\enspace .
	\]
	If $\alpha' \neq \run$, we can split $\run$ into
$
	\run = \alpha \cdot G \xrightarrow{l} G' \cdot \beta
	$
where $\alpha' = \alpha \cdot G$, $G'$ denotes the state following $G$, and $\beta$ denotes the suffix of $\run$ following $\alpha \cdot G \cdot G'$.
We call $\alpha \cdot G \xrightarrow{l} G' \cdot \beta$ the unique splitting of $\run$ for $\procA$ matching $w$. 
	We omit the role $\procA$ when obvious from context.
	This splitting is always unique because the maximal prefix of any $\run \in \globcomplocal{\GG}{\procA}{w}$ matching $w$ is unique.
\end{definition}

When role $\procA$ fires a send transition $\snd{\procA}{\procB}{\val}$,
any run $\run = \alpha \cdot G \xrightarrow{l} G' \cdot \beta$
in $\procA$'s possible run with $\SyncToAsync(l) \wproj_{\Alphabet_\procA} \neq \snd{\procA}{\procB}{\val}$ is eliminated.
While the resulting possible run set could no longer contain runs that end with $G' \cdot \beta$, Send Validity guarantees that it must contain runs that begin with $\alpha \cdot G$. 
This is formalized by the following lemma. 

\begin{restatable}{lemma}{sndPrefixPreservation}
	\label{lm:sndPrefixPreservation}
	Let $\GG$ be a global type and $\CSMl{\subsetproj{\GG}{\procA}}$ be the subset projection.
	Let $wx$ be a trace of
	$\CSMl{\subsetproj{\GG}{\procA}}$ such that $x \in \Alphabet_!
	\cap \Alphabet_\procA$ for some $\procA \in \Procs$.
	Let $\rho$ be a run in $I(w)$, and $\alpha \cdot G \xrightarrow{l} G' \cdot \beta$ be the unique splitting of $\rho$ for $\procA$ with respect to $w$.
	Then, there exists a run $\rho'$ in $I(wx)$ such that $\alpha \cdot G \leq \rho'$.
\end{restatable}

This concludes our discussion of the send and receive cases in the inductive step to show the non-emptiness of the intersection of all roles' possible run sets. 
The full proofs and additional definitions can be found in \appendixRef{app:soundness}.

 \section{Completeness}
\label{sec:completeness}

In this section, we prove completeness of our approach.
While soundness states that if a global type's subset projection is defined, it then implements the global type, completeness considers the reverse direction. 

\begin{restatable}[Completeness]{theorem}{completenessThm}
\label{thm:completeness}
If $\GG$ is implementable, then
$\CSMl{ \subsetproj{\GG}{\procA} }$ is defined.
\end{restatable}

We sketch the proof and refer to \appendixRef{app:completeness} for the full proof.

From the assumption that $\GG$ is implementable, we know there exists a witness CSM that implements $\GG$.
While the soundness proof picks our subset projection as the existential witness for showing implementability -- thereby allowing us to reason directly about a particular implementation -- completeness only guarantees the existence of some witness CSM.
We cannot assume without loss of generality that this witness CSM is our subset construction; however, we must use the fact that it implements $\GG$ to show that Send and Receive Validity hold on our subset construction.

We proceed via proof by contradiction: we assume the negation of Send and Receive Validity for the subset construction, and show a contradiction to the fact that this witness CSM implements $\GG$.
In particular, we contradict protocol fidelity (\cref{def:implementability}(\ref{def:implementability-protocol-fidelity})), stating that the witness CSM generates precisely the language $\lang(\GG)$.
To do so, we exploit a simulation argument: we first show that the negation of Send and Receive Validity forces the subset construction to recognize a trace that is not a prefix of any word in $\lang(\GG)$.
Then, we show that this trace must also be recognized by the witness CSM, under the assumption that the witness CSM implements $\GG$.

To highlight the constructive nature of our proof, we convert our proof obligation to a witness construction obligation. 
To contradict protocol fidelity, it suffices to construct a witness trace $v_0$ satisfying two~properties, where $\CSM{B}$ is our witness CSM:
\begin{enumerate}[(a)]
	\item \label{cond:compl-proof-1}
	$v_0$ is a trace of $\CSM{B}$, and
	\item \label{cond:compl-proof-2}
	the run intersection set of $v_0$
	is empty:
	$I(v_0)
        =
	\Inters_{\procA \in \Procs} \globcomplocal{\GG}{\procA}{v_0}
        =
	\emptyset$.
\end{enumerate}

We first establish the sufficiency of conditions \ref{cond:compl-proof-1} and \ref{cond:compl-proof-2}.
Because $\CSM{B}$ is deadlock-free by assumption, every prefix extends to a maximal trace.
Thus, to prove the inequality of the two languages $\lang(\CSM{B})$ and $\lang(\GG)$, it suffices to prove the inequality of their respective prefix sets.  In turn, it suffices to show the existence of a prefix of a word in one language that is not a prefix of any word in the other.
We choose to construct a prefix in the CSM language that is not a prefix in $\lang(\GG)$.
We again leverage the definition of intersection sets~(\cref{def:possible-run-sets}) to weaken the property of language non-membership to the property of having an empty intersection set as follows.
By the semantics of $\lang(\GG)$, for any $w \in \lang(\GG)$, there exists $w' \in \SyncToAsync(\lang(\semglobal(\GG)))$ with $w \interswap w'$.
For any $w' \in \SyncToAsync(\lang(\semglobal(\GG)))$, it trivially holds that $w'$ has a non-empty intersection set.
Because intersection sets are invariant under the indistinguishability relation $\interswap$, $w$ must also have a non-empty intersection set.
Since intersection sets are monotonically decreasing, if the intersection set of $w$ is non-empty, then for any $v \leq w$, the intersection set of $v$ is also non-empty.
Modus tollens of the chain of reasoning above tells us that in order to show a word is not a prefix in $\lang(\GG)$, it suffices to show that its intersection set is empty. 

Having established the sufficiency of properties \ref{cond:compl-proof-1} and \ref{cond:compl-proof-2} for our witness construction, we present the steps to construct $v_0$ from the negation of Send and Receive Validity respectively. We start by constructing a trace in $\CSM{\subsetcons{\GG}{\procA}}$ that satisfies \ref{cond:compl-proof-2}, and then show that $\CSM{B}$ also recognizes the trace, thereby satisfying \ref{cond:compl-proof-1}.
In both cases, let $\procA$ be the role and $s$ be the state for which the respective validity condition is violated.

\myparagraph{Send Validity (\cref{cond:send-state-validity-transition-origins}).}  
Let $s \xrightarrow{\snd{\procA}{\procB}{\val}} s' \in \delta_{\procA}$ be a transition such that
\[
\transAnnoFunc(s  \xrightarrow{\snd{\procA}{\procB}{\val}} s') \neq s \enspace.
\]
First, we find a trace $u$ of $\CSM{\subsetcons{\GG}{\procA}}$ that satisfies:
\begin{inparaenum}[(1)]
\item \label{cond:compl-send-validity-1}
role $\procA$ is in state~$s$ in the CSM configuration reached via $u$, and
\item \label{cond:compl-send-validity-2}
the run of $\semglobal(\GG)$ on $u$ visits a state in $s \setminus \transAnnoFunc(s  \xrightarrow{\snd{\procA}{\procB}{\val}} s')$.
\end{inparaenum}
We obtain such a witness $u$ from the $\SyncToAsync(\trace(-))$ of a run prefix of $\semglobal(\GG)$ that ends in some state in \mbox{$s \setminus \transAnnoFunc(s  \xrightarrow{\snd{\procA}{\procB}{\val}} s')$}.
Any prefix thus obtained satisfies (\ref{cond:compl-send-validity-1}) by definition of $\subsetcons{\GG}{\procA}$, and satisfies (\ref{cond:compl-send-validity-2}) by construction.
Due to the fact that send transitions are always enabled in a CSM, $u \cdot \snd{\procA}{\procB}{\val}$ must also be a trace of $\CSMl{\subsetcons{\GG}{\procA}}$, thus satisfying property \ref{cond:compl-proof-1} by a simulation argument.
We then argue that $u \cdot \snd{\procA}{\procB}{\val}$ satisfies property \ref{cond:compl-proof-2}, stating that $I(u \cdot \snd{\procA}{\procB}{\val})$ is empty: the negation of Send Validity gives that there exist no run extensions from our candidate state in $s \setminus \transAnnoFunc(s  \xrightarrow{\snd{\procA}{\procB}{\val}} s')$ with the immediate next action $\procA \xrightarrow{} \procB: \val$, and therefore there exists no maximal run in $\semglobal(\GG)$ consistent with $u \cdot \snd{\procA}{\procB}{\val}$.

\myparagraph{Receive Validity (\cref{cond:rcv-state-validity}).}
Let
$s \xrightarrow{\rcv{\procB_1}{\procA}{\val_1}} s_1$ and 
$s \xrightarrow{\rcv{\procB_2}{\procA}{\val_2}} s_2 \in \delta_\procA$ be two transitions,
and let $G_2 \in \transAnnoFunDest(s \xrightarrow{\rcv{\procB_2}{\procA}{\val_2}} s_2)$ 
such that 
\[
\procB_1 \neq \procB_2 
\text{ and }
\snd{\procB_1}{\procA}{\val_1} \in \semavail^{\procA}_{(G_2 \ldots)} \enspace.
\]
Constructing the witness $v_0$ pivots on finding a trace $u$ of $\CSMl{\subsetcons{\GG}{\procA}}$ 
such that both
$u \cdot \rcv{\procB_1}{\procA}{\val_1}$
and 
$u \cdot \rcv{\procB_2}{\procA}{\val_2}$
are traces of
$\CSMl{\subsetcons{\GG}{\procA}}$. 
Equivalently, we show there exists a reachable configuration of $\CSMl{\subsetcons{\GG}{\procA}}$ in which $\procA$ can receive either message from distinct senders $\procB_1$ and $\procB_2$.
Formally, the local state of $\procA$ has two outgoing states labeled with $\rcv{\procB_1}{\procA}{\val_1}$ and $\rcv{\procB_2}{\procA}{\val_2}$, and the channels $\channel{\procB_1}{\procA}$ and $\channel{\procB_2}{\procA}$ have $m_1$ and $m_2$ at their respective heads.
We construct such a $u$ by considering a run in $\semglobal(\GG)$ that contains two transitions labeled with $\procB_1 \xrightarrow{} \procA: m_1$ and $\procB_2 \xrightarrow{} \procA: m_2$. Such a run must exist due to the negation of Receive Validity. 
We start with the split trace
of this run, and argue that, from the definition of $\semavail(\hole)$ and the indistinguishability relation $\interswap$, we can perform iterative reorderings using $\interswap$ to bubble the send action $\snd{\procB_1}{\procA}{\val_1}$ to the position before the receive action $\rcv{\procB_2}{\procA}{\val_2}$.
Then, \ref{cond:compl-proof-1}
for $u \cdot \rcv{\procB_1}{\procA}{\val_1}$ holds by a simulation argument.
We then separately show that \ref{cond:compl-proof-2} holds for $\rcv{\procB_1}{\procA}{\val_1}$ using similar reasoning as the send case to complete the proof that $u \cdot \rcv{\procB_1}{\procA}{\val_1}$ suffices as a witness for $v_0$. 

It is worth noting that the construction of the witness prefix $v_0$ in the proof immediately yields an algorithm for computing counterexample traces \mbox{to~implementability}.

\begin{remark}[Mixed Choice is Not Needed to Implement Global Types]
\cref{thm:completeness} basically shows the necessity of Send Validity for implementability.
\cref{cor:no-mixed-choice} shows that Send Validity precludes states with both send and receive outgoing transitions. 
Together, this implies that an implementable global type can always be implemented without mixed choice.
Note that the syntactic restrictions on global types do not inherently prevent mixed choice states from arising in a role's subset construction, as evidenced by $\procC$ in the following type:
{ \small
$
\msgFromTo{\procA}{\procB}{l}.\,\msgFromTo{\procB}{\procC}{m}.\,0 \; + \;
\msgFromTo{\procA}{\procB}{r}.\,\msgFromTo{\procC}{\procB}{m}.\,0
.$
}
Our completeness result thus implies that this type is not implementable.
Most MST frameworks~\cite{DBLP:conf/sfm/CoppoDPY15,DBLP:conf/popl/HondaYC08,DBLP:conf/concur/MajumdarMSZ21} implicitly force \emph{no mixed choice} through syntactic restrictions on local types.
We are the first to prove that mixed choice states are indeed not necessary for completeness. This is interesting because mixed choice is known to be crucial for the expressive power of the synchronous $\pi$-calculus compared to its asynchronous variant~\cite{DBLP:journals/mscs/Palamidessi03}. 

\end{remark}

 \section{Complexity}
\label{sec:complexity}

In this section, we establish that checking implementability for global types is in PSPACE and show that generating an implementation can require exponential~time.

\begin{theorem}
  \label{thm:implementability-pspace-complete}
  The MST implementability problem is in \mbox{PSPACE}. 
\end{theorem}

\begin{proof}
The decision procedure enumerates for each role $\procA$ the subsets of $\projerasure{\GG}{\procA}$.
This can be done in polynomial space and exponential time.
For each $\procA$ and $s \subseteq Q_{\GG}$, it then
\begin{inparaenum}[(i)]
\item checks membership of~$s$ in $Q_{\procA}$ of $\subsetcons{\GG}{\procA}$, and \label{check-membership}
\item if $s \in Q_\procA$, checks whether all outgoing transitions of~$s$ in $\subsetcons{\GG}{\procA}$ satisfy Send and Receive Validity. \label{check-valid}
\end{inparaenum}
Check~(\ref{check-membership}) can be reduced to the intersection non-emptiness problem for nondeterministic finite state machines, which is in \mbox{PSPACE}~\cite{NODBLPthesiswehar}.
It is easy to see that check~(\ref{check-valid}) can be done in polynomial time.
In particular, the computation of available messages for Receive Validity only requires a single unfolding of every loop in $\GG$.

Note that the synthesis problem has the same complexity.
The subset construction to determinize $\projerasure{\GG}{\procA}$ can be done using a PSPACE transducer.
While the output can be of exponential size, it is written on an extra tape that is not counted towards memory usage.
However, this means we need to perform the validity checks as described above instead of using the computed deterministic state machines.

\end{proof}

This result and the fact that local languages are preserved by the subset projection (\cref{lm:languages-of-roles}) leads to the following observation.

\begin{corollary}
Let $\GG$ be an implementable global type.
Then, the subset projection $\CSMl{\subsetproj{\GG}{\procA}}$ is a local language preserving implementation for $\GG$, i.e., $\lang(\subsetproj{\GG}{\procA}) = \lang(\GG)\wproj_{\Alphabet_\procA}$ for every $\procA$, and can be computed in PSPACE.
\end{corollary}

While this establishes a PSPACE upper bound, we also show that generating an implementation may require exponential time, even for global types with directed choice.

\begin{lemma}
	\label{thm:implementability-exptime-necessary}
	Constructing an implementation for a global type with directed choice may require exponential time.
\end{lemma}
\begin{proof}
	We construct a family of directed global types $G_k$ that is implementable and requires a projection of exponential size for $\procB$.
	Hence, constructing the projections requires exponential time in the size of the input.
	\begingroup
	\addtolength{\jot}{.5em}
	{ \small
		\begin{align*}
			G_k & \is
			\mu t \, . \,
			+
			\begin{cases}
				\msgFromTo{\procA}{\procC}{s} \, . \,
				+
				\begin{cases}
					\msgFromTo{\procA}{\procB}{a} \, . \,
					t
					\\
					\msgFromTo{\procA}{\procB}{b} \, . \,
					t
				\end{cases}
				\\
				\msgFromTo{\procA}{\procC}{l} \, . \,
				+
				\begin{cases}
					\msgFromTo{\procA}{\procB}{a} \, . \,
					G_{a,k-1}
					\\
					\msgFromTo{\procA}{\procB}{b} \, . \,
					G_{b,k-1}
				\end{cases}
			\end{cases}
			\text{\hspace{-7ex}}
			&& \text{where}
			\\
			G_{a,i} & \is
			+
			\begin{cases}
				\msgFromTo{\procA}{\procB}{a} \, . \,
				G_{a,i-1}
				\\
				\msgFromTo{\procA}{\procB}{b} \, . \,
				G_{a,i-1}
			\end{cases}
			\; \text{ for } i > 0
			&
			G_{b,i} & \is
			+
			\begin{cases}
				\msgFromTo{\procA}{\procB}{a} \, . \,
				G_{b,i-1}
				\\
				\msgFromTo{\procA}{\procB}{b} \, . \,
				G_{b,i-1}
			\end{cases}
			\; \text{ for } i > 0
			\\
			G_{a,0} & \is
			\msgFromTo{\procA}{\procB}{d} \, . \,
			\msgFromTo{\procB}{\procA}{a} \, . \, 0
			&
			G_{b,0} & \is
			\msgFromTo{\procA}{\procB}{d} \, . \,
			\msgFromTo{\procB}{\procA}{b} \, . \, 0
		\end{align*}
	}
	\endgroup
	
	Intuitively, $\procA$ sends a sequence of letters from $\set{a, b}$ to $\procB$, followed by $d$ to indicate that the sequence is over.
	Then, $\procB$ needs to send the $k$th last letter back to $\procA$.
	Hence, $\procB$ needs to remember the last $k$ letters at all times, yielding at least $2^k$ different states for its projection.
	This is a variation of a language that is well-known to be recognisable only by a DFA with exponentially many states, compared to a minimal NFA:
	all words over $\set{a,b}$ where $a$ is the $k$th last letter.
	Note that our result does not explicitly require that the implementation is local language preserving but any implementation for $G_k$ will be by construction.
\end{proof}

 \section{Evaluation}
\label{sec:eval}

We consider the following three aspects in the evaluation of our approach:
\begin{inparaenum}[(E1)]
\item \label{eval:diff}difficulty of implementation
\item \label{eval:completeness}completeness, and
\item \label{eval:state-of-the-art}comparison to state of the art.
\end{inparaenum}

\begin{table}[t]

\begin{center}
\scalebox{0.8}{
\setlength{\tabcolsep}{4.5pt}
\begin{tabular}{c l c |
@{\hspace{2.5ex}} c @{\hspace{-1.5ex}} r c c c |
@{\hspace{2.5ex}} c @{\hspace{-1.5ex}} r}
 \toprule
 Source &
 Name &
 Impl. &
 \multicolumn{2}{c}{Subset Proj.} &
 Size &
 $|𝓟|$ &
 Size &
 \multicolumn{2}{c}{\cite{DBLP:conf/concur/MajumdarMSZ21}} \\
 & &  &
 \multicolumn{2}{c}{(complete)} & & &
 Proj's  &
 \multicolumn{2}{c}{(incomplete)} \\
 \midrule
 \multirow{3}{*}{\cite{10.1145/3291638}} &
 Instrument Contr.\ Prot.\ A &
 \yessymbol &
 \yessymbol & $0.4$\,ms &
 $22$ &
 $3$ &
 $61$ &
 \yessymbol & $0.2$\,ms
 \\
 &
 Instrument Contr.\ Prot.\ B &
 \yessymbol &
 \yessymbol & $0.3$\,ms &
 $17$ &
 $3$ &
 $47$ &
 \yessymbol & $0.1$\,ms
 \\
 &
 OAuth2 &
 \yessymbol &
 \yessymbol & $0.1$\,ms &
 $10$ &
 $3$ &
 $23$ &
 \yessymbol & $<\negthinspace0.1$\,ms
 \\
 \cmidrule{1-1}
 \multirow{1}{*}{\cite{DBLP:conf/ecoop/ScalasDHY17}}
 &
 Multi Party Game &
 \yessymbol &
 \yessymbol & $0.5$\,ms &
 $21$ &
 $3$ &
 $67$ &
 \yessymbol & $0.1$\,ms
 \\
 \cmidrule{1-1}
 \multirow{1}{*}{\cite{DBLP:conf/popl/HondaYC08}}
 &
 Streaming &
 \yessymbol &
 \yessymbol & $0.2$\,ms &
 $13$ &
 $4$ &
 $28$ &
 \yessymbol & $<\negthinspace0.1$\,ms
 \\
 \cmidrule{1-1}
\multirow{1}{*}{\cite{DBLP:journals/corr/abs-1203-0780}} &
 Non-Compatible Merge &
 \yessymbol &
 \yessymbol & $0.2$\,ms &
 $11$ &
 $3$ &
 $25$ &
 \yessymbol & $0.1$\,ms
\\
 \cmidrule{1-1}
\multirow{1}{*}{\cite{springhibernate}} &
 Spring-Hibernate &
 \yessymbol &
\yessymbol & $1.0$\,ms &
 $62$ &
 $6$ &
 $118$ &
 \yessymbol & $0.7$\,ms
\\
 \cmidrule{1-1}
 \multirow{4}{*}{\cite{DBLP:conf/concur/MajumdarMSZ21}} &
 Group Present &
 \yessymbol &
 \yessymbol & $0.6$\,ms &
 $51$ &
 $4$ &
 $85$ &
 \yessymbol & $0.6$\,ms
 \\
 &
 Late Learning &
 \yessymbol &
 \yessymbol & $0.3$\,ms &
 $17$ &
 $4$ &
 $34$ &
 \yessymbol & $0.2$\,ms
 \\
 &
 Load Balancer ($n=10$) &
 \yessymbol &
 \yessymbol & $3.9$\,ms &
 $36$ &
 $12$ &
 $106$ &
 \yessymbol & $2.4$\,ms
 \\
 &
 Logging ($n=10$) &
 \yessymbol &
 \yessymbol & $71.5$\,ms &
 $81$ &
 $13$ &
 $322$ &
 \yessymbol & $10.0$\,ms
 \\
 \cmidrule{1-1}
 \multirow{4}{*}{\cite{ecoop-draft}} &
 2 Buyer Protocol &
 \yessymbol &
 \yessymbol & $0.5$\,ms &
 $22$ &
 $3$ &
 $60$ &
 \yessymbol & $0.2$\,ms
 \\
 &
 2B-Prot.~Omit No &
 \yessymbol &
 \yessymbol & $0.4$\,ms &
 $19$ &
 $3$ &
 $56$ &
 (\negthinspace\nosymbol\negthinspace) & $0.1$\,ms
 \\
 &
 2B-Prot.~Subscription&
 \yessymbol &
 \yessymbol & $0.7$\,ms &
 $46$ &
 $3$ &
 $95$ &
 (\negthinspace\nosymbol\negthinspace) & $0.3$\,ms
 \\
 &
 2B-Prot.~Inner Recursion&
 \yessymbol &
 \yessymbol & $0.4$\,ms &
 $17$ &
 $3$ &
 $51$ &
 \yessymbol & $0.1$\,ms
 \\
 \cmidrule{1-1}
 \multirow{7}{*}{ New} &
 Odd-even (\cref{ex:odd-even}) &
 \yessymbol &
 \yessymbol & $0.5$\,ms &
 $32$ &
 $3$ &
 $70$ &
 (\negthinspace\nosymbol\negthinspace) & $0.2$\,ms
 \\
 &
 $\GG_r$ -- Receive Val.~Violated (\S\ref{sec:motivation}) &
 \nosymbol &
 \nosymbol & $0.1$\,ms &
 $12$ &
 $3$ &
 $\hole$ &
 (\negthinspace\nosymbol\negthinspace) & $<\negthinspace0.1$\,ms
 \\
 &
 $\GG'_r$ -- Receive Val.~Satisfied (\S\ref{sec:motivation}) &
 \yessymbol &
 \yessymbol & $0.2$\,ms &
 $16$ &
 $3$ &
 $35$ &
 \yessymbol & $0.1$\,ms
 \\
 &
 $\GG_s$ -- Send Val.~Violated (\S\ref{sec:motivation}) &
 \nosymbol &
 \nosymbol & $<\negthinspace0.1$\,ms &
 $8$ &
 $3$ &
 $\hole$ &
 (\negthinspace\nosymbol\negthinspace) & $<\negthinspace0.1$\,ms
 \\
 &
 $\GG'_s$ -- Send Val.~Satisfied (\S\ref{sec:motivation}) &
 \yessymbol &
 \yessymbol & $<\negthinspace0.1$\,ms &
 $7$ &
 $3$ &
 $17$ &
 \yessymbol & $<\negthinspace0.1$\,ms
 \\
 &
 $\GG_{\operatorname{fold}}$ (\S\ref{sec:discussion}) &
 \yessymbol &
 \yessymbol & $0.4$\,ms &
 $21$ &
 $3$ &
 $50$ &
 (\negthinspace\nosymbol\negthinspace) & $0.1$\,ms
 \\
 &
 $\GG_{\operatorname{unf}}$ (\S\ref{sec:discussion}) &
 \yessymbol &
 \yessymbol & $0.4$\,ms &
 $30$ &
 $3$ &
 $61$ &
 \yessymbol & $0.2$\,ms
 \\
\bottomrule
\end{tabular}
 }
\end{center}

\caption{Projecting Global Types. For every protocol, we report whether it is implementable \yessymbol\xspace or not \nosymbol, the time to compute our subset projection and the generalized projection by Majumdar et al.~\cite{DBLP:conf/concur/MajumdarMSZ21} as well as the outcome as \yessymbol\xspace for ``implementable'', \nosymbol\xspace for ``not implementable'' and (\negthinspace\nosymbol\negthinspace) for ``not known''.
We also give the size of the protocol (number of states and transitions), the number of roles, the combined size of all subset projections (number of states and transitions).}
\label{tb:proj}

\end{table}

For this, we implemented our subset projection in a prototype tool \cite{prototype,artifact}.
It takes a global type as input and computes the subset projection for each role. It was straightforward to implement the core functionality in approximately 700 lines of Python3 code closely following the formalization (E\ref{eval:diff}). 

We consider global types (and communication protocols) from seven different sources as well as all examples from this work (cf.\ 1st column of \cref{tb:proj}).
Our experiments were run on a computer with an Intel Core i7-1165G7 CPU and used at most 100MB of memory.
The results are summarized in \cref{tb:proj}.
The reported size is the number of states and transitions of the respective state machine, which allows not to account for multiple occurrences of the same subterm.
As expected, our tool can project every implementable protocol we have considered~(E\ref{eval:completeness}).

Regarding the comparison against the state of the art (E\ref{eval:state-of-the-art}), we directly compared our subset projection to the incomplete approach by Majumdar et al.~\cite{DBLP:conf/concur/MajumdarMSZ21}, and found that the run times are in the same order of magnitude in general (typically a few milliseconds).
However, the projection of \cite{DBLP:conf/concur/MajumdarMSZ21} fails to project four implementable protocols (including \cref{ex:odd-even}).
We discuss some of the other examples in more detail in the next section.
We further note that most of the run times reported by Scalas and Yoshida~\cite{DBLP:journals/pacmpl/ScalasY19} on their model checking based tool are around 1 second and are thus two to three orders of magnitude slower.

 \section{Discussion}
\label{sec:discussion}

\myparagraph{Success of Syntactic Projections Depends on Representation.}
Let us illustrate how unfolding recursion helps syntactic projection operators to succeed.
Consider this implementable global type, which is not syntactically projectable: 

\vspace{-0.7ex}
\begin{center}
{ \scriptsize $
  \GG_{\operatorname{fold}} \is
  + \;
  \begin{cases}
    \msgFromTo{\procA}{\procB}{\msgO}. \,
    \mu t_1. \,
    (
    \msgFromTo{\procA}{\procB}{\msgO}. \,
    \msgFromTo{\procB}{\procC}{\msgO}. \,
    t_1
    \; + \;
    \msgFromTo{\procA}{\procB}{\msgB}. \,
    \msgFromTo{\procB}{\procC}{\msgB}. \,
    0
    )
    \\
    \msgFromTo{\procA}{\procB}{\msgM}. \,
    \msgFromTo{\procB}{\procC}{\msgM}. \,
    \mu t_2. \,
    (
    \msgFromTo{\procA}{\procB}{\msgO}. \,
    \msgFromTo{\procB}{\procC}{\msgO}. \,
    t_2
    \; + \;
    \msgFromTo{\procA}{\procB}{\msgB}. \,
    \msgFromTo{\procB}{\procC}{\msgB}. \,
    0
    )
\end{cases}
 $ }.
\end{center}
\vspace{-0.7ex}
\noindent
Similar to projection by erasure, a syntactic projection erases events that a role is not involved in and immediately tries to \emph{merge} different branches.
The merge operator is a partial operator that checks sufficient conditions for implementability.
Here, the merge operator fails for $\procC$ because it cannot merge a recursion variable binder and a message reception.
Unfolding the global type preserves the represented protocol and resolves this~issue:

\vspace{-0.7ex}
\begin{center}
{ \scriptsize $
  \GG_{\operatorname{unf}} \is
 + \;
 \begin{cases}
    \msgFromTo{\procA}{\procB}{\msgO}. \,
    \begin{cases}
    \msgFromTo{\procA}{\procB}{\msgB}. \,
    \msgFromTo{\procB}{\procC}{\msgB}. \,
    0
    \\
    \msgFromTo{\procA}{\procB}{\msgO}. \,
    \msgFromTo{\procB}{\procC}{\msgO}. \,
    \mu t_1. \,
    (
    \msgFromTo{\procA}{\procB}{\msgO}. \,
    \msgFromTo{\procB}{\procC}{\msgO}. \,
    t_1
    \; + \;
    \msgFromTo{\procA}{\procB}{\msgB}. \,
    \msgFromTo{\procB}{\procC}{\msgB}. \,
    0
    )
    \end{cases}
    \\
    \msgFromTo{\procA}{\procB}{\msgM}. \,
    \msgFromTo{\procB}{\procC}{\msgM}. \,
    \mu t_2. \,
    (
    \msgFromTo{\procA}{\procB}{\msgO}. \,
    \msgFromTo{\procB}{\procC}{\msgO}. \,
    t_2
    \; + \;
    \msgFromTo{\procA}{\procB}{\msgB}. \,
    \msgFromTo{\procB}{\procC}{\msgB}. \,
    0
    )
 \end{cases}
 $ }.
 \end{center}
\vspace{-0.7ex}

\iftoggle{techrep}{
\noindent (We refer to \cref{fig:folded-unfolded} \,in \cref{app:discussion-visual-reps} for visual representations of both global types.) }{
\noindent (We refer to \cite{li2023complete} for visual representations of both global types.)
}
This global type can be projected with most syntactic projection operators and shows that the representation of the global type matters for syntactic projectability.
However, such unfolding tricks do not always work, e.g. for the odd-even protocol (\cref{ex:odd-even}).
We avoid this brittleness using automata and separating the synthesis from checking implementability.

\myparagraph{Entailed Properties from the Literature.} We defined implementability for a global type as the question of whether there exists a deadlock-free CSM that generates the same language as the global type.
Various other properties of implementations and protocols have been proposed in the literature.
Here, we give a brief overview and defer to \appendixRef{app:discussion-entailed} for a detailed analysis.
\emph{Progress} \cite{DBLP:conf/sfm/CoppoDPY15}, a common property, requires that every sent message is eventually received and every expected message will eventually be sent.
With deadlock freedom, our subset projection trivially satisfies progress for finite traces.
For infinite traces, as expected, fairness assumptions are required to enforce progress.
Similarly, our subset projection prevents \emph{unspecified receptions}~\cite{DBLP:journals/iandc/CeceF05} and \emph{orphan messages}~\cite{DBLP:conf/icalp/DenielouY13,DBLP:conf/concur/BocchiLY15}, respectively interpreted in our multiparty setting with sender-driven choice.
We also ensure that every local transition of each role is \emph{executable}~\cite{DBLP:journals/iandc/CeceF05}, i.e. it is taken in some run of the CSM.
Any implementation of a global type has the \emph{stable property}~\cite{DBLP:conf/cav/LangeY19}, i.e., one can always reach a configuration with empty channels from every reachable configuration.
While the properties above are naturally satisfied by our subset projection, the following ones can be checked directly on an implementable global type without explicitly constructing the implementation.
A global type is \emph{terminating}~\cite{DBLP:journals/pacmpl/ScalasY19} iff it does not contain recursion and \emph{never-terminating}~\cite{DBLP:journals/pacmpl/ScalasY19} iff it does not contain term $0$.

 \section{Related Work}
\label{sec:related}

MSTs were introduced by Honda et al. \cite{DBLP:conf/popl/HondaYC08} with a process algebra semantics, and
the connection to CSMs was established soon afterwards~\cite{DBLP:conf/esop/DenielouY12}. 

In this work, we present a complete projection procedure for global types with sender-driven choice.
The work by Castagna et al. \cite{DBLP:journals/corr/abs-1203-0780} is the only one to present a projection that aims for completeness.
Their semantic conditions, however, are not effectively computable and their notion of completeness is ``less demanding than the classical ones''\cite{DBLP:journals/corr/abs-1203-0780}.
They consider multiple implementations, generating different sets of traces, to be sound and complete with regard to a single global type \cite[Sec.\,5.3]{DBLP:journals/corr/abs-1203-0780}.
In addition, the algorithmic version of their conditions does not use global information as our message availability analysis does.

\renewcommand*{\thefootnote}{\arabic{footnote}}\setcounter{footnote}{0}

MST implementability relates to safe realizability of HMSCs, which is undecidable in general but decidable for certain classes~\cite{DBLP:journals/tcs/Lohrey03}.
Stutz~\cite{ecoop-draft} showed that implementability of global types that are always able to terminate is decidable.\footnote{This syntactic restriction is referred to as 0-reachability in \cite{ecoop-draft}.}
The EXPSPACE decision procedure is obtained via a reduction to safe realizability of globally-cooperative HMSCs, by proving that the HMSC encoding \cite{DBLP:journals/corr/abs-2209-10328} of any implementable global type is globally-cooperative and generalizing results for infinite executions.
Thus, our PSPACE result both generalizes and tightens the earlier decidability result obtained in~\cite{ecoop-draft}.
Stutz \cite{ecoop-draft} also investigates how HMSC techniques for safe realizability can be applied to the MST setting -- using the formal connection between MST implementability and safe realizability of HMSCs -- and establishes an undecidability result for a variant of MST implementability with a relaxed indistinguishability~relation.

Similar to the MST setting, there have been approaches in the HMSC literature that tie branching to a role making a choice.
We refer the reader to the work by Majumdar et al.~\cite{DBLP:conf/concur/MajumdarMSZ21} for a survey.

Standard MST frameworks project a global type to a set of \emph{local types} rather than a CSM. 
Local types are easily translated to FSMs~\cite[Def.11]{DBLP:conf/concur/MajumdarMSZ21}.
Our projection operator, though, can yield FSMs that cannot be expressed with the limited syntax of \mbox{local} types.
Consider this implementable global type:
$\msgFromTo{\procA}{\procB}{\msgO}.\,0 ~ + ~ \msgFromTo{\procA}{\procB}{\msgM}.\,\msgFromTo{\procA}{\procC}{\msgB}.\,0$\;.
The subset projection for $\procC$ has two final states connected by a transition labeled $\rcv{\procA}{\procC}{\msgB}$.
In the syntax of local types, $0$~is the only term indicating termination, which means that final states with outgoing transitions cannot be expressed.
In contrast to the syntactic restrictions for global types, which are key to effective verification, we consider local types unnecessarily restrictive.
Usually, local implementations are type-checked against their local types and subtyping gives some implementation freedom~\cite{DBLP:conf/ppdp/ChenDY14,DBLP:journals/lmcs/ChenDSY17,DBLP:conf/fossacs/LangeY17,DBLP:journals/tcs/BravettiCZ18}.
However, one can also view our subset projection as a local specification of the actual implementation. 
We conjecture that subtyping would then amount to a variation of alternating refinement~\cite{DBLP:conf/concur/AlurHKV98}.

CSMs are Turing-powerful \cite{DBLP:journals/jacm/BrandZ83} but decidable classes were obtained for different semantics:
restricted communication topology~\cite{DBLP:journals/acta/PengP92,DBLP:conf/tacas/TorreMP08},
half-duplex communication (only for two roles)~\cite{DBLP:journals/iandc/CeceF05},
input-bounded~\cite{DBLP:conf/concur/BolligFS20},
and unreliable channels~\cite{DBLP:conf/cav/AbdullaBJ98,DBLP:conf/lics/AbdullaAA16}.
Global types (as well choreography automata \cite{DBLP:conf/coordination/BarbaneraLT20}) can only express existentially $1$-bounded, $1$-synchronizable and half-duplex communication~\cite{DBLP:journals/corr/abs-2209-10328}.
Key to this result is that sending and receiving a message is specified atomically in a global type -- a feature Dagnino et al.~\cite{DBLP:conf/coordination/DagninoGD21} waived for their deconfined global types.
However, Dagnino et al.~\cite{DBLP:conf/coordination/DagninoGD21} use deconfined types to capture the behavior of a given system rather than projecting to obtain a system that generates specified behaviors.

This work relies on reliable communication as is standard for MST frameworks.
Work on fault-tolerant MST frameworks~\cite{DBLP:journals/pacmpl/VieringHEZ21,DBLP:conf/concur/BarwellSY022} attempts to relax this restriction. In the setting of reliable communication, both context-free~\cite{dblp:conf/icfp/thiemannv16,dblp:journals/toplas/keizerbp22} and parametric~\cite{dblp:journals/scp/charalambidesda16,DBLP:journals/corr/abs-1208-6483} versions of session types have been proposed to capture more expressive protocols and entire protocol families respectively.
Extending our approach to these generalizations is an interesting direction for future work.

\subsubsection*{Acknowledgements.}
This work is funded in part by the National Science Foundation under grant~1815633.
Felix Stutz was supported by the Deutsche Forschungsgemeinschaft project 389792660 TRR 248—CPEC.

\phantomsection\label{paper-last-page}
    \clearpage

\bibliographystyle{splncs04}

\clearpage
\appendix
\section{Additional Material for \cref{sec:prelim}}
\label{app:prelim}

\subsection{Additional Definitions}
Given a word $w = w_0 \ldots w_n$, we use $w[i]$ to denote the i-th symbol $w_i \in \Alphabet$, and $w[0..i]$ to denote the subword between and including $w_0$ and $w_i$, $w_0 \ldots w_i$. 

\subsection{Indistinguishability Relation \cite{DBLP:conf/concur/MajumdarMSZ21}}
\label{app:indist-rel}
We define a family of \emph{indistinguishability relations}
${\interswap_i} \subseteq \AlphAsync^* \times \AlphAsync^*$ for $i\geq 0$
as follows.
For all $w\in\Alphabet^*$, we have $w \interswap_0 w$.
For $i=1$, we define:
\vspace{-1ex}
\begin{enumerate}[label=(\arabic*)]
	\item
	If $\procA ≠ \procC$, then
	$
	w.\snd{\procA}{\procB}{\val}.\snd{\procC}{\procD}{\val'}.u
	\; \interswap_{1} \;
	w.\snd{\procC}{\procD}{\val'}.\snd{\procA}{\procB}{\val}.u
	$.
	
	\item
	If $\procB ≠ \procD$, then
	$
	w.\rcv{\procA}{\procB}{\val}.\rcv{\procC}{\procD}{\val'}.u
	\; \interswap_{1} \;
	w.\rcv{\procC}{\procD}{\val'}.\rcv{\procA}{\procB}{\val}.u
	$.
	
	\item
	If $\procA ≠ \procD \land (\procA ≠ \procC ∨ \procB ≠ \procD)
	$, then
	$
	w.\snd{\procA}{\procB}{\val}.\rcv{\procC}{\procD}{\val'}.u
	\; \interswap_{1} \;
	w.\rcv{\procC}{\procD}{\val'}.\snd{\procA}{\procB}{\val}.u
	$.
	\item
	If $\card{w \wproj_{\snd{\procA}{\procB}{\_}}} >
	\card{w \wproj_{\rcv{\procA}{\procB}{\_}}}$,
	then
	$
	w.\snd{\procA}{\procB}{\val}.\rcv{\procA}{\procB}{\val'}.u
	\; \interswap_{1} \;
	w.\rcv{\procA}{\procB}{\val'}.\snd{\procA}{\procB}{\val}.u
	$.
\end{enumerate}
Let $w, w', w''$ be sequences of events s.t.~$w \interswap_1 w'$ and $w' \interswap_i w''$ for some~$i$.
Then, $w \interswap_{i+1} w''$.
We define $w \interswap u$ if $w \interswap_n u$ for some $n$.

It is easy to see that $\interswap$ is an equivalence relation.
Define $u \preceq_\interswap v$ if there is $w\in\Sigma^*$ such that $u.w \interswap v$.
Observe that $u \interswap v$ iff
$u \preceq_\interswap v$ and $v \preceq_\interswap u$.

For infinite words $u, v\in\Sigma^\omega$, we define $u \preceq_\interswap^\omega v$ 
if for each finite prefix $u'$ of $u$, there is a finite prefix $v'$ of $v$ such that
$u' \preceq_\interswap v'$.
Define $u \interswap v$ iff $u \preceq_\interswap^\omega v$ and $v\preceq_\interswap^\omega u$.

We lift the equivalence relation $\interswap$ on $\Sigma^\infty$ to languages:
\[
  \interswaplang(L) = \left\{ w' \mid \bigvee
    \begin{array}{l}
    w' \in \Alphabet^* \land ∃ w ∈ \Alphabet^*. \; w \in L \text{ and } w' \interswap w \\
    w' ∈ \Alphabet^ω \land \exists w \in \Alphabet^\omega. \; w \in L \text{ and } w' \preceq_\interswap^\omega w
  \end{array} \right\}
\]
For the infinite case, we take the downward closure w.r.t.~$\preceq_\interswap^\omega$.
Notice that the closure operator is asymmetric.
Consider the protocol $(\snd{\procA}{\procB}{\val}.\rcv{\procA}{\procB}{\val})^ω$.
Since we do not make any fairness assumption on scheduling, we need to include in the closure the execution where only the sender is scheduled, i.e., $(\snd{\procA}{\procB}{\val})^ω \preceq_\interswap^\omega (\snd{\procA}{\procB}{\val}.\rcv{\procA}{\procB}{\val})^ω$.

\subsection{State Machine} 
A \emph{state machine} $A$ is a $5$-tuple $(Q, \Delta, \delta, q_{0}, F)$
where $Q$ is a finite set of states,
$\Delta$ is a finite alphabet,
$\delta \subseteq Q \times (\Alphabet \union \set{\emptystring}) \times Q$ is a transition relation,
$q_0 \in Q$ is the initial state, and
$F \subseteq Q$ is the set of final states.
As is standard, we write $q \xrightarrow{x} q'$ for $(q, x, q')\in\delta$ and $q \xrightarrow{w} \mathrel{\vphantom{\to}}^* q'$ for its reflexive and transitive hull for $w \in \Delta^*$.
We define the runs and traces in the standard way.
A run is maximal if it is infinite or if it ends at a final state.
The \emph{language} $\lang(A)$ is the set of (finite or infinite) maximal traces.

     \section{Additional Material for \cref{sec:constructing-implementations}}
\label{app:constructing-implementations}

\subsection{Proofs for \cref{sec:synthesizing-implementations}}
\label{proof:synthesizing-implementations}
\label{proof:projection-preserves-per-process-runs}
\constructionProperties*
\begin{proof}
All claims are rather straightforward from the definitions and constructions and the proofs exploit the connection to the projection by erasure.
We still spell them out to familiarize the reader with these.

We prove the first claim first.
By construction, for every run $\run$ in $\semglobalsync(\GG)$, there exists a run $\run'$ in the projection by erasure $\projerasure{\GG}{\procC}$.
Let $\run$ be the run for trace $w$ in $\semglobal(\GG)$. 
Then, $\run$ is also a run in $\projerasure{\GG}{\procC}$ with trace $w \wproj_{\Alphabet_\procA}$. 
Since $\projerasure{\GG}{\procC}$ might be non-deterministic, we apply the subset construction from~\cref{def:subset-construction}.
For the reachable states, this is equivalent to the definition by Sipser~\cite[Thm.~1.39]{DBLP:books/daglib/0086373}.Thus, the constructed deterministic finite state machine can mimic any run (which is initial by definition) in $\projerasure{\GG}{\procC}$:
for every run~$\run'$ in $\projerasure{\GG}{\procC}$ with trace $w'$, there is a run $\run''$ in $\subsetcons{\GG} {\procC}$ with trace~$w'$.

For the second claim, we consider a trace $u$ of $\subsetcons{\GG}{\procC}$.
Because of the subset construction, it holds that for every run $\run'$ in $\subsetcons{\GG} {\procC}$ with trace $w'$,
there is a run $\run''$ in the projection by erasure $\projerasure{\GG}{\procC}$ with trace $w'$.
By definition of the projection by erasure,
a run $\run'''$ in $\semglobalsync(\GG)$ exists
with the same sequence of syntactic subterms as $\run''$ and
$\SyncToAsync(\trace(\run''')) \wproj_{\Alphabet_\procA} = w'$. 

From this, it easily follows that
$
\lang(\subsetcons{\GG}{\procC})
=
\lang(\projerasure{\GG}{\procC})
$
and, thus,
$
\lang(\subsetcons{\GG}{\procC})
=
\lang(\GG) \wproj_{\Alphabet_\procC}
$.
\proofend
\end{proof}

\constructionPreservesBehaviors*
\begin{proof}
	Given that $\CSMl{\subsetcons{\GG}{\procA}}$ is deterministic, to prove language inclusion it suffices to prove the inclusion of the respective prefix sets: 
	\[
		\text{pref}(\lang(\GG)) \subseteq \text{pref}(\lang\CSMl{\subsetcons{\GG}{\procA}})
	\]
	We prove this via structural induction on $w$. 
	The base case, $w = \emptystring$, is trivial.
	For the inductive step, let 
	$wx \in \text{pref}(\lang(\GG))$. 
	From the induction hypothesis, 
	$w \in \text{pref}(\lang\CSMl{\subsetcons{\GG}{\procA}})$. 
	It suffices to show that the transition labeled with $x$ is enabled for the active role in $x$. 
	Let $(\vec{s},\xi)$ denote the $\CSMl{\subsetcons{\GG}{\procA}}$ configuration reached on $w$. 
	In the case that $x \in \Alphabet_!$, let $x = \snd{\procA}{\procB}{\val}$. 
	The existence of an outgoing transition $\xrightarrow{\snd{\procA}{\procB}{\val}}$ from $\vec{s}_\procA$ follows from the fact that $\lang(\subsetcons{\GG}{\procA}) = \lang(\GG) \wproj_{\Alphabet_\procA}$ (\cref{lm:languages-of-roles}). 
	The fact that $wx \in \text{pref}(\lang\CSMl{\subsetcons{\GG}{\procA}})$ follows immediately from this and the fact that send transitions in a CSM are always enabled.
	In the case that $x \in \Alphabet_?$, let $x = \rcv{\procA}{\procB}{\val}$.
	We obtain an outgoing transition $\xrightarrow{\rcv{\procA}{\procB}{\val}}$ from $\vec{s}_\procA$ analogously. 
	We additionally need to show that $\xi(\procB,\procA)$ contains $m$ at the head. 
	This follows from \cite[Lemma~20]{DBLP:conf/concur/MajumdarMSZ21} and the induction hypothesis. 
	This concludes our proof of prefix set inclusion. 
	
	Let $w$ be a word in $\lang(\GG)$. 
	Let $(\vec{s},\xi)$ denote the $\CSMl{\subsetcons{\GG}{\procA}}$ configuration reached on $w$. 
	In the case that $w$ is finite, all states in $\vec{s}$ are final from \cref{lm:languages-of-roles} and all channels in $\xi$ are empty from the fact that all send events in $w$ contain matching receive events. 
In the case that $w$ is infinite, we show that $w$ has an infinite run in $\CSMl{\subsetcons{\GG}{\procA}}$ using König's Lemma. 
	We construct an infinite graph $\mathcal{G}_w(V, E)$ with 
	$V \is \{v_{\run} \mid \trace(\run) \leq w\}$ and 
	$E \is \{(v_{\run_1}, v_{\run_2}) \mid \exists~x \in \AlphAsync.~\trace(\run_2) = \trace(\run_1)\cdot x\}$. 
	Because $\CSMl{\subsetcons{\GG}{\procA}}$ is deterministic, $\mathcal{G}_w$ is a tree rooted at $v_\emptystring$, the vertex corresponding to the empty run.
	By König's Lemma, every infinite tree contains either a vertex of infinite degree or a ray. 
	Because $\CSMl{\subsetcons{\GG}{\procA}}$ consists of a finite number of communicating state machines, the last configuration of any run has a finite number of next configurations, and $\mathcal{G}_w$ is finitely branching. 
	Therefore, there must exist a ray in $\mathcal{G}_w$ representing an infinite run for $w$, and thus 
	$w \in \lang(\CSMl{\subsetcons{\GG}{\procA}})$. 
\proofend
\end{proof}

     \section{Additional Material for \cref{sec:soundness}}
\label{app:soundness}

The definition $\semavaildef{\blockedset}{T}{G}$ computes the set of available messages on the syntax of global types and follows the one by Majumdar et al.~\cite[Sec.\ 2.2]{DBLP:conf/concur/MajumdarMSZ21}.
For their proofs, they introduce another version, which computes these sets on the semantics of the global type using a concept called blocked languages.
In Lemma 37, they prove the syntactic version always yields a superset of the semantic version.
The proof easily generalizes to equality.
Therefore, we use $\semavaildef{\blockedset}{T}{\hole}$ in our proofs and refer to their work for details.

\begin{corollary}[Intersection sets are invariant under $\interswap$]
	\label{cor:I-set-indist-invar}
	Let $\GG$ be a global type.
	Let $w, w' \in \AlphAsync^*$ and $w \interswap w'$. Then, $I(w) = I(w')$. 
\end{corollary}
\begin{proof}
	It follows immediately from $w \interswap w'$ that 
	\[
		\forall \procA \in \Procs.~w \wproj_{\Alphabet_\procA} = w' \wproj_{\Alphabet_\procA}
	\]
	By the definition of $I$, 
	\[
		\forall \run.~ \run \in I(w) 
		\Leftrightarrow 
		\forall \procA\in \Procs.~
		w \wproj_{\Alphabet_\procA} \leq \SyncToAsync(\trace(\run)) \wproj_{\Alphabet_\procA}
	\]
	Let $\run$ be a run in $\semglobal(\GG)$. Then, 
	\begin{align*}
		\run \in I(w)  
		&\Leftrightarrow
		\forall \procA\in \Procs.~
		w \wproj_{\Alphabet_\procA} \leq \SyncToAsync(\trace(\run)) \wproj_{\Alphabet_\procA} \\ 
		&\Leftrightarrow 
		 \forall \procA\in \Procs.~
		 w' \wproj_{\Alphabet_\procA} \leq \SyncToAsync(\trace(\run)) \wproj_{\Alphabet_\procA} \\ 
		 &\Leftrightarrow 
		 \run \in I(w')
	\end{align*}
\proofend
\end{proof}

\begin{proposition}[Structural properties of $\semavail^{\blockedset}_{(G \ldots)}$]
	Let $\blockedset \subseteq \Procs$ and let $G$ be a syntactic subterm of $\GG$. 
	Let $\snd{\procB}{\procA}{\val} \in \semavail^{\blockedset}_{(G \ldots)}$. 
	Then, it holds that: 
\begin{enumerate}[(1)]
\item $\semavail^{\blockedset}_{(G \ldots)}$ does not contain any events whose active role is blocked: \\ $\forall~\procA\in \blockedset.~\semavail^{\blockedset}_{(G \ldots)} \inters \Alphabet_{\procA} = \emptyset$
\item There exists a run suffix $\beta$ such that:
	\begin{enumerate}[i.]
		\item $~G \cdot \beta \text{ is the suffix of a maximal run in } \semglobal(\GG)$,
		\item $\xrightarrow{\procB \xrightarrow{} \procA: \val} \text{ occurs in } \beta$, and 
\item $\blockedset \text{ monotonically increases during the computation of } \semavail^{\blockedset}_{(G \ldots)}$
\end{enumerate}
\end{enumerate}
\end{proposition}
\begin{proof}
	Immediate from the definition of available messages in~\cite[Sec.\,2.2]{DBLP:conf/concur/MajumdarMSZ21}.
\end{proof}

\begin{proposition}[Correspondence between unique splittings and local states]
	\label{prop:correspondence-unique-splittings-local-states}
	Let $\GG$ be a global type, and $\CSMl{\subsetcons{\GG}{\procA}}$ be the subset construction for each role. Let $w$ be a trace of $\CSMl{\subsetcons{\GG}{\procA}}$, and $(\vec{s},\xi)$ be the CSM configuration reached on $w$. Let $\procA$ be a role.
	Then, it holds that:
	\[
		\Union_{\run \in \globcomplocal{\GG}{\procA}{w}}\{ G \mid \alpha \cdot G \xrightarrow{l} G' \cdot \beta \text{ is the unique splitting of }\run \text{ matching } w\}
		\subseteq 
		\vec{s}_\procA
	\]
\end{proposition}
\begin{proof}
	Let $\run$ be a run in $\globcomplocal{\GG}{\procA}{w}$, and let $\alpha \cdot G \xrightarrow{l} G' \cdot \beta$ be its unique splitting for $\procA$ matching $w$. 
	It follows from the definition of unique splitting that 
	$\SyncToAsync(\trace(\alpha \cdot G')) \wproj_{\Alphabet_\procA} = w \wproj_{\Alphabet_\procA}$.
	From the subset construction, there exists a run $s_1, \ldots, s_n$ in $\subsetcons{\GG}{\procA}$ such that $s_1 = s_{0,\procA}$ and $s_1 \xrightarrow{w \wproj_{\Alphabet_\procA}} \mathrel{\vphantom{\to}^*} s_n$.
	By the definition of $Q_\procA$, we know that $s_n$ contains all global syntactic subterms in $\GG$ that are reachable via $q_{0,\GG} \xrightarrow{w \wproj_{\Alphabet_\procA}} \xrightarrow{\emptystring}\mathrel{\vphantom{\to}^*}$ in $\semglobal(\GG)_\downarrow$, of which $G$ is one. 
	Hence, $G \in s_n$. 
	By assumption, $\subsetcons{\GG}{\procA}$ reached state $\vec{s}_\procA$ on $w \wproj_{\Alphabet_\procA}$. 
	Because subset constructions are deterministic, it follows that $s_n = \vec{s}_\procA$. 
	We conclude that $G \in \vec{s}_\procA$. 
\proofend
\end{proof}

\begin{proposition}[No send transitions from final states in subset projection]
	\label{prop:final-state-outgoing-rcv}
	Let $\GG$ be a global type, and $\subsetproj{\GG}{\procA}$ be the subset projection for $\procA$. 
	Let $s \in F_\procA$, and $s \xrightarrow{x} s' \in \delta_\procA$. 
	Then, $x \in \Alphabet_{\procA,?}$. 
\end{proposition}
\begin{proof}
	Assume by contradiction that $x \in \Alphabet_{\procA,!}$.  
	We instantiate Send Validity with $s \xrightarrow{x} s'$ to obtain:
	\[
		x \in \Alphabet_{\procA,!} \implies 
		\transAnnoFunc(s  \xrightarrow{x} s') = s
	\]
	By the definition of $\transAnnoFunc(\hole)$, for every syntactic subterm $G \in \transAnnoFunc(s  \xrightarrow{x} s')$:
	\[
		\exists G' \in s'.~G \xrightarrow{x} \mathrel{\vphantom{\to}^*} G' \in \delta \proj_{\Alphabet_\procA}
	\]
	Because $s$ is a final state in $\subsetproj{\GG}{\procA}$, by definition it must contain a syntactic subterm that is a final state in $\semglobal(\GG) \proj_\procA$. 
	Because $\semglobal(\GG) \proj_\procA$ and $\semglobal(\GG)$ share the same set of final states \ref{def:projection-by-erasure}, $s$ must contain a syntactic subterm that is a final state in $\semglobal(\GG)$. 
	Let $G_0$ denote this final state. 
	By the structure of $\semglobal(\GG)$, there exists no outgoing transition from $G_0$. 
	Therefore, $G'$ does not exist. 
	We reach a contradiction. 
\proofend
\end{proof}

\begin{definition}[\gtcomplete{\GG} words of $\CSM{A}$]
	Let $\GG$ be a global type, $\CSM{A}$ be a CSM, and $w$ be a trace of $\CSM{A}$.
	We say $w$ is \emph{\gtcomplete{\GG}} if for all roles $\procA$ and for all runs $\run \in \Inters_{\procA \in \Procs} \globcomplocal{\GG}{\procA}{w}$, 
	\[
	w \wproj_{\Alphabet_\procA} =
	\bigl(
	\SyncToAsync(\trace(\run))
	\bigr)
	\wproj_{\Alphabet_\procA}.
	\]
\end{definition}

\begin{definition}[Channel-compliant~\cite{DBLP:conf/concur/MajumdarMSZ21}]
A word $w\in \Sigma^\infty$ is \channelcompliant if 
for every prefix $w' \leq w$ and every $\procA,\procB \in \Procs$,
$\MsgVals(w'\wproj_{\rcv{\procA}{\procB}{\_}}) \leq
\MsgVals(w'\wproj_{\snd{\procA}{\procB}{\_}})$.
\end{definition}

\rcvIntersectionSetEquality*

\begin{proof}
Let $x$ = $\rcv{\procB}{\procA}{\val}$.
Because $wx$ is a trace of $\CSMl{\subsetproj{\GG}{\procA}}$, there exists a run
$(\vec{s}_0, \xi_0) \xrightarrow{w} \mathrel{\vphantom{\to}^*} (\vec{s}, \xi) \xrightarrow{x} (\pvec{s}',\xi')$
such that
$m$ is at the head of $\xi(\procA,\procB)$.

We assume that $I(w)$ is non-empty; if $I(w)$ is empty then $I(wx)$ is trivially empty.
To show $I(w) = I(wx)$, it suffices to show the following claim.

\noindent
\textit{Claim 1 : }
It holds that
$
\globcomplocal{\GG}{\procA}{w} \inters \globcomplocal{\GG}{\procB}{w}
=
\globcomplocal{\GG}{\procA}{wx} \inters \globcomplocal{\GG}{\procB}{wx}
$.\medskip

We first show Claim 1's sufficiency for $I(w) = I(wx)$: 
By definition,
$I(w) = \Inters_{\procC \in \Procs} \globcomplocal{\GG}{\procC}{w} \subseteq
\globcomplocal{\GG}{\procB}{w} \inters \globcomplocal{\GG}{\procA}{w}$.
With \textit{Claim 1}, it holds that
$
I(w)
\subseteq
\globcomplocal{\GG}{\procB}{wx} \inters \globcomplocal{\GG}{\procA}{wx}
$.
From this, it follows that
$I(w) \subseteq \globcomplocal{\GG}{\procA}{wx}$ (H1).
Since $\procA$ is the active role for the receive event $x$,
i.e. $x \in \Alphabet_\procA$,
it holds for any $\procC \neq \procA$,
that
$(wx) \wproj_{\Alphabet_\procC} = w \wproj_{\Alphabet_\procC}$
and 
$\globcomplocal{\GG}{\procC}{w} = \globcomplocal{\GG}{\procC}{wx}$ (H2).
Again, by definition of $\globcomplocal{\GG}{\hole}{\hole}$, it holds that
$
\globcomplocal{\GG}{\procA}{wx}
\subseteq
\globcomplocal{\GG}{\procA}{w}
$
(H3).

\noindent We apply the observations to $I(wx)$:
\begin{align*}
	I(wx)
	& \overset{\phantom{\text{(H2)}}}{=}
	\Inters_{\procC \in \Procs} \globcomplocal{\GG}{\procC}{wx}
	\\
	& \overset{\text{(H2)}}{=}
	\globcomplocal{\GG}{\procA}{wx}
	\;
	\inters
	\Inters_{\procC \in \Procs \land \procC \neq \procA} \globcomplocal{\GG}{\procC}{w}
	\\
	& \overset{\text{(H3)}}{=}
	\globcomplocal{\GG}{\procA}{wx}
	\; \inters \;
	\globcomplocal{\GG}{\procA}{w}
	\; \inters \;
	\Inters_{\procC \in \Procs \land \procC \neq \procA} \globcomplocal{\GG}{\procC}{w} \\
	& \overset{\phantom{\text{(H3)}}}{=}
	\globcomplocal{\GG}{\procA}{wx}
	\; \inters \;
	I(w)
	\\
	& \overset{\text{(H1)}}{=}
	I(w) \enspace.
\end{align*}

\noindent
\textit{Proof of Claim 1: }
We instantiate (H2) for role $\procB$, which yields
$
\globcomplocal{\GG}{\procB}{wx}
=
\globcomplocal{\GG}{\procB}{w}
$.
The proof of Claim 1 therefore amounts to showing:
\[
	\globcomplocal{\GG}{\procA}{w} \inters \globcomplocal{\GG}{\procB}{w} = 										\globcomplocal{\GG}{\procA}{wx} \inters \globcomplocal{\GG}{\procB}{w} \enspace.
\]
The right direction,
i.e., $\globcomplocal{\GG}{\procA}{wx} \subseteq \globcomplocal{\GG}{\procA}{w}$,
follows from (H3).
For the left direction,
i.e., $\globcomplocal{\GG}{\procA}{w} \inters \globcomplocal{\GG}{\procB}{w} \subseteq \globcomplocal{\GG}{\procA}{wx}$,
assume by contradiction that there exists a run~$\run_0$ such that
\[
	\run_0 \in \globcomplocal{\GG}{\procA}{w}
	\; \land \;
	\run_0 \in \globcomplocal{\GG}{\procB}{w}
	\; \land \;
	\run_0 \notin \globcomplocal{\GG}{\procA}{wx}
	\enspace. \]

Let $\run'$ be a run in $\globcomplocal{\GG}{\procA}{w} \setminus \globcomplocal{\GG}{\procA}{wx}$.
Let $\alpha' \cdot G'_{\mathit{pre}} \xrightarrow{l'} G'_{\mathit{post}} \cdot \beta'$ be the unique splitting of $\run'$ for $\procA$ matching $w$.

Let $\run'_\procA$ denote the largest consistent prefix of $\run'$ for $\procA$; it is clear that 
$\run'_\procA = \alpha' \cdot G'_{pre}$. Formally,
\[
\run'_\procA = max\{\run~|~\run \leq \run' ~\land~ \bigl(
\SyncToAsync(\trace(\run))
\bigr)
\wproj_{\Alphabet_\procA} \preforder  w\wproj_{\Alphabet_\procA}
\}\enspace.
\]
Let $\run'_\procB$ be defined analogously.

We claim that $\procB$ is ahead of $\procA$ in $\run'$, i.e. $\run'_\procA < \run'_\procB$.
Intuitively, this claim follows from the half-duplex property of CSMs and the fact that $\procB$ is the sender. 
Formally, Lemma~19 in \cite{DBLP:conf/concur/MajumdarMSZ21} implies $\xi(\procB,\procA) = u$ where
$\mathcal{V}(w \wproj_{\snd{\procB}{\procA}{\_}}) = \MsgVals(w \wproj_{\rcv{\procB}{\procA}{\_}}).u$.
Because $\xi(\procB, \procA)$ contains at least $m$ by assumption,
$|\MsgVals(w\wproj_{\snd{\procB}{\procA}{\_}})| > |\MsgVals(w \wproj_{\rcv{\procB}{\procA}{\_}})|$.
Because 
$\MsgVals(w \wproj_{\rcv{\procB}{\procA}{\_}}) < \MsgVals(w\wproj_{\snd{\procB}{\procA}{\_}})$ 
and traces of CSMs are channel-compliant (Lemma~19, \cite{DBLP:conf/concur/MajumdarMSZ21}), it holds that
$\run'_\procB$ contains all 
$|\MsgVals(w \wproj_{\rcv{\procB}{\procA}{\_}})|$
transition labels of the form  
$\procB \xrightarrow{} \procA: \hole$ 
that are contained in $\run'_\procB$, 
plus at least one more of the form 
$\procB \xrightarrow{} \procA: m$. 
Because both $\run'_\procA$ and $\run'_\procB$ are prefixes of $\run'$, 
it must be the case that 
$\run'_\procA < \run'_\procB$. 
This concludes the proof of the above claim.

By assumption, $\run' \notin \globcomplocal{\GG}{\procA}{wx}$ and therefore $l' \neq \procB \xrightarrow{} \procA:m$.
By the definition of unique splittings, $\procA$ must be the active role in $l'$; by \cref{cor:no-mixed-choice}, $\procA$ must be the receiving role in $l'$. 
In other words, $l'$ must be of the form $\procC \xrightarrow{} \procA:\val'$, where either $\procC \neq \procB$ or $\val' \neq \val$.

\noindent
\textit{Case:} $\procC = \procB$ and $\val' \neq \val$.

We discharge this case by showing a contradiction to the assumption that $m$ is at the head of the channel between $\procB$ and $\procA$. 

Because $\alpha' \cdot G'_{pre} \leq \run'_\procA$ and $\run'_\procA < \run'_\procB$ from the claim above, it must be the case that
$\alpha' \cdot G'_{pre} \xrightarrow{l'} G'_{post} \leq  \run'_\procB$
and
 $\snd{\procB}{\procA}{\val'}$ is in $w \wproj_{\Alphabet_{\procB}}$.
From Lemma 19 \cite{DBLP:conf/concur/MajumdarMSZ21}, it follows that
$\mathcal{V}(w \wproj_{\snd{\procB}{\procA}{\_}}) = \MsgVals(w \wproj_{\rcv{\procB}{\procA}{\_}}).\val'.u'$ and
$\xi(\procB,\procA)  =  \val'.u'$, i.e. $m'$ is at the head of the channel between $\procB$ and $\procA$. 
This contradicts the assumption that $m$ is at the head of $\xi(\procA,\procB)$.

\noindent
\textit{Case:} $\procC \neq \procB$.

We discharge this case by showing a contradiction to Receive Validity.
We instantiate Receive Validity with
$\vec{s}_\procA \xrightarrow{x} \pvec{s}'_\procA$
to obtain
\[
	\forall~\vec{s}_\procA \xrightarrow{\rcv{\procB_2}{\procA}{\val_2}} s_2 \in \delta_\procA.~
	\procB \neq \procB_2
	\; \implies \;
	\forall~G_2 \in \transAnnoFunDest(\vec{s}_\procA \xrightarrow{\rcv{\procB_2}{\procA}{\val_2}} s_2). \;
	\snd{\procB}{\procA}{\val} \notin \semavail^{\procA}_{(G_2 \ldots)}\enspace.
\]
We prove the negation, stated as follows: 
\[
	\procB \neq \procC
	\land
	\exists~s_2 \in Q_\procA, G_2 \in \transAnnoFunDest(\vec{s}_\procA \xrightarrow{\rcv{\procC}{\procA}{\val'}} s_2). \;
	\snd{\procB}{\procA}{\val} \in \semavail^{\procA}_{(G_2 \ldots)}\enspace.
\]
The left conjunct follows immediately.
From the existence of $\run'$ and Lemma \ref{lm:languages-of-roles}, there exists an $s_2$ such that 
$\vec{s}_\procA \xrightarrow{\rcv{\procC}{\procA}{\val'}} s_2 \in \delta_\procA$. 
The fact that 
$G'_{post} \in \transAnnoFunDest(\vec{s}_\procA \xrightarrow{\rcv{\procC}{\procA}{\val'}} s_2)$ is trivial from the unique splitting of $\rho'$ for $\procA$ matching $w$:
\[
\run' = \alpha' \cdot G'_{\mathit{pre}} \xrightarrow{\procC \xrightarrow{} \procA:\val'} G'_{\mathit{post}} \cdot \beta'\enspace.
\]

Therefore, all that remains is to show that 
$\snd{\procB}{\procA}{\val} \in \semavail^{\procA}_{(G'_{\mathit{post}} \ldots)}$.
Because 
$\run' \in \globcomplocal{\GG}{\procB}{w}$ 
and 
$\alpha' \cdot G'_{\mathit{pre}} \xrightarrow{l'} G'_{post} \leq \run'_\procB$,
where $\procB$ is not the active role in $l'$, 
there must exist a transition labeled $\procB \xrightarrow{} \procA: m$ 
that occurs in the suffix $G'_{\mathit{post}} \cdot \beta'$ of $\run'$.
Let 
$G_0 \xrightarrow{\procB \xrightarrow{} \procA: m} G_0'$
be the earliest occurrence of such a transition in the suffix, then:
\[
	\run'_\procB = \alpha' \cdot G'_{\mathit{pre}} \xrightarrow{l'} G'_{\mathit{post}} \ldots G_0 \xrightarrow{\procB \xrightarrow{} \procA: m} G_0'\dots\enspace.
\]
Note that $G_0$ must be a syntactic subterm of $G'_{\mathit{post}}$.
In order for 
$\snd{\procB}{\procA}{\val} \in \semavail^{\procA}_{(G'_{\mathit{post}} \ldots)}$. 
to hold, it suffices to show that 
$\procB \notin \blockedset$ 
in the recursive call to 
$\semavail^{\blockedset}_{(G_0 \dots)}$. 

We argue this from the definition of $\semavail$ and the fact that $\run'_\procA = \alpha' \cdot G'_{pre}$.
Suppose for the sake of contradiction that 
$\procB \in \blockedset$.
Because $\semavail$ only adds receivers of already blocked senders to $\blockedset$ and $\semavail^{\procA}_{(G'_{\mathit{post}} \ldots)}$ starts with $\blockedset=\{\procA\}$, there must exist a chain of message exchanges $\procD_{i+1} \xrightarrow{} \procD_i: \val_{i}$ in $G'_{\mathit{post}}$ with $1 \leq i < n$, $\procA=\procD_{n}$, and $\procB=\procD_1$. That is, $G'_{\mathit{post}} \cdot \beta'$ must be of the form
\[
  G'_{\mathit{post}} \dots G_{n-1} \xrightarrow{\procA \xrightarrow{} \procD_{n-1}: \val_{n-1}} G_{n-1}' \ldots G_{1} \xrightarrow{\procD_{2} \xrightarrow{} \procB: \val_{1}} G_{1}' \ldots G_{0} \xrightarrow{\procB \xrightarrow{} \procA: m} G_{0}' \ldots\enspace.
\]
Let $\val_0=\val$ and $\procD_0=\procA$. We show by induction over $i$ that for all $i \in [1,n]$
\[
  \alpha' \cdot G'_{\mathit{pre}} \xrightarrow{l'} G'_{\mathit{post}} \dots G_{i} \xrightarrow{\procD_i \xrightarrow{} \procD_{i-1}: \val_{i-1}} G_{i}' ~\preforder~ \run'_{\procD_{i}}\enspace.
\]
We then obtain the desired contradiction with the fact that $\run'_{\procD_n} = \run'_{\procA} =  \alpha' \cdot G'_{\mathit{pre}}$. 

The base case of the induction follows immediately from the construction. For the induction step, assume that
\[
  \alpha' \cdot G'_{\mathit{pre}} \xrightarrow{l'} G'_{\mathit{post}} \dots G_{i} \xrightarrow{\procD_i \xrightarrow{} \procD_{i-1}: \val_{i-1}} G_{i}' ~\preforder~ \run'_{\procD_{i}}\enspace.
\]
From the definition of $\run'_{\procD_i}$ and the fact that $\procD_{i}$ is the active role in $\rcv{\procD_{i+1}}{\procD_{i}}{\val_i}$, it follows that $\rcv{\procD_{i+1}}{\procD_{i}}{\val_i} \in w$. Hence, we must also have $\snd{\procD_{i+1}}{\procD_{i}}{\val_i} \in w$. Since $\procD_{i+1}$ is the active role in $\snd{\procD_{i+1}}{\procD_{i}}{\val_i}$, we can conclude
\[
  \alpha' \cdot G'_{\mathit{pre}} \xrightarrow{l'} G'_{\mathit{post}} \dots G_{i} \xrightarrow{\procD_{i+1} \xrightarrow{} \procD_{i}: \val_{i}} G_{i+1}' ~\preforder~ \run'_{\procD_{i+1}}\enspace.
\]

\proofend

\end{proof}

\sndPrefixPreservation*

\begin{proof}
	Let $x$ = $\snd{\procA}{\procB}{\val}$.
	We prove the claim by induction on the length of $w$.

	\paragraph{\textbf{Base Case.}} $w = \emptystring$.
	By definition, $I(\emptystring)$ contains all maximal runs in $\semglobal(\GG)$, and the unique splitting prefix of any run $\run \in I(\emptystring)$ for $\procA$ with respect to $\emptystring$ is $\emptystring$.
	Because $\emptystring$ is a prefix of any run, we need only show the non-emptiness of $I(x)$.
	By \ref{lm:languages-of-roles}, $\lang(\GG) \wproj_{\Alphabet_{\procA}} = \lang(\subsetproj{\GG}{\procA})$. 
	Because $x$ is the prefix of a word in $\lang(\subsetproj{\GG}{\procA})$, there exists $w' \in \lang(\GG)$ such that 
	$x \leq w' \wproj_{\Alphabet_{\procA}}$. 
	By the semantics of $\lang(\GG)$, there exists a run $\run' \in \semglobal(\GG)$ such that 
	$x$ is the first symbol in $\SyncToAsync(\trace(\run')) \wproj_{\Alphabet_{\procA}}$, and therefore
	$\run' \in I(x)$. 

	\paragraph{\textbf{Induction Step.}} Let $wx$ be an extension of $w$ by $x \in \Alphabet_!$.

Let $\run$ be a run in $I(w)$, and let $\alpha \cdot G \xrightarrow{l} G' \cdot \beta$ be the unique splitting of $\run$ for role $\procA$ with respect to $w$.
	To re-establish the induction hypothesis, we need to show the existence of a run $\bar \run$ in $I(wx)$ such that 
	$\alpha \cdot G \leq \bar \run$.
	Since $\procA$ is the active role in $x$, it holds for any $\procC \neq \procA$ that
	$\globcomplocal{\GG}{\procC}{w} = \globcomplocal{\GG}{\procC}{wx}$.
	Therefore, to prove the existential claim, it suffices to construct a run $\bar \run$ that satisfies:
	\begin{enumerate}
		\item \label{claim:soundness-snd-case-in-extension-run-set}
		 $\bar \run \in \globcomplocal{\GG}{\procA}{wx}$,
		\item \label{claim:soundness-snd-case-in-original-run-set}
		$\bar \run \in I(w)$, and 
		\item \label{claim:soundness-snd-case-correct-prefix}
		$\alpha \cdot G \leq \bar \run$.
	\end{enumerate}

	In the case that $l \wproj_{\Alphabet_{\procA}} = x$, we
        are done:
        Property~\ref{claim:soundness-snd-case-correct-prefix} and \ref{claim:soundness-snd-case-in-original-run-set} hold by construction, and Property~\ref{claim:soundness-snd-case-in-extension-run-set} holds by the definition of possible run sets.
	In the case that $l \wproj_{\Alphabet_{\procA}} \neq x$, we show the existence of a transition label and state $\xrightarrow{\bar{l}} \bar{G'}$, and a maximal suffix $\bar \beta$ such that $\alpha \cdot G \xrightarrow{\bar{l}} \bar{G'} \cdot \bar{\beta}$ satisfies all three conditions.
	
	Let ($\vec{s}_w, \xi_w)$ denote the CSM configuration reached on $w$:
	$(\vec{s}_0, \xi_0) \xrightarrow{w} \mathrel{\vphantom{\to}^*} (\vec{s}_w, \xi_w) $ 
	Send Validity states that every transition in $\vec{s}_{w,\procA}$ originates in all global states in $\vec{s}_{w,\procA}$.
	By assumption, $\snd{\procA}{\procB}{\val}$ is a transition in $\vec{s}_{w,\procA}$.
	By Proposition \ref{prop:correspondence-unique-splittings-local-states}, $\rho \in I \subseteq \globcomplocal{\GG}{\procA}{w}$, and therefore $G \in \vec{s}_{w,\procA}$.
	Therefore, Send Validity gives the existence of some $\bar G' \in Q_{\semglobal(\GG)}$ such that $G \xrightarrow{\procA \xrightarrow{} \procB:m} \bar G' \in \delta_{\semglobal(\GG)}$.
	Because $\alpha \cdot G$ is a run in $\semglobal(\GG)$ and $G \xrightarrow{\procA \xrightarrow{} \procB:m} \bar G'$ is a transition in $\semglobal(\GG)$, $\alpha \cdot G \xrightarrow{\procA \xrightarrow{} \procB:m} \bar G'$ is a run in $\semglobal(\GG)$.

	The construction thus far satisfies Property~\ref{claim:soundness-snd-case-in-extension-run-set} and \ref{claim:soundness-snd-case-correct-prefix} regardless of our choice of maximal suffix: for all choices of $\bar \beta$ such that $\alpha \cdot G \xrightarrow{\procA \xrightarrow{} \procB:m} \bar{G'} \cdot \bar{\beta}$ is a maximal run, $wx \wproj_{\Alphabet_\procA} \leq  \SyncToAsync(\trace(\alpha \cdot G \xrightarrow{\procA \xrightarrow{} \procB:m} \bar G' \cdot \bar \beta)) \wproj_{\Alphabet_\procA}$ and $\alpha \cdot G \leq \alpha \cdot G \xrightarrow{\procA \xrightarrow{} \procB:m} \bar G' \cdot \bar \beta$.

	Property~\ref{claim:soundness-snd-case-in-original-run-set}, however, requires that the projection of $w$ onto each role is consistent with $\bar \run$, and this cannot be ensured by the prefix alone.

	We construct the remainder of $\bar \run$ by picking an arbitrary maximal suffix to form a candidate run, and iteratively performing suffix replacements on the candidate run until it lands in $I$.
	Let $\bar \beta$ be a run suffix such that $\alpha \cdot G  \xrightarrow{\procA \xrightarrow{} \procB:m} \bar{G'} \cdot \bar \beta$ is a maximal run in $\semglobal(\GG)$.
	Let $\run_c$ denote our candidate run $\alpha \cdot G  \xrightarrow{\procA \xrightarrow{} \procB:m} \bar{G'} \cdot \bar \beta$.
	If $\rho_c \in I$, we are done.
	Otherwise, $\run_c \notin I$ and there exists a non-empty set of processes $\mathcal{S} \subseteq \Procs$ such that for each $\procC \in \mathcal{S}$,
	\begin{align}\label{eq:not-prefix-of-rho-c}
		w \wproj_{\Alphabet_\procC} \nleq \SyncToAsync(\trace(\run_c)) \wproj _{\Alphabet_\procC}\enspace.
	\end{align}

	By the fact that $\rho \in I$,
	\begin{align}\label{eq:prefix-of-rho}
		w \wproj_{\Alphabet_\procC} \leq \SyncToAsync(\trace(\run)) \wproj _{\Alphabet_\procC}\enspace.
	\end{align}
	We can rewrite (\ref{eq:not-prefix-of-rho-c}) and (\ref{eq:prefix-of-rho}) above as:
	\begin{align}
          \label{eq:not-prefix-of-rho-c-expanded}
		w \wproj_{\Alphabet_{\procC}} &\nleq \SyncToAsync(\trace(\alpha \cdot G  \xrightarrow{\procA \xrightarrow{} \procB:m} \bar{G'} \cdot \bar \beta)) \wproj _{\Alphabet_\procC} \\
		w \wproj_{\Alphabet_{\procC}} &\leq \SyncToAsync(\trace(\alpha \cdot G  \xrightarrow{l} G' \cdot \beta)) \wproj _{\Alphabet_\procC}\enspace.
          \label{eq:prefix-of-rho-expanded}
	\end{align}
	By the definition of unique splitting, $\procA$ is the active role in $l$. 
	By Lemma \ref{lm:languages-of-roles}, 
	$\lang(\GG) \wproj_{\Alphabet_\procA} = \lang(\subsetproj{\GG}{\procA})$, and because 
	$\SyncToAsync(\trace(\rho)) \in \lang(\GG)$, 
	it holds that 
	$\SyncToAsync(\trace(\rho)) \wproj_{\Alphabet_\procA} \in \lang(\GG) \wproj_{\Alphabet_\procA}$, 
	and
	$\SyncToAsync(\trace(\rho)) \wproj_{\Alphabet_\procA} \in \lang(\subsetproj{\GG}{\procA})$.
	By assumption, $\subsetproj{\GG}{\procA}$ is in state $\vec{s}_{w,\procA}$ upon consuming $w \wproj_{\Alphabet_{\procA}}$.
	Then, there must exist an outgoing transition from $\vec{s}_{w,\procA}$ labeled with $\SyncToAsync(l) \wproj_{\Alphabet_{\procA}}$.
	By No Mixed Choice (\cref{cor:no-mixed-choice}), all outgoing transitions from $\vec{s}_{w,\procA}$ must be send actions.
	Therefore, $l$ must be of the form $\msgFromTo{\procA}{\procB'}{\val'}$.
	By assumption, $\procB' \neq \procB \lor \val' \neq \val$.

	We can further rewrite (\ref{eq:not-prefix-of-rho-c-expanded}) and (\ref{eq:prefix-of-rho-expanded}) to make explicit their shared prefix:
	\begin{align}
          \label{eq:not-prefix-of-rho-c-shared}
		w \wproj_{\Alphabet_\procC} &\nleq (\SyncToAsync(\trace(\alpha \cdot G)).~\snd{\procA}{\procB}{\val}.~\rcv{\procA}{\procB}{\val}.~ \SyncToAsync(\trace(\bar \beta))) \wproj _{\Alphabet_\procC} \\
          w \wproj_{\Alphabet_\procC} &\leq (\SyncToAsync(\trace(\alpha \cdot G)).~\snd{\procA}{\procB'}{\val'}.~\rcv{\procA}{\procB'}{\val'}.~\SyncToAsync(\trace(\beta))) \wproj _{\Alphabet_\procC}                                        
          \label{eq:prefix-of-rho-shared}
	\end{align}

	It is clear that in order for both (\ref{eq:not-prefix-of-rho-c-shared}) and (\ref{eq:prefix-of-rho-shared}) to hold, it must be the case that
	$\SyncToAsync(\trace(\alpha \cdot G)) \wproj_{\Alphabet_\procC} \leq w \wproj_{\Alphabet_{\procC}}$.

	We formalize the point of disagreement between $w \wproj_{\Alphabet_\procC}$ and $\run_c$ using an index $i_\procC$ representing the position of the first disagreeing transition label in $\trace(\run_c)$:
	\[
	i_\procC \is \text{max}\{i \mid  \SyncToAsync(\trace(\run_c[0..i-1])) \wproj_{\Alphabet_{\procC}} \leq w \wproj_{\Alphabet_{\procC}} \}\enspace.
	\]
	Then, $\SyncToAsync(\trace(\run_c[i_\procC])) \wproj_{\Alphabet_{\procC}} \neq \emptystring$ and from (\ref{eq:not-prefix-of-rho-c-shared}) and (\ref{eq:prefix-of-rho-shared}) we know that $i_\procC > 2*|\SyncToAsync(\trace(\alpha \cdot G))|$.

	We identify the role in $\mathcal{S}$ with the \textit{earliest disagreement} in $\run_c$: let $\bar{\procC}$ be the role with the smallest $i_{\bar{\procC}}$ in $\mathcal{S}$.
	Let $y_{\bar \procC}$ denote $\SyncToAsync(\trace(\run_c[i_{\bar \procC}])) \wproj_{\Alphabet_{\bar \procC}}$.

	\paragraph{Claim.} $y_{\bar \procC}$ must be a send event.

	Assume by contradiction that $y_{\bar \procC}$ is a receive event.
	We identify the symbol in $w$ that disagrees with $y_{\bar \procC}$: let $w'$ be the largest prefix of $w$ such that $w' \wproj_{\Alphabet_{\bar \procC}} \leq \SyncToAsync(\trace(\run_c))$.
	By definition, $w' \wproj_{\Alphabet_{\bar \procC}} = \SyncToAsync(\trace(\run_c[0..i_{\bar \procC}-1])) \wproj_{\Alphabet_{\bar \procC}}$.
	Let $z$ be the next symbol following $w'$ in $w$; then $z \in \Alphabet_{\bar \procC}$ and $z \neq y_{\bar \procC}$.
	Furthermore, by No Mixed Choice (\ref{cor:no-mixed-choice}) we have that $z \in \Alphabet_?$. 

	By assumption,
	$w'z \nleq \SyncToAsync(\trace(\run_c[0..i_{\bar \procC}]))$.
	Therefore, any run that begins with $\run_c[0..i_{\bar \procC}]$ cannot be contained in $\globcomplocal{\GG}{\bar \procC}{w'z}$, or consequently in $I(w'z)$.
	We show however, that $I(w'z)$ must contain some runs that begin with $\run_c[0..i_{\bar \procC}]$.
	From Lemma \ref{lm:rcvIntersectionSetEquality} for traces $w'$ and $w' z$, we obtain that
	$I(w') = I(w'z)$.
	Therefore, it suffices to show that $I(w')$ contains runs that begin with $\run_c[0..i_{\bar \procC}]$.

	\paragraph{Claim} $\forall w'' \leq w'.\, I(w'')$ contains runs that begin with $\run_c[0..i_{\bar \procC}]$. 
	
	We prove the claim via induction on $w'$.

	The base case is trivial from the fact that $I(\emptystring)$ contains all maximal runs.

	For the inductive step, let $w''y \leq w'$.

	In the case that $y \in \Alphabet_?$, from Lemma \ref{lm:rcvIntersectionSetEquality} $I(w''y) = I(w'')$ and the witness from $I(w'')$ can be reused.

	In the case that $y \in \Alphabet_!$, let $\procD$ be the active role of $y$ and let $\run'$ be a run in $I(w'')$ beginning with $\run_c[0..i_{\bar \procC}]$ given by the inner induction hypothesis.
	Let $\alpha' \cdot G' \xrightarrow{l'} G'' \cdot \beta'$ be the unique splitting of $\run'$ for $\procD$ with respect to $w''$.
	If $\SyncToAsync(l') \wproj_{\Alphabet_\procD} = y$, then $\rho'$ can be used as the witness.
	Otherwise, $\SyncToAsync(l') \wproj_{\Alphabet_\procD} \neq y$, and $\rho' \notin \globcomplocal{\GG}{\procD}{w''y}$.

The outer induction hypothesis holds for all prefixes of $w$: we instantiate it with $w''$ and $y$ to obtain:
	\[
		\exists~\rho'' \in I(w''y).~\alpha' \cdot G' \leq \rho''\enspace.
	\]
	Let $i_\procD$ be defined as before; it follows that $\run'[i_\procD] = G'$.
	It must be the case that $i_\procD > i_{\bar \procC}$: if $i_\procD \leq i_{\bar \procC}$, because $\run_c$ and $\run'$ share a prefix $\run_c[0..i_{\bar \procC}]$ and $w''y \leq w$, $\procD$ would be the earliest disagreeing role instead of $\bar \procC$.

	Because $i_\procD > i_{\bar \procC}$, 
	$\run_c[0..i_{\bar \procC}] = \run'[0..i_{\bar \procC}] \leq \run'[0..i_\procD]$.
	Because $\run'[0..i_\procD] = \alpha' \cdot G' \leq \run''$,
        it follows from prefix transitivity that $\run_c[0..i_{\bar \procC}] \leq \run''$,
        thus re-establishing the induction hypothesis for $w''y$ with $\run''$ as a witness run that begins with $\run_c[0..i_{\bar \procC}]$.

	This concludes our proof that $I(w')$ contains runs that begin with $\run_c[0..i_{\bar \procC}]$, and in turn our proof by contradiction that $y_{\bar \procC}$ must be a receive event.

	We can rewrite candidate run $\run_c$ as follows:
	\[
	\run_c = G_0 \xrightarrow{l_0} G_1 \ldots G_{i_{\bar \procC}} \xrightarrow{l_{i_{\bar \procC}}} G_{i_{\bar \procC}+1} \ldots\enspace.
	\]
	We have established that $l_{i_{\bar \procC}}$ must be a send event for $\bar \procC$.
	We can reason from Send Validity similarly to our construction of $\bar \run$'s prefix above, and conclude that there exists a transition label and maximal suffix from $G_{i_{\bar \procC}}$ such that the resulting run no longer disagrees with $w \wproj_{\Alphabet_{\bar \procC}}$.
	We update our candidate run $\run_c$ with the correct transition label and maximal suffix, update the set of states $\mathcal{S} \in \Procs$ to the new set of roles that disagree with the new candidate run, and repeat the construction above on the new candidate run until $\mathcal{S}$ is empty.

	Termination is guaranteed in at most $|w|$ steps by the fact that the number of symbols in $w$ that agree with the candidate run up to $i_{\bar \procC}$ must increase.

	Upon termination, the resulting $\bar \run$ satisfies the final remaining property \ref{claim:soundness-snd-case-correct-prefix}: $\bar \run \in I$.
	This concludes the proof of the inductive step, and consequently the proof of the prefix-preservation of send transitions.
\proofend
\end{proof}

\intersNonempty*

\begin{proof}
We prove the claim by induction on the length of $w$.

\paragraph{\textbf{Base Case.}} $w = \emptystring$.
The trace $w = \emptystring$ is trivially consistent with all maximal runs, and $I(w)$ therefore contains all maximal runs. 
By definition of $\GG$, language $\lang(\GG)$ is non-empty and at least one maximal run exists.
Thus, $I(w)$ is non-empty.

\paragraph{\textbf{Induction Step.}} Let $wx$ be an extension of $w$ by $x \in \AlphAsync$.

The induction hypothesis states that 
$I(w) \neq \emptyset$. 
To re-establish the induction hypothesis, we need to show
$I(wx) \neq \emptyset$. 
We proceed by case analysis on whether $x$ is a receive or send event.

\paragraph{Receive Case.} Let $x$ = $\rcv{\procB}{\procA}{\val}$.
By \cref{lm:rcvIntersectionSetEquality}, $I(wx) = I(w)$. $I(wx) \neq \emptyset$ follows trivially from the induction hypothesis and this equality.

\paragraph{Send Case.} Let $x$ = $\snd{\procA}{\procB}{\val}$.
By \cref{lm:sndPrefixPreservation}, there exists a run in $I(wx)$ that shares a prefix with a run in $I(w)$. $I(wx) \neq \emptyset$ again follows trivially.
\proofend
\end{proof}

\begin{lemma}
	\label{lm:terminatedEntailsGTComplete}
	Let $\GG$ be a global type and $\CSMl{\subsetcons{\GG}{\procA}}$ be the subset construction.
	Let $w$ be a trace of $\CSMl{\subsetcons{\GG}{\procA}}$. If $w$ is \terminated, then $w$ is \gtcomplete{\GG}.
\end{lemma}
\begin{proof}
	We prove the claim by contraposition and assume that $w$ is not \gtcomplete{\GG}.
	Then, there exists a run $\run \in I(w)$ and a non-empty set of roles $\mathcal{S}$ such that for every $\procC \in \mathcal{S}$, it holds that
	$
	w \wproj_{\Alphabet_\procC} \neq
	\bigl(
	\SyncToAsync(\trace(\run))
	\bigr)
	\wproj_{\Alphabet_\procC}
	$ (*).
	Since $w$ is a trace, we know there exists a run
	$(\vec{s}_0, \xi_0) \xrightarrow{w_0}
	\ldots
	\xrightarrow{w_{n-1}} (\vec{s}_n, \xi_n)$
	of $\CSMl{\subsetcons{\GG}{\procA}}$
	such that
	$w = w_0 \ldots w_{n-1}$.
	We need to show that there exists $(\vec{s}_{n+1}, \xi_{n+1})$ with
	$(\vec{s}_n, \xi_n) \xrightarrow{w_n} (\vec{s}_{n+1}, \xi_{n+1})$ for some $w_n$.
Given some role $\procA$, let $\run_\procA$ denote the largest prefix of $\run$ that contains $\procA$'s local view of $w$. Formally,
	\[
	\run_\procA = max\{\run'~|~\run' \leq \run ~\land~ 
	\SyncToAsync(\trace(\run'))
	\wproj_{\Alphabet_\procA} = w\wproj_{\Alphabet_\procA}
	\}\enspace.
	\]
	Note that due to maximality, the next transition in $\run$ after $\run_\procA$ must have $\procA$ as its active role. 
	Let $\procB$ be the role in $\mathcal{S}$ for whom $\run_\procB$ is the smallest.
	From \cref{lm:projection-preserves-per-process-runs} and (*), it follows that $\vec{s}_{n, \procB}$ has outgoing transitions.
	If $q_{n, \procB}$ has outgoing send transitions, then $(q_{n+1}, \xi_{n+1})$ exists trivially. 
	If $q_{n, \procB}$ has outgoing receive transitions, it must be the case that the next transition in $\run$ after $\run_\procB$ is of the form $\msgFromTo{\procA}{\procB}{\val}$ for some $\procA$ and $\val$.
	From the fact that $\procB$ is the role with the smallest $\run_\procB$, we know that $\run_\procB < \run_\procA$, and from the FIFO property of CSM channels it follows that $\val$ is in $\xi_n(\procA,\procB)$. Then, the receive transition is enabled for $\procB$, and there exists $(\vec{s}_{n+1}, \xi_{n+1})$ with 
	$(\vec{s}_{n}, \xi_{n}) \xrightarrow{\rcv{\procA}{\procA}{\val}} (\vec{s}_{n+1}, \xi_{n+1})$. 
	This shows that $w = w_1 \ldots w_{n-1}$ is not terminated and concludes the proof.
\proofend
\end{proof}

\begin{lemma}
	\label{lm:gtcompleteEntailsLangMemb}
	Let $\GG$ be a global type and $\CSMl{\subsetproj{\GG}{\procA}}$ be the subset projection.
	Let $w$ be a trace of $\CSMl{\subsetproj{\GG}{\procA}}$.
	If $w$ is \gtcomplete{\GG}, then $w \in \lang(\GG)$.
\end{lemma}
\begin{proof}
	By definition of $w$ being \gtcomplete{\GG}, 
	\[
		\forall \procA\in \Procs,~\run \in I(w).~
		w \wproj_{\Alphabet_\procA} =
		\bigl(
		\SyncToAsync(\trace(\run))
		\bigr)
		\wproj_{\Alphabet_\procA}\enspace.
	\]
	From \cref{lm:intersNonempty}, 
	$I(w)$ is non-empty.
	Let $\run$ be a run in $I(w)$, and let 
	$w' = \SyncToAsync(\trace(\run)) \in \lang(\GG)$.
	By the semantics of $\lang(\GG)$, $\lang(\GG)$ is closed under the $\interswap$ relation, and thus it suffices to show that $w \interswap w'$.
	\cite[Lemma~23]{DBLP:conf/concur/MajumdarMSZ21} states that if $w$ is channel-compliant \cite[Definition~19]{DBLP:conf/concur/MajumdarMSZ21}, then 
	$w \interswap w'$ iff $w'$ is \channelcompliant and forall $\procA \in \Procs$,
	$w \wproj_{\Alphabet_\procA} = w' \wproj_{\Alphabet_\procA}$.
	The fact that $w$ is channel-compliant follows from \cite[Lemma~20]{DBLP:conf/concur/MajumdarMSZ21} and $w$ being a a CSM trace; 
	$w'$ is channel-compliant by construction, and the last condition is satisfied by assumption that $w$ is \gtcomplete{\GG} and by definition of $w'$. 
	Thus, we conclude that $w \interswap w'$. 
\proofend
\end{proof}

\soundnessTheorem*

\begin{proof}
First, we show that $\CSMl{\subsetproj{\GG}{\procA}}$ is deadlock-free, namely, that every finite trace extends to a maximal trace. 
Let $w$ be a trace of $\CSMl{\subsetproj{\GG}{\procA}}$. 
Let $w'$ denote the extension of $w$. 
If $w' \in \AlphAsync^\omega$, then $w'$ is maximal and we are done. 
Otherwise, we have $w' \in \AlphAsync^*$. Let $(\pvec{s}',\xi')$ denote the $\CSMl{\subsetproj{\GG}{\procA}}$ configuration reached on $w'$. 
By definition of $w'$ being the largest extension, $w'$ is a terminated trace, and there exists no configuration reachable from $(\pvec{s}',\xi')$. 
By \cref{lm:terminatedEntailsGTComplete}, $w'$ is \gtcomplete{\GG}. 
By \cref{lm:gtcompleteEntailsLangMemb}, $w' \in \lang(\GG)$. 
Therefore, all states in $\pvec{s}'$ are final and all channels in $\xi$ are empty, and $w'$ is a maximal trace in $\CSMl{\subsetproj{\GG}{\procA}}$. 

This concludes our proof that $\CSMl{\subsetproj{\GG}{\procA}}$ is deadlock-free. 

Next, we show that 
$\lang(\CSMl{\subsetproj{\GG}{\procA}}) = \lang(\GG)$. 
The backward direction, $\lang(\GG) \subseteq \lang(\CSMl{\subsetproj{\GG}{\procA}})$, is given by \cref{lm:constructionPreservesBehaviors}.
For the forward direction, let $w \in \lang(\CSMl{\subsetproj{\GG}{\procA}})$, and let $(\vec{s},\xi)$ denote the configuration reached on $w$. 
We proceed by case analysis on whether $w$ is a finite or infinite maximal trace. 
\paragraph{Case:} $w \in \AlphAsync^*$. 
We show a stronger property: $w$ is a terminated trace. 
Then, we use \cref{lm:terminatedEntailsGTComplete} and \cref{lm:gtcompleteEntailsLangMemb} as above to obtain $w \in \lang(\GG)$. 
By definition of $(\vec{s},\xi)$ being final, all states in $\vec{s}$ are final and all channels in $\xi$ are empty. 
We argue there does not exist a configuration reachable from $(\vec{s},\xi)$. 
From \cref{prop:final-state-outgoing-rcv}, all outgoing states from states in $\vec{s}$ must be receive transitions. 
However, no receive transitions are enabled because all channels in $\xi$ are empty. 
Therefore, $(\vec{s},\xi)$ is a terminated configuration and $w$ is a terminated trace. 

\paragraph{Case:} $w \in \Alphabet^\omega$. 
By the semantics of $\lang(\GG)$, to show $w \in \lang(\GG)$ it suffices to show:
\[
	\exists w' \in \Alphabet^\omega.~
	w' \in \SyncToAsync(\lang(\semglobal(\GG))) 
	\land 
	w \preceq_\interswap^\omega w' \enspace.
\]
\paragraph{Claim.} $\Inters_{u \leq w} I(u)$ contains an infinite run. 

First, we show that there exists an infinite run in $\semglobal(\GG)$. 
We apply König's Lemma to an infinite tree where each vertex corresponds to a finite run. 
We obtain the vertex set from the intersection sets of $w$'s prefixes; each prefix ``contributes'' a set of finite runs.
Formally, for each prefix $u \leq w$, let $V_u$ be defined as: 
\[
	V_u \is \Union_{\run_u \in I(u)} \text{min}\{\run' \mid \run' \leq \run_u \land \forall \procA \in \Procs.~u \wproj_{\Alphabet_{\procA}} \leq \SyncToAsync(\trace(\run')) \wproj_{\Alphabet_{\procA}} \} \enspace.
\]
By \cref{lm:intersNonempty}, $V_u$ is guaranteed to be non-empty. 
We construct a tree $\mathcal{T}_w(V,E)$ with 
$V \is \Union_{u \leq w} V_u$ and 
$E \is \{(\run_1, \run_2) \mid \run_1 \leq \run_2\}$. 
The tree is rooted in the empty run, which is included $V$ by $V_\emptystring$. 
$V$ is infinite because there are infinitely many prefixes of $w$. 
$\mathcal{T}_w$ is finitely branching due to the finiteness of $\delta_\GG$ and the fact that each vertex represents a finite run.
Therefore, there must exist a ray in  $\mathcal{T}_w$ representing an infinite run in $\semglobal(\GG)$. 

Let $\run'$ be such an infinite run. 
We now show that $\run' \in \Inters_{u \leq w} I(u)$. 
Let $v$ be a prefix of $w$. 
To show that $\run' \in I(v)$, it suffices to show that one of the vertices in $V_v$ lies on $\run'$. In other words, 
\[
	V_v \inters \{v \mid v \in \run'\} \neq \emptyset \enspace.
\]
Assume by contradiction that $\run'$ passes through none of the vertices in $V_v$. 
Then, for any $u' \geq u$, because intersection sets are monotonically decreasing, it must be the case that $\run'$ passes through none of the vertices in $V_u'$.
Therefore, $\run'$ can only pass through vertices in $V_u''$, where $u'' \leq u$. 
However, the set $\Union_{u'' \leq u} V_u''$ has finite cardinality. 
We reach a contradiction, concluding our proof of the above claim. 

Let $\run' \in \Inters_{u \leq w} I(u)$, and let 
$w' = \SyncToAsync(\trace(\run'))$. 
It is clear that 
$w' \in \AlphAsync^\omega$ and $w' \in \SyncToAsync(\lang(\semglobal(\GG)))$.
It remains to show that 
$w \preceq_\interswap^\omega w'$. 
By the definition of $ \preceq_\interswap^\omega$, it further suffices to show that: 
\[
	\forall u \leq w,~
	\exists u' \leq w', v \in \Alphabet^*.~
	uv \interswap u' \enspace.
\]
Let $u$ be an arbitrary prefix of $w$. 
Because by definition $\run' \in I(u)$, it holds that 
$u \wproj_{\Alphabet_\procA} \leq \SyncToAsync(\trace(\run')) \wproj_{\Alphabet_\procA}$.

For each role $\procA \in \Procs$, let $\run'_\procA$ be defined as the largest prefix of $\run'$ such that 
$\SyncToAsync(\trace(\run'_\procA)) \wproj_{\Alphabet_\procA} = u \wproj_{\Alphabet_\procA}$. 
Such a run is well-defined by the fact that $u$ is a prefix of an infinite word $w$, and there exists a longer prefix $v$ such that $u \leq v$ and 
$v \wproj_{\Alphabet_\procA} \leq \SyncToAsync(\trace(\run')) \wproj_{\Alphabet_\procA}$.

Let $\procD$ be the role with the maximum $|\run'_\procD|$ in $\Procs$. 
Let $u' = \SyncToAsync(\trace(\run'_\procD))$. 
Clearly, $u' \leq w'$. 
Because $u'$ is $\SyncToAsync(\trace(\run'_\procD))$ for the role with the longest $\run'_\procD$, it holds for all roles $\procA \in \Procs$ that 
$u \wproj_{\Alphabet_\procA} \leq u' \wproj_{\Alphabet_\procA}$. 
Then, there must exist $y_\procA \in \Alphabet_\procA^*$ such that 
\[
	u \wproj_{\Alphabet_\procA} \cdot y_\procA = u' \wproj_{\Alphabet_\procA}\enspace.
\]
Let $y_\procA$ be defined in this way for each role. 
We construct $v \in \Alphabet^*$ such that $uv \interswap u'$. 
Let $v$ be initialized with $\emptystring$.
If there exists some role in $\Procs$ such that $y_\procA[0] \in \Alphabet_{\procA,!}$, append $y_\procA$ to $v$ and update $y_\procA$. 
If not, for all roles $\procA\in \Procs$, $y_\procA[0] \in \Alphabet_{\procA,?}$. 
Each symbol $y_\procA[0]$ for all roles appears in $u'$. 
Let $i_\procA$ denote for each role the index in $u'$ such that $u'[i] = y_\procA[0]$.
Let $\procC$ be the role with the minimum index $i_\procC$. 
Append $y_\procC$ to $v$ and update $y_\procC$. 
Termination is guaranteed by the strictly decreasing measure of $\sum_{\procA \in \Procs} |y_\procA|$. 

We argue that $uv$ satisfies the inductive invariant of channel compliancy. 
In the case where $v$ is extended with a send action, channel compliancy is trivially re-established. 
In the receive case, channel compliancy is re-established by the fact that the append order for receive actions follows that in $u'$, which is channel-compliant by construction.
We conclude that $uv \interswap u'$ by applying \cite[Lemma~22]{DBLP:conf/concur/MajumdarMSZ21}.  
\proofend
\end{proof}

     \section{Additional Material for \cref{sec:completeness}}
\label{app:completeness}

\begin{lemma}
\label{lm:run-for-subterm-and-state}
Let $\procA$ be a role, $\GG$ be a global type,
$G'$ be a syntactic subterm of~$\GG$, and
$\CSMl{ \subsetcons{\GG}{\procA} }$
be its subset construction.
Let $s$ be some state in $Q_{\procA}$ with $G' \in s$.
Then, there is a run $\run_{\GG}$ in $\semglobalsync(\GG)$ ending in state $q_G'$,
i.e.\
\[
\run_{\GG} = q_{0, \GG} \xrightarrow{\trace(\rho_{\GG})}\mathrel{\vphantom{\to}^*} q_G',
\]
such that $\subsetcons{\GG}{\procA}$ will reach $s$ on the projected trace, i.e.,
\[
\run_{\procA} = s_{0, \procA} \xrightarrow{\SyncToAsync(\trace(\rho_{\GG})) \wproj_{\Alphabet_\procA}}\mathrel{\vphantom{\to}^*} s.
\]
\end{lemma}
\begin{proof}
Recall that the set of states for $\subsetcons{\GG}{\procA}$ is defined as a least fixed point:
\[
Q_{\procA} \is
\lfp_{\set{s_{0,\procA}}}^\subseteq \lambda Q.\, Q \cup \set{ \delta(s,a) \mid s \in Q \land a \in \Alphabet_{\procA}} \setminus \set{\emptyset}
\]
where $\delta(s,a)$ is an intermediate transition relation that is defined for all subsets $s \subseteq Q_\GG$ and every event $a \in \Alphabet_\procA$ as follows: 
\[
\delta(s, a) \is
\set{q' \in Q_{\GG}
	\mid
	\exists q \in s,
	q \xrightarrow{a} \xrightarrow{\emptystring}\mathrel{\vphantom{\to}^*} q' \in \projerasuretrans
}
\]
From the definition of $Q_\procA$, there exists a sequence of states $s_1, \ldots, s_n$ such that $s_1 = s_{0,\procA}$, $s_n = s$ and for every $i \in \set{1, \ldots, n-1}$, it holds that 
\[
	\exists a \in \Alphabet_\procA.~\delta(s_i,a) = s_{i+1}
\]
Let $a_i$ denote the existential witness for each $i$. 
From the definition of $\delta(s,a)$, for every $i \in \set{1, \ldots, n-1}$, it follows that
\[
	\forall q' \in s_{i+1}.~\exists q \in s_i.~
	q \xrightarrow{a_i} \xrightarrow{\emptystring}\mathrel{\vphantom{\to}^*} q' \in \projerasuretrans
\]
By assumption, $G' \in s$. 
There then exists a sequence of global syntactic subterms $G_1, \ldots, G_n$ such that $G_1 = \GG$, $G_n = G'$ and for every $i \in \set{1, \ldots, n-1}$, it holds that
\[
	G_i \in s_i 
	\land 
	G_{i+1} \in s_{i+1} 
	\land 
	G_i \xrightarrow{a_i} \xrightarrow{\emptystring}\mathrel{\vphantom{\to}^*} G_{i+1} \in \projerasuretrans
\]
We can expand $\emptystring^*$:
for every $i \in \set{1, \ldots, n-1}$, there exists $k_i \geq 0$ and a sequence of syntactic subterms $G_{i,0}, \ldots, G_{i,k_i}$ such that $G_{i,0} = G_i$ and $G_{i,k_i} = G_{i+1}$ and
\[
		G_{i,0} \xrightarrow{a} G_{i,1} \in \projerasuretrans
\text{ and }
		G_{i,j} \xrightarrow{\emptystring} G_{i,j+1} \in \projerasuretrans
\text{ for every } j \in \set{1, k_i - 1}.
\]
This expansion yields a run $\run_\downarrow$ in the projection by erasure $\projerasure{\GG}{\procA}$.
Because of recursion terms, the expansion might not be unique, but we can pick the smallest~$k_i$ possible for every $i$.
With the definition of $\projerasuretrans$, it is trivial to translate this run in $\projerasure{\GG}{\procA}$ to a run $\run_\GG$ in $\semglobalsync(\GG)$:
the events $a \in \Alphabet_\procA$ become $a' \in \AlphSync$ such that
$\SyncToAsync(a') \wproj_{\Alphabet_\procA} = a$ 
and $\emptystring$ becomes $b \in \AlphSync$ such that
$\SyncToAsync(b) \wproj_{\Alphabet_\procA} = \emptystring$.

It is clear by construction that $s_1, \ldots, s_n$ (with its corresponding transitions) serves as a witness for $\run_\procA$, while $G_{1,0}, \ldots, G_{1,k_1}, \ldots, G_n$ (with its respective transitions) serves as a witness for $\run_\GG$.
\proofend
\end{proof}

\begin{lemma}
\label{lm:impl-subset-cons-step-for-step}
	Let $\GG$ be a global type and let $\CSM{B}$ implement $\GG$ with $B_\procA = (Q_{B,\procA}, \delta_{B,\procA}, s_{B,0,\procA}, F_{B,0,\procA})$ for role $\procA$.
	Let $s \in Q_\procA$, $x \in \Alphabet_{\procA}$ and $s \xrightarrow{x} t \in \delta_\procA$ from the subset construction $\subsetcons{\GG}{\procA}$. 
	Let $u \in \Alphabet_\procA^*$ such that $s_{0, \procA} \xrightarrow{u} \mathrel{\vphantom{\to}^*} s$. 
	Then, there exists $s', t' \in Q_{B,\procA}$ such that $s_{B,0,\procA} \xrightarrow{u} \mathrel{\vphantom{\to}^*} s'$ and $s' \xrightarrow{x} t' \in \delta_{B,\procA}$.
\end{lemma}
\begin{proof} 
	Because $\CSM{B}$ implements $\GG$, it must hold that $\lang(\GG) \wproj_{\Alphabet_\procA} \subseteq \lang(B_\procA)$ since $B_\procA$ must produce at least the behaviors as specified by $\GG$ for its role.
	It follows that $\pref(\lang(\GG) \wproj_{\Alphabet_\procA}) \subseteq \pref(\lang(B_\procA))$.
	From \cref{lm:languages-of-roles}, we know that $\lang(\GG) \wproj_{\Alphabet_\procA} = \subsetcons{\GG}{\procA}$. By construction of $\subsetcons{\GG}{\procA}$, if $ux$ is reachable from the initial state in $\subsetcons{\GG}{\procA}$ then $ux$ is the prefix of some word in $\lang(\GG) \wproj_{\Alphabet_\procA}$.
Therefore, it holds that $ux \in \pref(\lang(\GG) \wproj_{\Alphabet_\procA})$ and consequently $ux \in \pref(\lang(B_\procA))$.
	Because $B_\procA$ is deterministic, there exists a unique $t'$ such that $B_\procA$ reaches $t'$ from the initial state on $ux$.
	This concludes our proof. 
\end{proof} 

\completenessThm*

\begin{proof}
\label{proof:completeness}

From the fact that $\GG$ is implementable, we know there exists a CSM $\CSM{B}$ that implements $\GG$.
Showing that the subset projection is defined amounts to showing that Send and Receive Validity (\cref{cond:send-state-validity-transition-origins,cond:rcv-state-validity}) hold for the subset construction. 

We proceed by contradiction and assume the negation of Send and Receive Validity in turn, and in each case derive a contradiction to the fact that $\CSM{B}$ implements~$\GG$.
Specifically, we contradict protocol fidelity (\cref{def:implementability}(\ref{def:implementability-protocol-fidelity})), and show that $\lang(\GG) \neq \lang(\CSM{B})$.

To prove the inequality of the two languages, it suffices to prove the inequality of their respective prefix sets, i.e. 
\[
	\{ u \mid u \leq w \land w \in \lang(\GG)\} 
	\neq 
	\{ u \mid u \leq w \land w \in \lang(\CSM{B}) \}
\]
Specifically, we show there is $v \in \AlphAsync^*$ such that
\begin{align*}
	v &\in \{ u \mid u \leq w \land w \in \lang(\CSM{B})\}
	~\land \\
	v &\notin \{ u \mid u \leq w \land w \in \lang(\GG)\} \enspace.
\end{align*}
Because $\CSM{B}$ is deadlock-free by assumption, every trace either can be extended to end in a final configuration or to be infinite. Therefore, any word $v \in \AlphAsync^*$ that is a trace of $\CSM{B}$ is a member of the prefix set, i.e.
\[
	\exists ~(\vec{s},\xi).~(\vec{s}_0, \xi_0) \xrightarrow{v} \mathrel{\vphantom{\to}^*} (\vec{s}, \xi) 
	\implies v \in \{ u \mid u \leq w \land w \in \lang(\CSM{B})\}\enspace.
\] 
By the semantics of $\lang(\GG)$, for any $w \in \lang(\GG)$, there exists $w' \in \lang(\semglobal(\GG))$ with $w \interswap \SyncToAsync(w')$. 
For any $w' \in \lang(\semglobal(\GG))$, it is straightforward that \mbox{$I(\SyncToAsync(w')) \neq \emptyset$}.
Because intersection sets are closed under the indistinguishability relation (Corollary \ref{cor:I-set-indist-invar}), it holds that $I(w) \neq \emptyset$.
Because $I(\hole)$ is monotonically decreasing, if $I(w)$ is non-empty then for any $v \leq w$, $I(v)$ is non-empty.
By the following, to show that a word $v$ is not a member of the prefix set of $\lang(\GG)$ it suffices to show that $I(v)$ is empty:
\[
	\forall v \in \AlphAsync^*.~
	I(v) = \emptyset \implies \forall w. ~v \leq w \implies w \notin \lang(\GG)\enspace.
\]
Therefore, under the assumption of the negation of Send or Receive Validity  respectively, we explicitly construct a witness $v_0$ satisfying:
\begin{enumerate}[(a)]
	\item \label{cond:word-is-trace}
	$v_0$ is a trace of $\CSM{B}$, and
	\item \label{cond:word-I-set-empty}
	$I(v_0) = \emptyset$.
\end{enumerate}

\myparagraph{Send Validity (\cref{cond:send-state-validity-transition-origins}).}  
Assume that Send Validity does not hold for some role $\procA \in \Procs$.
Let $s \in Q_{\procA}$ be a state and $s \xrightarrow{\snd{\procA}{\procB}{\val}} s' \in \delta_{\procA}$ a transition in the subset construction $\subsetcons{\GG}{\procA}$ such that
\[
 	\transAnnoFunc(s  \xrightarrow{\snd{\procA}{\procB}{\val}} s') \neq s \enspace.
\]

Let $D$ denote $s \setminus \transAnnoFunc(s  \xrightarrow{\snd{\procA}{\procB}{\val}} s')$.
By the negation of Send Validity, $D$ is non-empty.
Let $G'$ be a syntactic subterm in $D$.

Because $G' \in s$, it follows from Lemma \ref{lm:run-for-subterm-and-state} that there exists $\alpha$ such that $\alpha \cdot G'$ is a run in $\semglobal(\GG)$. 
Let $\bar w$ be $\SyncToAsync(\trace(\alpha \cdot G'))$. 
Because $\CSM{B}$ implements $\GG$, there exists a configuration $(\vec{t},\xi)$ of $\CSM{B}$ such that
$(\vec{t}_0,\xi_0) \xrightarrow{\bar w}\mathrel{\vphantom{\to}^*} (\vec{t},\xi)$.
Instantiating \cref{lm:impl-subset-cons-step-for-step} with $s$, $s  \xrightarrow{\snd{\procA}{\procB}{\val}} s'$ and $\SyncToAsync(\trace(\alpha \cdot G')) \wproj_{\Alphabet_\procA}$, it follows that $\vec{t}_\procA$ has an outgoing transition labeled $\snd{\procA}{\procB}{\val}$. Let $\vec{t}_\procA \xrightarrow{\snd{\procA}{\procB}{\val}} t''$ be this transition.

The send transitions of any local machine in a CSM are always enabled. Formally, for all
$w \in \AlphAsync^*$, $x \in \Alphabet_!$, and $\procC \in \Procs$, if
	$w$ is a trace of $\CSM{B}$ and
	$\vec{t}_{w,\procC} \xrightarrow{x} \pvec{t}'_{w,\procC} \in \delta_\procC$, then 
	$wx$  is a trace of $\CSM{B}$.
Instantiating this fact with $\bar w$ and $\vec{t}_\procA \xrightarrow{\snd{\procA}{\procB}{\val}} t''$, we obtain that $\bar w \cdot \snd{\procA}{\procB}{\val}$ is a trace of $\CSM{B}$.

Let $\bar w \cdot \snd{\procA}{\procB}{\val}$ be our witness $v_0$; it then follows that $v_0$ satisfies \ref{cond:word-is-trace}.
It remains to show that $v_0$ satisfies \ref{cond:word-I-set-empty}, namely $I(\bar w \cdot \snd{\procA}{\procB}{\val}) = \emptyset$.

\textit{Claim.} All runs in $I(\bar w)$ begin with $\alpha \cdot G'$.

\textit{Proof of Claim.}
Recall that $\bar w$ is defined as $\SyncToAsync(\trace(\alpha \cdot G'))$.
Assume by contradiction that $\run' \in I(\bar w)$ and $\run'$ does not begin with $\alpha \cdot G'$.
Due to the syntactic structure of global runs, the first divergence between two runs must correspond to a syntactic subterm of the form $\sum_{i ∈ I} \msgFromTo{\procA'}{\procB'_{i}}{\val'_i.G'_i}$.
Let $\procA'$ be the sender in the first divergence between $\run'$ and $\alpha \cdot G'$, and let the two runs respectively contain the subterms $G'_i$ and $G'_j$.
Because $\run'$ is in $\globcomplocal{\GG}{\procA'}{\bar w}$, it holds that
$\bar w \wproj_{\Alphabet_{\procA'}} \leq \SyncToAsync(\trace(\rho')) \wproj_{\Alphabet_{\procA'}}$.
Because $\bar w = \SyncToAsync(\trace(\alpha \cdot G'))$, we can rewrite the inequality as
$\SyncToAsync(\trace(\alpha \cdot G')) \wproj_{\Alphabet_{\procA'}} \leq \SyncToAsync(\trace(\rho')) \wproj_{\Alphabet_{\procA'}}$.

We know that $\SyncToAsync(\trace(\alpha \cdot G')) \wproj_{\Alphabet_{\procA'}}$ and $\SyncToAsync(\trace(\rho')) \wproj_{\Alphabet_{\procA'}}$ share a common prefix, followed by different send actions from $\procA'$, i.e., they are respectively of the form $x' \cdot \snd{\procA'}{\procB_j}{\val'_j} \cdot y'$ and $x' \cdot \snd{\procA'}{\procB_i}{\val'_i} \cdot z'$.
We arrive at a contradiction.

\textit{End Proof of Claim.}

Recall that $G' \in D$ and $D = s \setminus \transAnnoFunc(s  \xrightarrow{\snd{\procA}{\procB}{\val}} s')$.
By the definition of $\transAnnoFunc(-)$ (\cref{def:subset-construction}), there does not exist a global syntactic subterm $G''$ with
$G' \xrightarrow{l'}\mathrel{\vphantom{\to}^*} G'' \in \delta_{\GG}$ such that $l' \wproj_{\Alphabet_\procA} = \snd{\procA}{\procB}{\val}$.
Therefore, there does not exist a maximal run in $\globcomplocal{\GG}{\procA}{\bar w \cdot \snd{\procA}{\procB}{\val}}$, and
$I(\bar w \cdot \snd{\procA}{\procB}{\val}) = \emptyset$
follows.

Our witness $v_0 = \bar w \cdot \snd{\procA}{\procB}{\val}$ thus satisfies both conditions 
 \ref{cond:word-is-trace} and \ref{cond:word-I-set-empty} required for a contradiction. This concludes our proof that Send Validity is required to hold.

\myparagraph{Receive Validity (\cref{cond:rcv-state-validity}).}
 Assume that Receive Validity does not hold for some role $\procA \in \Procs$.
 In other words, there exists $s \in Q_\procA$ with two transitions 
 $s \xrightarrow{\rcv{\procB_1}{\procA}{\val_1}} s_1$, 
 $s \xrightarrow{\rcv{\procB_2}{\procA}{\val_2}} s_2 \in \delta_\procA$
 and $G_2 \in \transAnnoFunDest(s \xrightarrow{\rcv{\procB_2}{\procA}{\val_2}} s_2)$ 
 such that 
 \[
 	\procB_1 \neq \procB_2 
 	\land  
 	\snd{\procB_1}{\procA}{\val_1} \in \semavail^{\procA}_{(G_2 \ldots)} \enspace.
 \]

\textit{Claim I.} There exists $u \in \AlphAsync^*$ such that both
$ u \cdot \rcv{\procB_1}{\procA}{\val_1}$ 
and 
$ u \cdot \rcv{\procB_2}{\procA}{\val_2}$ 
are traces of $\CSM{B}$.

\textit{Proof of Claim I.}
By the negation of Receive Validity, 
$G_2 \in \transAnnoFunDest(s \xrightarrow{\rcv{\procB_2}{\procA}{\val_2}} s_2) \subseteq s_2$. 
From \cref{lm:run-for-subterm-and-state} for $s_2$ and~$G_2 \in s_2$, 
there exists $\run'$ such that $\run'$ ends in $G_2$ and is a run in
$\semglobal(\GG)$. 
Because $\SyncToAsync(\trace(\run'))$ is a prefix in $\lang(\GG)$ and by assumption $\CSM{B}$ implements $\GG$, there exists a $\CSM{B}$ configuration $(\vec{t},\xi)$ such that
$(\vec{t}_0,\xi_0)
\xrightarrow{\SyncToAsync(\trace(\run'))}\mathrel{\vphantom{\to}^*}
(\vec{t},\xi)$.
By the subset construction, it holds that $\subsetcons{\GG}{\procA}$ reaches $s$ on $\SyncToAsync(\trace(\run'))$. 
Instantiating \cref{lm:impl-subset-cons-step-for-step} twice with $s \xrightarrow{\rcv{\procB_1}{\procA}{\val_1}} s_1$,
$s \xrightarrow{\rcv{\procB_2}{\procA}{\val_2}} s_2$ and $\SyncToAsync(\trace(\run')) \wproj_{\Alphabet_\procA}$, we obtain $t_1 \xrightarrow{\rcv{\procB_1}{\procA}{\val_1}} t_1'$ and $t_2 \xrightarrow{\rcv{\procB_2}{\procA}{\val_2}} t_2'$.
From the determinacy of $B_\procA$, it holds that $t_1 = t_2$.
Therefore, it holds that $\vec{t}_\procA= t_1$ and there exist two outgoing transitions from $\vec{t}_\procA$ labeled with $\rcv{\procB_1}{\procA}{\val_1}$ and $\rcv{\procB_2}{\procA}{\val_2}$.

From the fact that  
$s \xrightarrow{\rcv{\procB_2}{\procA}{\val_2}} s_2 \in \delta_\procA$, 
there exist $G_1 \in s$ and $G_2' \in \transAnnoFunDest(s \xrightarrow{\rcv{\procB_2}{\procA}{\val_2}} s_2) \subseteq s_2$ such that 
$G_1 \xrightarrow{\procB_2 \xrightarrow{} \procA: \val_2} G_2' \in \delta_{\GG}$. 
Either $G_2 = G_2'$, or $G_2$ is reachable from $G_2'$ via $\emptystring$-transitions for $\procA$.
Without loss of generality, assume that $G_2 = G_2'$; if $G_2' \neq G_2$ then $G_2'$ can also be picked as the witness from the definition of $M$. 
We rewrite $\run'$ as follows: 
\[
	\run' \is \alpha \cdot G_1 \xrightarrow{\procB_2 \xrightarrow{} \procA: \val_2} G_2
\]
From the negation of Receive Validity, we know that 
\[
	\snd{\procB_1}{\procA}{\val_1} \in \semavail^{\procA}_{(G_2 \ldots)}
\]

Then, there exists some suffix $\beta$ such that the transition 
$\xrightarrow{\procB_1 \xrightarrow{} \procA: \val_1}$ 
occurs in $\beta$ and 
$ \alpha \cdot G_1 \xrightarrow{\procB_2 \xrightarrow{} \procA: \val_2} G_2 \cdot \beta$
is a maximal run. 
Let $\run$ denote this maximal run. 
Let 
$G_3  \xrightarrow{\procB_1 \xrightarrow{} \procA: \val_1} G_4$ 
be the earliest occurrence of
$\xrightarrow{\procB_1 \xrightarrow{} \procA: \val_1}$ 
in $\beta$. 
We rewrite the suffix $\beta$ in $\run$ to reflect the existence of $G_3, G_4$: 
\[
	\run \is \alpha 
					\cdot 
					G_1 \xrightarrow{\procB_2 \xrightarrow{} \procA: \val_2} G_2 
					\cdot 
					\beta_1 
					\cdot 
					G_3  \xrightarrow{\procB_1 \xrightarrow{} \procA: \val_1} G_4 
					\cdot 
					\beta_2
\]
Note that $\beta_1$ does not contain any transitions of the form $\xrightarrow{\procB_1 \xrightarrow{} \procA: \val_1}$.

Let $\bar w$ denote $\SyncToAsync(\trace(\alpha))$, and $\bar v$ denote $\SyncToAsync(\trace(\beta_1))$. 
To produce a witness for $u$, we show that 
$\bar w \cdot \snd{\procB_2}{\procA}{\val_2} \cdot \snd{\procB_1}{\procA}{\val_1}$ 
is a trace of 
$\CSM{B}$,
and in the resulting CSM configuration $(\pvec{s}',\xi')$,
$\pvec{s}'_\procA$ has two outgoing transitions labeled 
$\rcv{\procB_1}{\procA}{\val_1}$ 
and 
$\rcv{\procB_2}{\procA}{\val_2}$. 
Moreover, we show that the channels
$\xi'(\procB_1,\procA)$ and $\xi'(\procB_2,\procA)$ 
respectively contain the messages $m_1$ and $m_2$ at the head. 

First, we show that 
$\bar w \cdot \snd{\procB_2}{\procA}{\val_2} \cdot \snd{\procB_1}{\procA}{\val_1}$ 
is a trace of 
$\CSM{B}$.

By assumption that $\CSM{B}$ implements $\GG$, both
$\bar w \cdot \snd{\procB_2}{\procA}{\val_2}$ 
and 
$\bar w \cdot \snd{\procB_2}{\procA}{\val_2} \cdot \rcv{\procB_2}{\procA}{\val_2}$
are traces of $\CSM{B}$.
Let $(\pvec{s}'',\xi'')$ and $(\pvec{s}''', \xi''')$ respectively denote $\CSM{B}$ configurations such that
\[
(\vec{s}_0, \xi_0) 
\xrightarrow{\bar w \cdot \snd{\procB_2}{\procA}{\val_2}} 
(\pvec{s}'', \xi'')
\xrightarrow{\rcv{\procB_2}{\procA}{\val_2} \cdot \bar v} 
(\pvec{s}''', \xi''')
\enspace. 
\]

Because send actions are always enabled in a CSM, it suffices to show that $\pvec{s}''_{\procB_1}$ has an outgoing transition label $\snd{\procB_1}{\procA}{\val_1}$. 
We do so by showing that 
$\pvec{s}''_{\procB_1} = \pvec{s}'''_{\procB_1}$:
it is clear from the fact that 
$
\bar w 
\cdot 
\snd{\procB_2}{\procA}{\val_2} 
\cdot 
\rcv{\procB_2}{\procA}{\val_2}
\cdot 
\bar v
\cdot 
\snd{\procB_1}{\procA}{\val_1}
$ 
is a trace of 
$\CSM{B}$
that 
$\pvec{s}'''_{\procB_1}$ 
has an outgoing transition label 
$\snd{\procB_1}{\procA}{\val_1}$.

Due to the determinacy of subset construction, it suffices to show that 
\[
	(\bar w \cdot \snd{\procB_2}{\procA}{\val_2}) \wproj_{\Alphabet_{\procB_1}}
	= 
	(\bar w \cdot \snd{\procB_2}{\procA}{\val_2} \cdot\rcv{\procB_2}{\procA}{\val_2} \cdot \bar v) \wproj_{\Alphabet_{\procB_1}} 
	\enspace.
\]
This equality follows from the definition of $\semavail$ and the fact that 
$\snd{\procB_1}{\procA}{\val_1} \in \semavail^{\procA}_{(G_2 \ldots)}$: because the blocked set of roles in $\semavail$ monotonically increases, and for any $G', B$, no actions in a run suffix starting with $G'$ involving roles in $B'$ are included in $\semavail^{B'}_{G'}$,
we know that $\snd{\procB_1}{\procA}{\val_1}$ must be the lexicographically earliest action involving $\procB_1$ in
$\bar v
\cdot 
\snd{\procB_1}{\procA}{\val_1} 
\cdot
\rcv{\procB_1}{\procA}{\val_1}
$. 
In other words, 
$\bar v \wproj_{\Alphabet_{\procB_1}} = \emptystring$.

This concludes the reasoning that 
$\bar w \cdot \snd{\procB_2}{\procA}{\val_2} \cdot \snd{\procB_1}{\procA}{\val_1}$ 
is a trace of 
$\CSM{B}$.

Recall that 
$(\pvec{s}',\xi')$
is the 
$\CSM{B}$ configuration reached on
$\bar w \cdot \snd{\procB_2}{\procA}{\val_2} \cdot \snd{\procB_1}{\procA}{\val_1}$.
We showed above that 
$\vec{s}_\procA$ 
has two outgoing transitions labeled 
$\rcv{\procB_1}{\procA}{\val_1}$ 
and 
$\rcv{\procB_2}{\procA}{\val_2}$.
It follows from the equality below that $\pvec{s}'_\procA$ likewise has two outgoing transitions labeled $\rcv{\procB_1}{\procA}{\val_1}$ 
and 
$\rcv{\procB_2}{\procA}{\val_2}$:
\[
(\bar w \cdot \snd{\procB_2}{\procA}{\val_2}) \wproj_{\Alphabet_{\procA}}
= 
(\bar w \cdot \snd{\procB_2}{\procA}{\val_2} \cdot \snd{\procB_1}{\procA}{\val_1}) \wproj_{\Alphabet_{\procA}}\enspace.
\]

We now show that the channels
$\xi'(\procB_1,\procA)$ and $\xi'(\procB_2,\procA)$ 
respectively contain the messages $m_1$ and $m_2$ at the head. 
Recall that $\bar w$ is defined as $\SyncToAsync(\trace(\alpha))$; this from the fact that $\xi_{\bar w}$ is uniquely determined by $\bar w$ and all channels in $\xi_{\bar w}$ are empty.

Let $u \is \bar w \cdot \snd{\procB_2}{\procA}{\val_2} \cdot \snd{\procB_1}{\procA}{\val_1}$. 
This concludes our proof that both 
$u \cdot \rcv{\procB_1}{\procA}{\val_1}$ 
and 
$u \cdot \rcv{\procB_2}{\procA}{\val_2}$ 
are traces of $\CSM{B}$.

\textit{End Proof of Claim I.}

The next claim establishes that our witness $u \cdot \rcv{\procB_1}{\procA}{\val_1}$ satisfies \ref{cond:word-I-set-empty}. 

\textit{Claim II.} It holds that $I(u \cdot \rcv{\procB_1}{\procA}{\val_1}) = \emptyset$.

\textit{Proof of Claim II.}
This claim follows trivially from the observation that every run in 
$I(\bar w \cdot \snd{\procB_2}{\procA}{\val_2})$ 
must begin with 
$\alpha \cdot G_1 \xrightarrow{\procB_2 \xrightarrow{} \procA: \val_2} G_2$. 
Because 
$I(u \cdot \rcv{\procB_1}{\procA}{\val_1}) 
\subseteq I(\bar w \cdot \snd{\procB_2}{\procA}{\val_2})$, 
and the $\SyncToAsync(\trace(\hole))$ of every run in 
$I(\bar w \cdot \snd{\procB_2}{\procA}{\val_2})$ 
starts with 
$\bar w \cdot \snd{\procB_2}{\procA}{\val_2} \cdot \rcv{\procB_2}{\procA}{\val_2}$,
therefore 
$I(u \cdot \rcv{\procB_1}{\procA}{\val_1})$ 
is empty. 

\textit{End Proof of Claim II.}

From here, the reasoning that every run in
$I(\bar w \cdot \snd{\procB_2}{\procA}{\val_2})$ must begin with 
$\alpha \cdot G_1 \xrightarrow{\procB_2 \xrightarrow{} \procA: \val_2} G_2$ 
is identical to the reasoning for the analogous claim in the Send Validity case, and thus omitted. 

By choosing 
$v_0 \is \bar u \cdot \rcv{\procB_1}{\procA}{\val_1}$,
we thus establish both conditions 
\ref{cond:word-is-trace} and \ref{cond:word-I-set-empty} required for a contradiction. This concludes our proof that Receive Validity is required to hold.
\proofend
\end{proof}

 \section{Additional Material for \cref{sec:discussion}}
\label{app:discussion}

\subsection{Visual Representations of $\GG_{\operatorname{fold}}$ and $\GG_{\operatorname{unf}}$}
\label{app:discussion-visual-reps}

\Cref{fig:folded-unfolded} depicts the examples in \cref{sec:discussion} visually.

\begin{figure}[h]
\begin{subfigure}[b]{0.42\textwidth}
\centering
\includegraphics[width=0.77\textwidth]{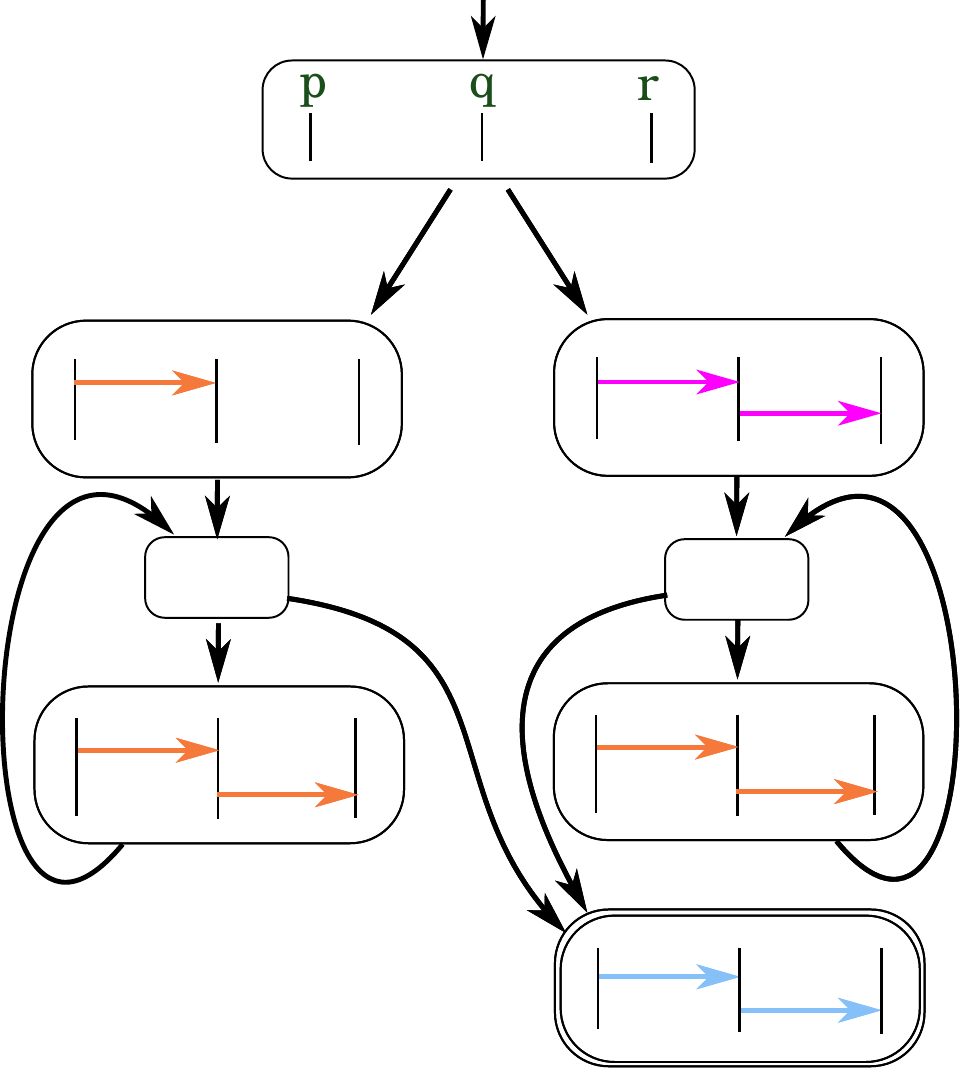}
\caption{
    Rejected by syntactic projection operators
    \label{fig:hmsc-folded-unproj}
}
\end{subfigure}
\hfill
\begin{subfigure}[b]{0.42\textwidth}
\centering
\includegraphics[width=0.75\textwidth]{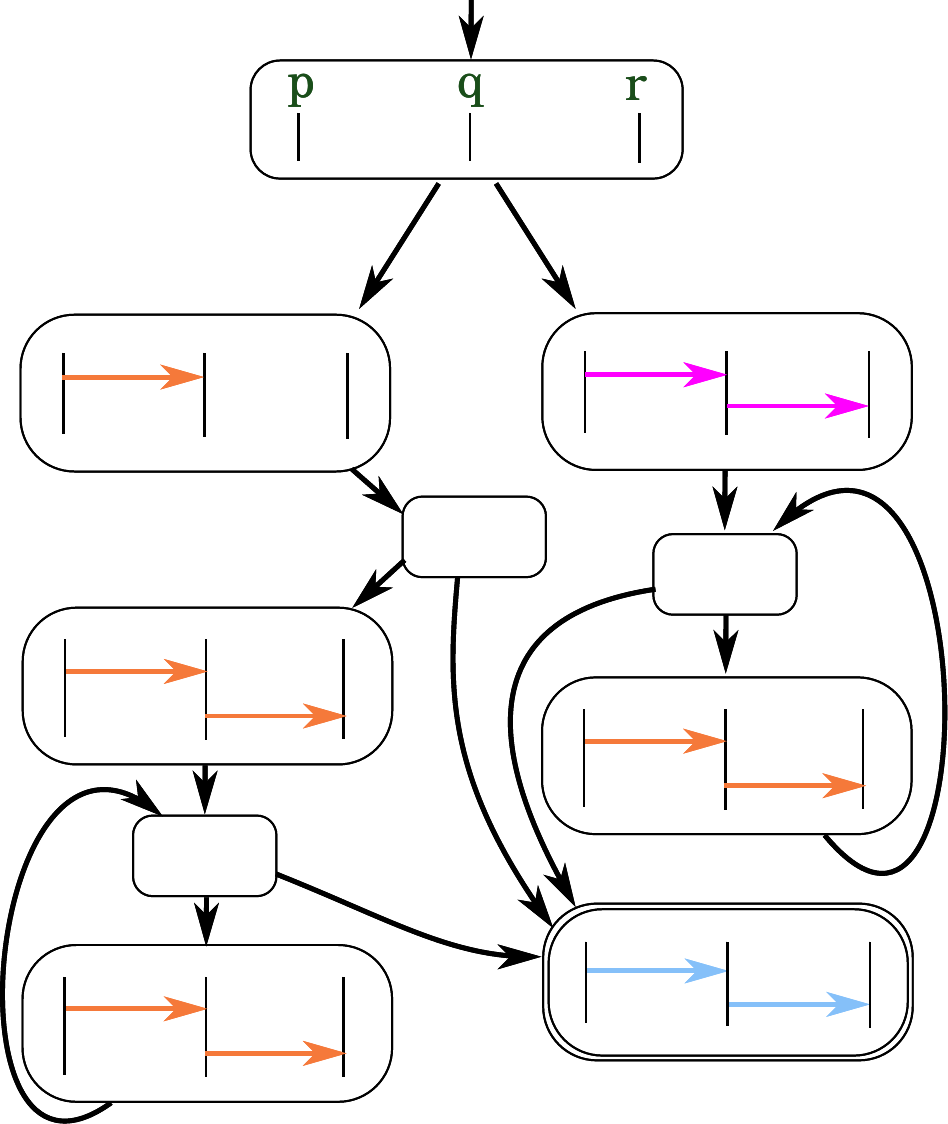}
\caption{
    Partially unfolded and accepted by syntactic projection operators
    \label{fig:hmsc-unfolded-proj}
}
\end{subfigure}

\caption{Two protocol specifications for the same protocol}
\label{fig:folded-unfolded}
\end{figure}

\subsection{Entailed Properties from the Literature -- Detailed Analysis}
\label{app:discussion-entailed}
For systems with two roles, an \emph{unspecified reception}~\cite[Def.~12]{DBLP:journals/iandc/CeceF05}
occurs if a receiver cannot receive the first message in its (only) channel.
Intuitively, this yields a deadlock for a binary system or an execution in which the other role will send messages indefinitely.
For multi-party systems with sender-driven choice, this is not necessarily the case and our receive validity condition ensures that a message in a role's channel can appear but will not confuse the receiver regarding which choices were made.
Thus, we could lift the property to multiparty systems with a universal quantification over channels.
Our subset projection would prevent this lifted version of unspecified receptions because of deadlock freedom and the fact that any non-deadlock configuration can be extended in a way that the unspecified reception can eventually be received.
Similarly, our subset projection ensures the \emph{absence of orphan messages}~\cite[Sec.~2]{DBLP:conf/icalp/DenielouY13}\cite[Sec.~3]{DBLP:conf/concur/BocchiLY15}.
However, we need to adapt the definition to our setting.
For most MST frameworks, the following two types of CSM configurations are equivalent:
final-state configurations, i.e., where each role is in a final state; and
sink-state configurations, i.e., where each role is in a state without outgoing transitions.
\emph{Orphan messages} have been defined using final-state configurations.
Our subset projection is more expressive so both configuration types do not necessarily coincide.
Our soundness proof ensures the absence of orphan messages in sink-state configurations and ensures that messages in final-state configurations can be received eventually.
We refer to \cref{sec:related} for a discussion on the expressiveness of local types and FSMs.

The standard notion of \emph{progress}~\cite[Sec.~1]{DBLP:conf/sfm/CoppoDPY15} asks that every sent message is eventually received and every process waiting for a message eventually receives one.
We proved our subset projection sound for finite as well as infinite CSM runs.
For finite runs, both properties trivially hold.
For infinite runs, our subset projection ensures that both is possible but one would require fairness assumptions to ensure that it will actually happen as is common for liveness~properties.

In our subset projection, it is also guaranteed that each transition of a local implementation can be fired in some execution of the subset projection.
This is called \emph{executable} by Cécé and Finkel~\cite[Def.~12]{DBLP:journals/iandc/CeceF05} while it is the property \emph{live} in work by Scalas and Yoshida~\cite[Fig.5(3)]{DBLP:journals/pacmpl/ScalasY19} and called \emph{liveness} by Barbanera et al.~\cite[Def.~2.9]{DBLP:conf/coordination/BarbaneraLT20}.

A CSM has the \emph{stable property} if any reachable configuration has a transition sequence to a configuration with empty channels.
With our proof technique, we showed every run of a CSM has a common path in the protocol that complies with all participants' local observations of the run.
There is a furthest point in this path and it is possible that all participants catch up to this point, having empty channels.

Scalas and Yoshida~\cite{DBLP:journals/pacmpl/ScalasY19} also consider two properties that are rather protocol-specific than implementation-specific, i.e., protocol fidelity and deadlock freedom trivially ensure that every implementation satisfies these properties if the protocol does.
First, a global type is \emph{terminating}~\cite[Fig.5(2)]{DBLP:journals/pacmpl/ScalasY19} if every CSM run is finite.
It is trivial that this is only true if there are no (used) recursion variable binders $\mu t$.
Second, a global type is \emph{never-terminating}~\cite[Fig.5(3)]{DBLP:journals/pacmpl/ScalasY19} if every CSM run is infinite.
Consequently, this is only the case if the global type has no term~$0$.

\end{document}